\documentclass[a4paper,12pt,reqno]{amsart}
\numberwithin{equation}{section}

\newtheorem{theorem}{Theorem}[section]
\newtheorem{lemma}[theorem]{Lemma}
\newtheorem{definition}[theorem]{Definition}
\newtheorem{proposition}[theorem]{Proposition}
\newtheorem{corollary}[theorem]{Corollary}

\theoremstyle{definition}
\newtheorem*{remark*}{Remark}
\newtheorem*{example*}{Example}
\newtheorem{remark}[theorem]{Remark}
\newtheorem{remarks}[theorem]{Remarks}
\newtheorem*{remarks*}{Remarks}
\newtheorem{example}[theorem]{Example}
\DeclareMathOperator{\Ran}{Ran}
\DeclareMathOperator{\re}{Re}
\DeclareMathOperator{\im}{Im}
\DeclareMathOperator{\diag}{diag}
\DeclareMathOperator{\supp}{supp}

\newcounter{Heq}

\makeatletter
\newenvironment{Hequations}{%
  \refstepcounter{Heq}%
  \protected@edef\theparentequation{\theequation}%
  \setcounter{parentequation}{\value{equation}}%
  \setcounter{equation}{\value{Heq}}%
  \def\theequation{H\theHeq}%
  \ignorespaces
}{%
  \setcounter{equation}{\value{parentequation}}%
  \ignorespacesafterend
}
\makeatother

\begin{document}

\title[Absence of quantum states]%
{Absence of quantum states corresponding to unstable classical channels} 

\author{Ira Herbst}
\address[I.  Herbst]{Department of Mathematics \\
  University of Virginia \\
  Charlottesville \\
  VA 22903\\ U.S.A.}
\email{iwh@weyl.math.virginia.edu}

\author{Erik Skibsted}
\address[E. Skibsted]{Institut for Matematiske Fag\\ 
Aarhus Universitet \\
Ny Munkegade 8000 Aarhus C\\
Denmark}
\email{skibsted@imf.au.dk}
\thanks{E. Skibsted is (partially) supported by  MaPhySto -- A
Network in Mathematical Physics and Stochastics, funded by
The Danish National Research Foundation.}

\maketitle

\tableofcontents

%\newpage

\section{Introduction and results}
\label{sec:1}

The purpose of this paper is to show in a class of models that there are no
quantum states corresponding to unstable classical channels. A
principal example treated in detail is the following: Consider a real-valued potential
  $ V$ on $ \mathbf{R}^{n}$, $ n\geq2$,
  which is smooth outside zero and homogeneous of degree zero. Suppose
  that the restriction of $ V$ to the
  unit sphere $ S^{n-1}$ is a Morse function. We prove that there
  are no $ L^{2}$--solutions to the Schr\"odinger equation 
  $i\partial_t \phi=(-2^{-1}\Delta +V)\phi$
  which asymptotically
  in time are concentrated near local maxima or saddle points of $
  V|S^{n-1}$. Consequently all states concentrate asymptotically
  in time in arbitrarily 
small open cones containing the local minima, cf. \cite {h} and \cite {hs3}.

In the bulk of the paper we consider the
following general situation: Suppose $h(x,\xi)$ is a real classical
Hamiltonian in $ C^{\infty }((\mathbf{R}^{n}\setminus\{0\})
\times\mathbf{R}^{n}) $, $ n\geq2$, satisfying 
\begin{equation}
 x\cdot \nabla _{x}h(x,\xi)
=0  \label{eq:234}
\end{equation}
 in a neighborhood of a point $ (\omega_{0},\xi_{0}) \in S^{n-1}\times
\mathbf{R}^{n}$.  Suppose in addition that this neighborhood is conic in the
$x$--variable and that the orbit $ \left(0,\infty \right) \ni t\rightarrow (
x(t),\xi(t))=\left(tk_{0} \omega_{0},\xi_{0}\right) $ with $k_{0}>0$ is a
solution to Hamilton{'}s equations
\begin{equation*}
{{dx}\over{dt}}=\nabla _{\xi}
h{\left(x,\xi\right)},\quad 
{{d\xi}\over{dt}}=-\nabla _{x}h{\left(
x,\xi\right)}, 
\end{equation*}
or equivalently,
\begin{equation}
\nabla _{x}h{\left(\omega_{0},
\xi_{0}\right)}=0,\quad \nabla _{
\xi}h{\left(\omega_{0},\xi_{0}\right)}
=k_{0}\omega_{0}.\label{eq:1}
\end{equation}
We consider situations in which for each energy $E$ near
$E_{0}=h{\left(\omega_{0},\xi_{0}\right)}$ there is a (typically
unique) ${\left(\omega{\left(E\right)},\xi {\left(E\right)}\right)}\in
S^{n-1}\times\mathbf{R}^{n}$ near $ {\left(\omega_{0},\xi_{0}\right)}
$ depending smoothly on $E$ such that the above structure persists,
namely
\begin{align}
&  h{\left(\omega{\left(E\right)},\xi {\left(E\right)}\right)}=E,
  \label{eq:2}
  \\
&  \nabla _{x}h{\left(\omega{\left(
          E\right)},\xi{\left(E\right)} \right)}=0,
  \label{eq:3}
  \\
&  \nabla _{\xi}h{\left(\omega
      {\left(E\right)},\xi{\left(E\right)}
    \right)}=k{\left(E\right)}\omega {\left(E\right)}.
  \label{eq:4}
\end{align}
Although we shall not elaborate here, we remark that one may easily
derive a criterion for (\ref{eq:2})--(\ref{eq:4}) using
the implicit function theorem.

Let us restrict attention to the constant energy surface
$h{\left(x,\xi\right)}=E$
and to values of 
$
{\left(\hat{x},\xi,E\right)}
$
 close to 
$
{\left(\omega{\left(E\right)},\xi
{\left(E\right)},E_{0}\right)}
$.
 (Here and henceforth 
$
\hat{x}=|x|^{-1}x$.) 
Introduce a change of variables
\begin{equation}
  \begin{split}
&x=x_{n}{\left(\omega{\left(
            E\right)}+u\right)},\quad
    \xi=\xi{\left(E\right)}+\eta 
    +\mu \omega{\left(E\right)};
\\
&u\cdot \omega
    {\left(E\right)}=\eta \cdot \omega{\left( E\right)}=0.
\end{split}
\label{eq:5}
\end{equation}
This amounts to considering coordinates 
$
{\left(u,x_{n},\eta ,\mu \right)}
\in \mathbf{R}^{n-1}\times\mathbf{R}\times \mathbf{R}^{%
n-1}\times \mathbf{R}$. We can solve the equation 
$
h{\left(\omega{\left(E\right)}+u
,\xi{\left(E\right)}+\eta +\mu 
\omega{\left(E\right)}\right)}
=E$
 for $ \mu $ using the implicit function
theorem, because 
$$
\partial _{\mu }h{\left(\omega{\left(
E\right)},\xi{\left(E\right)}
+ \mu \omega{\left(E\right)}\right)}_{|\mu =0}=k{\left(E\right)}>
0$$
 for $E$ near 
$E_{0}$. We obtain 
$
\mu =-g{\left(u,\eta ,E\right)}$
 where $g$ is smooth in a neighborhood of
$
{\left(0,0,E_{0}\right)}$
 and 
$
g{\left(0,0,E_{0}\right)}=0$. 
After introducing the {``}new time{''}
$
\tau =\ln x_{n}{\left(t\right)}
=\ln {\left(x{\left(t\right)}
\cdot \omega{\left(E\right)}\right)}
$
Hamilton{'}s equations reduce to
\begin{equation}
u+{{du}\over{d\tau }}=\nabla _{\eta }g{\left(u,\eta ,E\right)}
,\quad {{d\eta }\over{d\tau }}
=-\nabla _{u}g{\left(u,\eta ,E\right)}.
\label{eq:6}
\end{equation}
(See [A2, p. 243].) After linearization of these equations around the
fixed point 
$
{\left(u,\eta \right)}={\left(0
,0\right)}$
 we obtain with 
$
w={\left(u,\eta \right)}$
\begin{equation}
  \begin{split}
    &{{dw}\over{d\tau }}=B{\left(E\right)} w;\quad B{\left(E\right)}={
\begin{pmatrix}
  0&I \\ -I&0
\end{pmatrix}
}A{\left(E\right)} -{
    \begin{pmatrix}
      I&0\\ 0&0
\end{pmatrix}
},
\\
 &A{\left(E\right)} ={
    \begin{pmatrix}
      g_{u,u}&g_{u,\eta }\\ g_{\eta ,u}&g_{\eta ,\eta }
\end{pmatrix}
}.
\end{split}
\label{eq:7}
\end{equation}
Here the real symmetric matrix 
$
A{\left(E\right)}$
of second order derivatives is evaluated at
$
{\left(0,0,E\right)}$. We assume all eigenvalues of 
$
B{\left(E\right)}$
have nonzero real part (the hyperbolic case). These
eigenvalues are easily proved to come in quadruples, 
$
\lambda ,-1-\lambda$,
 and their complex conjugates (if
$ \lambda $ is not real). If all eigenvalues of
$
B{\left(E\right)}$
 have negative real part then this corresponds to a
stable channel. We prefer the word channel because in the case
considered 
$
x_{n}{\left(t\right)}$
 grows linearly in time. If at least one of the
eigenvalues of 
$B{\left(E\right)}$
 has a positive real part then the usual
stable/unstable manifold theorem shows that there are always classical
orbits (on the stable manifold) for which 
$
{\left(\hat{x}{\left(t\right)}
,\xi{\left(t\right)}\right)}
\rightarrow {\left(\omega{\left(E\right)}
,\xi{\left(E\right)}\right)}
$
 for 
$
t\rightarrow \infty $
 (throughout this paper we use the convention
$
t\rightarrow \infty $
 to mean 
$
t\rightarrow +\infty $
). 
In this situation the question is, do there exist
quantum states whose propagation is governed by a self-adjoint
quantization $H$ of 
$
h{\left(x,\xi\right)}$
 on 
$
L^{2}\left(\mathbf{R}^{n}\right)
$
 (possibly with the singularity at $x=0$
removed) which exhibit this behavior? With a mild further requirement (see (\ref{eq:9}) below), we will answer this question in
the negative.\par 

\par To be precise, let us first fix a (small) neighborhood
$
\mathcal{U}_{0}\subseteq (\mathbf{R}^{n}\setminus{\left\{
0\right\}})\times \mathbf{R}^{n}$
 of 
$
{\left(k{\left(E_{0}\right)}
\omega_{0},\xi_{0}\right)}
$. Then we consider a small open neighborhood
$I_{0}$
 of 
$E_{0}$
 and states of the form 
$\psi =f{\left(H\right)}\psi $
 with 
$f\in C^{\infty }_{0}{\left(I_{0}\right)}$
such that:
\begin{equation}
\begin{split}
&\text{For all }g_{1},g_{2}\in C^{\infty }_{0}{\left(\mathbf{R}^{n}\right)} 
\\
&\|{\left\{g_{1}{\left(
             t^{-1}x\right)}-g_{1}{\left(
             k{\left(H\right)}\omega{\left(H\right)} \right)}1_{I_{0}}{\left(H
           \right)}\right\}}\psi {\left( t\right)}\|\rightarrow 0\text{ for
     }t\rightarrow \infty ,
\\
&\|{\left\{g_{2}
         {\left(p\right)}-g_{2}{\left( \xi{\left(H\right)}\right)}
         1_{I_{0}}{\left(H\right)} \right\}}\psi {\left(t\right)}
     \|\rightarrow 0\text{ for }t\rightarrow \infty ;
\\
&\psi {\left(
         t\right)}=e^{-itH}\psi ,\quad p=-i\nabla _{x},
\end{split}
\label{eq:8}
\end{equation}
while 
\begin{equation}
  \begin{split}
&\int _{1}^{\infty }t^{
    -1}\left\|a^{w}{\left(t^{-1} x,p\right)}\psi {\left(t\right)}
  \right\|^{2}dt<\infty \text{ for all }a\in C^{\infty
  }_{0}{\left(\mathcal{U}_{0 }\setminus\gamma {\left(I_{0}\right)}
    \right)};
\\
&\gamma {\left(I_{0}\right)}
  ={\left\{{\left(k{\left(E\right)} \omega{\left(E\right)},\xi{\left(
E\right)}\right)}\mid E\in I_{%
0}\right\}}.
\end{split}
\label{eq:9}
\end{equation}
(Here 
$
a^{w}$
 signifies Weyl quantization, and 
$
1_{I_{0}}$
 is the characteristic function of 
$
I_{0}$.)\par 

\par Notice that by \eqref{eq:8}, at least intuitively, for all
such symbols $a$
\begin{equation}
\left\|a^{w}{\left(t^{-1
}x,p\right)}\psi {\left(t\right)}
\right\|\rightarrow 0\text{ for }t\rightarrow \infty ,
\label{eq:10}
\end{equation} 
so that \eqref{eq:9} appears as a weak additonal assumption (or as part of our definition of a quantum
channel). See the beginning of Section 3 where (\ref{eq:10}) is
proved from \eqref{eq:8}  and assumptions about the 
pseudodifferential nature of $H$ (conditions (H1)--(H3)). On the other hand, \eqref{eq:10} is also a consequence of
\eqref{eq:9} as may be shown by a subsequence argument (cf. the proof
of \eqref{OOl}).

The
states $ \psi $ obeying the above conditions (with fixed
$I_{0}$) form a subspace whose closure, say 
$\mathcal{H}_{0}$, is $H${--}reducing.\par 

\par We show the following (main) result.

\begin{theorem}\label{thm:1.1}
  Suppose $ B{\left(E_{0}\right)}$ has an eigenvalue with a positive real
  part. Then under a certain assumption concerning possible
resonances (and other technical
  conditions, see \emph{(H1)--(H8)} in Section 2) there exists a sufficiently
  small open neighborhood $ I_{0}$ of $ E_{0}$ such that
\begin{equation}
\mathcal{H}_{0}={\left\{0\right\}}.
\label{eq:11}
\end{equation}
\end{theorem} 

There is the following slightly more general result not involving
\eqref{eq:8}.

\begin{theorem}\label{thm:1.11}
  Under the conditions of Theorem \ref{thm:1.1} there exists  a sufficiently
  small open neighborhood $ I_{0}$ of $ E_{0}$ such that if a state $\psi(t) ={e^{-itH}f\left(H\right)}\psi $
 with 
$f\in C^{\infty }_{0}{\left(I_{0}\right)}$ obeys \eqref{eq:9}, then in fact  the pointwise decay \eqref{eq:10} holds for all $a\in C^{\infty
  }_{0}{\left(\mathcal{U}_{0 }
    \right)}$.
\end{theorem} 

\par A symbol satisfying the conditions \eqref{eq:3} and \eqref{eq:4}
was studied by Guillemin and Schaeffer [GS]. In their paper 
the roles of $x$ and $ \xi $ are reversed and
their $h$ is homogeneous of degree one in $ \xi $.
There is only one half-line of points in question rather than a one
parameter family of half-lines (their critical set of points is at zero
energy). Under the condition of no resonances  they obtain a
conjugation of $H$ to a simpler normal form from which they
draw conclusions about propagation of singularities for an equation of
the form 
$H\psi =\phi $.

\par To see what Theorem \ref{thm:1.1} means in the model where 
$h{\left(x,\xi\right)}=2^{-1}
\xi^{2}+V{\left(\hat{x}\right)}$
with $V$ a Morse function on
$S^{n-1}$  we recall from [H]: The spectrum of
$H=2^{-1}p^{2}+V{\left(\hat{x}\right)}$
 is purely absolutely continuous and 
\begin{equation}
I=\sum _{\omega_{l}\in \mathcal{C}_{%
r}}P_{l},
\label{eq:12}
\end{equation}
where 
$P_{l}$
 are $H${--}reducing orthogonal
projections defined as follows: Pick any family 
${\left\{\chi _{l}|\omega_{l}\in C_{r}\right\}}$
of smooth functions on 
$S^{n-1}$
with 
$\chi _{k}{\left(\omega_{l}\right)}
=\delta _{kl}$
 (the Kronecker symbol); here 
$C_{r}$
is the finite set of non-degenerate critical points in 
$S^{n-1}$
 for 
$V$. Then 
\begin{equation*}
P_{l}=s-\lim _{t
\rightarrow \infty }e^{itH}\chi _{l}{\left(
\hat{x}\right)}e^{-itH},
\end{equation*}
see [H] and [ACH]. Furthermore in [H] the existence of an 
asymptotic momentum 
$p^{+}$
was proved and its relationship to the above
projections was
shown.  (There was the restriction in [H] to
$n\geq 3$
 but this is easily removed using the Mourre estimate
[ACH, Theorem C.1].)\par 

\par We notice that \eqref{eq:12} has an analog in Classical
Mechanics: Any classical orbit (except for the exceptional ones that
collapse at the origin) obeys 
$|x|\rightarrow \infty $
with 
$\hat{x}\rightarrow \omega_{l}$
for some 
$
\omega_{l}\in C_{r}$.\par 

\par Obviously the collection \eqref{eq:2}--\eqref{eq:4}
corresponds in the potential model exactly to 
$C_{r}$: 
${\left(\omega{\left(E\right)},\xi
{\left(E\right)}\right)}={\left(
\omega_{l},\sqrt{2{\left(E-V{\left(
\omega_{l}\right)}\right)}
} \omega_{l}\right)}$
with 
$\omega_{l}\in C_{r}$. The assumption that the real part of one eigenvalue 
is positive corresponds to 
$\omega_{l}$
being either a local maximum or a saddle point of
$V$. Moreover we have the identification 
\begin{equation}
\mathcal{H}_{0}=\Ran {\left(P_{l}
1_{I_{0}}{\left(H\right)}
\right)}.
\label{eq:13}
\end{equation}
Whence, upon varying 
$I_{0}$, Theorem~\ref{thm:1.1} yields the following for the potential
model. 

\begin{theorem}\label{thm:1.2}
 Suppose $ \omega_{l}\in C_{r}$ is 
the location of a local maximum or
 a saddle point of $ V$. Then
 \begin{equation}
   P_{l}=0.
   \label{eq:14}
 \end{equation}
\end{theorem}

Of course we will need to verify (\ref{eq:13}) in order to use Theorem \ref{thm:1.1}
and this involves verifying (\ref{eq:8}) and (\ref{eq:9}) for $\psi \in \mbox{Ran} P_\ell$ satisfying $\psi = f(H)\psi$, $f \in C^\infty_0(I_0)$ (see Section 8).

\par A detailed analysis of the large time asymptotic behavior of states
in the range of the projections 
$P_{l}$
which correspond to local minima was accomplished
recently in [HS3]. In particular for any local minimum,
$P_{l}\not=0$. Moreover in this case we have \eqref{eq:13} for
the analogous space of that in Theorem~\ref{thm:1.1}. One may easily include in
Theorem~\ref{thm:1.2} a short-range perturbation 
$V_{1}=O{\left(|x|^{-1-\delta }\right)}$, $\delta >0$, $\partial^{\alpha
}_{x}V_{1}=O{\left(|x|^{-2}\right)}$, $|\alpha |=2$,
to the Hamiltonian $H$, see Remarks~\ref{rem:8.3} (1).

\par  The results Theorems~\ref{thm:1.1}  and ~\ref{thm:1.11} are  much more general than
Theorem~\ref{thm:1.2}. In particular, as a further application, we can
apply 
them  to a problem of a quantum particle in two dimensions subject to an
electromagnetic vector potential which is asymptotically homogeneous
of degree zero in $x$, see 
[CHS2]. (For another magnetic field problem in this class, see [CHS1].)
Let us also give a  simple example from Riemannian
  geometry. For further examples, see Examples \ref{exam:riema2},
  \ref{exam:riema3},  \ref{exam:riema4}
  and \ref{exam:riema33} which indeed may be viewed as proper examples due to
   symplectic
  invariance of the problem.

\begin{example} \label{exam:riema1}Consider the symbol $h$ on $(\mathbf{R}^{2}\setminus\{0\})\times \mathbf{R}^{2}$
  \begin{equation}\label{eq:ham}
 h=h(x,\xi)=\tfrac {1}{2}\big(1+ax_2^2|x|^{-2}\big)^{-1}\xi_1^2+\dfrac {1}{2}\xi_2^2;\;a> 0.   
  \end{equation} The family  $(1,0;\sqrt{2E},0)$,    
  $E>0$,  consists  of points obeying \eqref{eq:2},
  \eqref{eq:3} and \eqref{eq:4}. 
 For the   linearized reduced 
  flow  \eqref{eq:7}
  we  find the   eigenvalues $-\tfrac {1}{2}(-1\pm \sqrt{1+4a})$,  
  and we conclude that the
fixed points  are saddle points. If $a$ is irrational there are no resonances
  of any order (see Section \ref{sec:2} for definition), whence we may infer
  from Theorem~\ref{thm:1.1} that  there is no
  quantum  channel associated  to the  family of fixed points in this
  case.
Using the absence of low order  resonances condition \eqref
  {eq:res} we may in fact obtain this conclusion for $a\neq
  \tfrac{3}{4}, 2$; see Remark \ref{rem:remres} for a further discussion. We have tacitly assumed that the symbol
  \eqref{eq:ham} is suitably regularized at  $x=0$ (for the
  quantization). 
\end{example}

%\begin{example*} Consider the symbol $h$ on $\mathbf{R}^{2}\times \mathbf{R}^{2}$
%  \begin{equation*}
% h=h(x,\xi)=2^{-1}(1+ax_2^2x_1^{-2})^{-1}\xi^2;\;a\geq 0.   
%  \end{equation*} The family  $(1,0;\sqrt{2E},0)$,    
%  $E>0$,  consists  of points obeying \eqref{eq:2},
%  \eqref{eq:3} and \eqref{eq:4}. Although $h$ is singular at $x_1=0$
%  (and not only at the origin $x=0$) we notice that $h$ is regular near these
%   points. 
% For the   linearized reduced 
%  flow  \eqref{eq:7}
%  we  find the   eigenvalues $\tfrac {1}{2}(-1-\sqrt{1+4a})$
%  and  $\tfrac {1}{2}(-1+\sqrt{1+4a})$,  
%  and we conclude that the
%fixed points  are saddle points for $a> 0$. If $a$ is irrational there are no resonances
%  (of any order), whence we may infer
%  from Theorem~\ref{thm:1.1} that  there is no
%  quantum  channel associated  to the  family of fixed points in this case. 
%\end{example*}

\par Our proof of Theorem~\ref{thm:1.1} consists of three steps:

\par I) Assuming 
$\psi {\left(t\right)}=e^{-itH}\psi $ does localize in phase space as
$t\rightarrow \infty $ in the region $|u|+|\eta |\leq \epsilon $ for any $
\epsilon >0$ in the sense of (1.8) and (1.9), we prove a stronger localization. Namely, for some small
positive $ \delta $, the probability (assuming here that $ \psi $ is
normalized) that $ \psi {\left(t\right)}$ is localized in the region $
|u|+|\eta |\geq t^{-\delta }$ goes to zero as $ t\rightarrow \infty $.
See Section 4.
\par

\par II) Using I) and an iteration scheme, we construct an observable
$ \Gamma $ which decreases {``}rapidly{''} to zero.  This iteration scheme is
based on one used by Poincar\'e (see [A1, pp. 177{--}180]) to obtain a change
of coordinates which linearizes \eqref{eq:6}. The fact that if one eigenvalue
of $ B{\left(E\right)}$ has a positive real part then another has real part
$<-1$ is relevant here. Our observable $ \Gamma $ is in first approximation
roughly a quantization of a component of $w$ in \eqref{eq:7} which decays
as $ \exp {\left(\lambda \tau \right)} $ with $\re\lambda <-1$. See Section 5.

\par III) Using Mourre theory we prove an uncertainty principle lemma
for two self-adjoint operators $P$ and $Q$
satisfying 
$i{\left[P,Q\right]}\geq cI,\quad c>0$,
and some technical conditions. A consequence of this  lemma is that  if 
$0\leq \delta _{1}<\delta _{2}$
and 
$g_{1}\text{ and }g_{2}$
are two bounded compactly supported functions then
\begin{equation*}
\lim  _{t\rightarrow \infty }
\left\|g_{1}{\left(t^{-\delta _{
1}}Q\right)}g_{2}{\left(
t^{\delta _{2}}P\right)}\right\|
=0.
\end{equation*}
If $ \psi $ is normalized this bound implies that the
localizations of I) and II) are incompatible. See Sections 6 and 7.\par 

\par The basic theme of our paper may be phrased as absence of certain
quantum mechanical states which are present in the corresponding classical
model. Notice that given any critical point $\omega_{l}\in C_{r}$ (restricting
for convenience the discussion to the potential model) there are indeed
classical orbits with $|x|\rightarrow \infty $ and $ \hat{x}\rightarrow
\omega_{l}$; in particular this is the case for any given local maximum or
saddle point. Intuitively, Theorem~\ref{thm:1.1} is true because the
associated classical orbits occur for only a {``}rare{''} set of initial
conditions as fixed by the stable manifold theorem.  Alternatively, for some
components of ${\left(\hat{x},\xi\right)} $ the convergence to
${\left(\omega_{l},\xi^{+}\right)} $ is {``}too fast{''} thus being
incompatible with the uncertainty principle in Quantum Mechanics. These two
different explanations are actually connected.\par

\par For another example of this theme we refer to [G2], [S1] and
[S2].\par 

\par We addressed the problem of Theorem~\ref{thm:1.2} in a previous work, [HS1],
where we proved \eqref{eq:14} at local maxima but only had a partial result
for saddle points (using a different time-dependent method). Also in the case of 
homogeneous potentials  similar and related results were 
obtained in [HMV1] and [HMV2] by  stationary methods. The present
paper is an expanded version of the preprint \cite {hs2}.\par

\par This paper is organized as follows: In Section~\ref{sec:2} we elaborate on
all technical conditions needed for Theorem~\ref{thm:1.1} and give a more
detailed outline of its proof, cf. the steps I){--}III) indicated above. In
Section~\ref{sec:3} we have collected a few technical preliminaries. In
Section~\ref{sec:4} we prove the $ t^{-\delta }$ {--}localization, cf. step
I), while the localization of $ \Gamma $ is given in Section~\ref{sec:5}.
Finally, Section~\ref{sec:6} is devoted to the Mourre theory for this
observable. We complete the proof of Theorem~\ref{thm:1.1} in
Section~\ref{sec:7} (the proof of  Theorem~\ref{thm:1.11} is omitted
since 
it follows the same pattern) and give a few missing details of the proof of
Theorem~\ref{thm:1.2} in Section~\ref{sec:8}. In Appendix A we study 
 possible generalizations of the homogeneity
condition \eqref{eq:234}.
%; the  discussion is primarily concentrated on
%the classical theory.\par

\section{Technical conditions and outline of proof}
\label{sec:2}

\par We fix 
$ {\left(\omega_{0},\xi_{0}\right)} \in S^{n-1}\times \mathbf{R}^{n}$ and a
small open neighborhood $ I_{0}$ of $
E_{0}=h{\left(r\omega_{0},\xi_{0}\right)}$ as in Section 1. We shall elaborate
on conditions for the real-valued symbol $ h{\left(x,\xi\right)}$, see
(H1){--}(H8) below. For convenience we remove a possible singularity at $x=0$
caused by the imposed (local) homogeneity assumption of Section 1. This may be
done as follows. Let $ \mathcal{N}_{0}$ be as small open neighborhood of $
{\left(\omega_{0},\xi_{0}\right)} $. We shall now and henceforth assume that
for some $ r_{0}>0$
\begin{Hequations}
  \begin{equation}
  \begin{split}
    &\displaystyle h{\left(x,\xi\right)}= h{\left(r_{0}\hat{x},\xi
      \right)}\; in \;\mathcal{C}_{0}:={\left\{ {\left(x,\xi\right)}\mid
        {\left( \hat{x},\xi\right)}\in \mathcal{N}_{0},|x|>r_{0}\right\}},
    \\
    &\displaystyle{{}h\in C^{\infty }{\left(\mathbf{R}^{n}\times
          \mathbf{R}^{n}\right)}}.
\end{split}\label{H1}
\end{equation}
\end{Hequations}
Notice that this modification intuitively is irrelevant for the issue of
Theorem~\ref{thm:1.1} (which concerns states propagating linearly in time in
configuration space).\par

\par We assume that for some 
$r,l\geq 0$
\begin{Hequations}
\begin{equation}
   h\in S{\left({\left\langle \xi\right\rangle }^{
        r}{\left\langle x\right\rangle }^{l},g_{ 0}\right)};\enspace
  g_{0}={\left\langle x\right\rangle }^{-2}dx^{2}+d\xi^{ 2},\enspace
  {\left\langle x\right\rangle }= {\left(1+|x|^{2}\right)}^{ 1/2},
\label{H2}
\end{equation}
\end{Hequations}
and that 
\begin{Hequations}
  \begin{equation}
H=h^{w}{\left(x,p\right)}\;is \;essentially \;self\!\!-\!\!adjoint \;on\;
C^{\infty }_{0}{\left(\mathbf{R}^{n}
\right)}.\label{H3}
\end{equation}
\end{Hequations}
(See Section~\ref{sec:3} for notation.)\par 

\begin{remark*}
  There is some freedom in choosing a global condition like (\ref{H2}). For
  example it suffices to have (\ref{H2}) with $ g_{0}$ replaced by $
  {\left\langle x\right\rangle }^{-2\delta _{1} }dx^{2}+{\left\langle
      x\right\rangle }^{ 2\delta _{2}}d\xi^{2}$ with $ 0\leq \delta
  _{2}<\delta _{1}\leq 1$ .
\end{remark*}

\par We assume 
\begin{Hequations}
  \begin{equation}
\text{\eqref{eq:2}{--}\eqref{eq:4}} \;for \;E\in I_{0}.\label{H4}
\end{equation}
\end{Hequations}

\par We define 
$ \omega_{n}{\left(E\right)}=\omega {\left(E\right)}$, and  shrinking
$I_0$ if necessary we pick smooth
functions $$ \omega_{1}{\left(E\right)},\dots ,\omega_{n-1}{\left(E\right)}\in
S^{n-1}$$ such that $ \omega_{1}{\left(E\right)},\dots
,\omega_{n}{\left(E\right)}$ are mutually orthogonal. We define, cf.
\eqref{eq:5}, $x_{j}=x\cdot \omega_{j}{\left(E\right)} $ for $ j\leq n$, $
u_{j}=x_{j}/x_{n}$ and $ \eta _{j}={\left(\xi-\xi{\left( E\right)}\right)}\cdot
\omega_{ j}{\left(E\right)}$ for $ j\leq n-1$ and $ \mu
={\left(\xi-\xi{\left(E\right)} \right)}\cdot \omega_{n}{\left(E \right)}$.
Let $ w={\left(u,\eta \right)}={\left( u_{1},\dots ,u_{n-1},\eta _{1},\dots
    ,\eta _{n-1}\right)}$.

\par As for the matrix $B(E)$ of \eqref{eq:7} in these
coordinates we need the condition:
\begin{Hequations}
\begin{equation}
\parbox{210pt}{\em
The real part of each eigenvalue of $B(E)$ is nonzero for $E\in I_0$.}
\label{H5}
\end{equation}
\end{Hequations}

\par Let us order the eigenvalues as 
$ \beta ^{s}_{1}{\left(E\right)},\dots ,\beta ^{s}_{n^{s}}{\left(
    E\right)},\beta ^{u}_{1}{\left( E\right)},\dots ,\beta ^{u}_{n^{
    u}}{\left(E\right)}$ where $ \re {\left(\beta ^{s}_{j}{\left(
        E\right)}\right)}<0$ ($ \beta ^{s}_{j}{\left(E\right)} $ are the
stable ones) and $ \re {\left(\beta ^{u}_{j}{\left( E\right)}\right)}>0$ ( $
\beta ^{u}_{j}{\left(E\right)} $ are the unstable ones). Let $ \beta
{\left(E\right)}$ refer to the $ {\bf C}^{2n-2}$ {--}vector of eigenvalues $
{\left(\beta ^{s}_{1}{\left(E \right)},\dots ,\beta
    ^{u}_{n^{u}}{\left(E\right)}\right)} $ counted with multiplicity.\par

\par We are interested in the case
\begin{Hequations}
  \begin{equation}
n^{u}=n^{u}{\left(E\right)}
\geq 1.\label{H6}
\end{equation}
\end{Hequations}

\par Let 
$ V^{s}{\left(E\right)}$ and $ V^{u}{\left(E\right)}$ be the sum of the
generalized eigenspaces of $B(E)$ correponding to stable and unstable
eigenvalues, respectively. Then we have the decomposition 
\begin{equation*}
{\mathbf{C}^{2n-2}=V^{s}{\left(E\right)}
\oplus V^{u}{\left(E\right)}.}
\end{equation*}

\par Using basis vectors respecting this structure we can find a smooth
$ M_{2n-2}{\left(\mathbf{C}\right)}$ {--}valued function $T(E)$ such that
\begin{equation}
T{\left(E\right)}^{-1}B{\left( E\right)}T{\left(E\right)}=\diag
{\left(B^{s}{\left(E\right)},B^{u}{\left(E\right)}\right)}.
\label{eq:17}
\end{equation}

\par 
We may assume the following at $ E=E_{0}$ : Corresponding to the decomposition
into generalized eigenspaces 
\begin{align*}
&\mathbf{C}^{2n-2}=V^{s}\oplus V^{u}
 =V^{s}_{1}\oplus \cdots \oplus 
V^{s}_{n^{s}}\oplus V^{u}_{1}\oplus \cdots \oplus V^{u}_{n^{u}},
\\
&{T{\left(E_{0}\right)}^{-1}
B{\left(E_{0}\right)}T{\left(
E_{0}\right)}=\diag{\left(
B^{s}_{1},\dots ,B^{u}_{n^{u}
}\right)},}
\end{align*}
where for all entries $ N^{\#}_{j}:=B^{\#}_{j}- \beta ^{\#}_{j}{\left(E_{0}
  \right)}I_{\dim{\left(V^{\#}_{j}\right)}}$ is strictly lower triangular.
Given any $ \epsilon >0$ we may assume (by rescaling the basis vectors) that
\begin{equation}
\|N^{\#}_{j}\|\leq \epsilon .
\label{eq:18}
\end{equation}

\par We introduce a vector of new variables 
$ \gamma ={\left(\gamma ^{s},\gamma ^{u} \right)}={\left(\gamma _{1},\dots
    ,\gamma _{2n-2}\right)}$
\begin{equation}
\gamma =\gamma {\left(w{\left(E\right)},E\right)}=T{\left(E\right)}^{
-1}w{\left(E\right)},
\label{eq:19}
\end{equation}
where $ \gamma ^{s}$ and $ \gamma ^{u}$ are the vectors of coordinates of the
part of $ w{\left(E\right)}$ in $ V^{s}{\left(E\right)}$ and $
V^{u}{\left(E\right)}$, respectively.\par

\par We shall make the assumption (using
{``}$tr${''} to denote
transposed):
\begin{Hequations}
  \begin{equation}
\parbox{10cm}{%
\it{There exists a smooth eigenvector}
  $v{\left(
E\right)}$  of $B{\left(E\right)}^{tr}$ in $E\in I_{0}$,
such that $\re{\left(\lambda {\left(E\right)}\right)} <-1$
 for the corresponding
eigenvalue $\lambda {\left( E\right)}$.
}\label{H7}
\end{equation}
\end{Hequations}
See Remark~\ref{rem:2.2} below for an alternative condition.\par 

\par The ordering of the eigenvalues may be chosen such that 
\begin{equation}
\beta ^{s}_{1}{\left(E\right)}
=\lambda {\left(E\right)}.
\label{eq:20}
\end{equation}
It may also be assumed that $ v{\left(E\right)}$ is the first row of $
T{\left(E\right)}^{-1}$. Clearly by \eqref{eq:20} $ \beta
^{s}_{1}{\left(E\right)} $ is smooth for $ E\in I_{0}$.\par

\par We call 
$ E_{0}$ a resonance of order $ m\in {\left\{2,3,\dots \right\}}$ for an
eigenvalue $ \beta ^{\#}_{j}{\left(E_{0} \right)}$ if for some $
\alpha ={\left(\alpha _{1},\dots \alpha _{%
      2n-2}\right)}\in {\left({\bf N\cup } {\left\{0\right\}}\right)}^{2n-2}$
with $ {\left\lvert \alpha \right\rvert}=m$,
\begin{equation}
\beta^{\#}_{j}{\left(E_{0} \right)}=\beta {\left(E_{0}\right)} \cdot \alpha .
\label{eq:21}
\end{equation}

\par We assume that 
\begin{Hequations}
  \begin{equation}
  E_{0}\; is \;not \;a \;resonance \;of \;order \; \leq m_{0} 
  \;for \;\beta ^{s}_{1} {\left(E_{0}\right)}.\label{H8}
\end{equation}
\end{Hequations}
Here $ m_{0}$ may be extracted from the bulk of the paper; the condition
\begin{equation} \label{eq:res}
{m_{0}>\max {\left(4,{{1+\re 
{\left(\beta ^{s}_{1}{\left(E_{%
0}\right)}\right)}}\over{-\re 
{\left(\beta ^{s}_{1}{\left(E_{%
0}\right)}\right)}}},\dots ,{{1+\re {\left(\beta ^{s}_{%
n^{s}}{\left(E_{0}\right)}
\right)}}\over{-\re {\left(\beta ^{%
s}_{n^{s}}{\left(E_{0}\right)}
\right)}}}\right)}} 
\end{equation}
 suffices.\par 

\begin{remark}\label{rem:remres}
Typically the set of resonances of all orders will be dense in $I_0$.
The theorem proved with \eqref{H8} does not exclude cases where there
are low order resonances as long as they constitute a discrete
set. This is used in the proof of Theorem \ref{thm:1.2} in Section
\ref{sec:8}.
For the exceptional values $a=
  \tfrac{3}{4}$ and $a=
  2$ of Example \ref{exam:riema1} there are resonances
  of order $5$ and $4$, respectively. For these values of $a$ {\it
    all} positive energies
  are resonances, and consequently our theorem is not applicable.
\end{remark}

\par We shall build a (classical) observable $ \Gamma $ from
the first coordinate $ \gamma _{1}=\gamma _{1}{\left(w{\left(
        E\right)},E\right)}=v{\left( E\right)}\cdot w{\left(E\right)} $ of $
\gamma ^{s}=\gamma ^{s}{\left(w{\left( E\right)},E\right)}$
\begin{equation}
\Gamma =\gamma _{1}{\left(w{\left(E\right)},E\right)}+O{\left(|\gamma {\left(
w{\left(E\right)},E\right)}|^{2}\right)}.
\label{eq:22}
\end{equation}

\par In the study of an analogous quantum observable we consider in
detail the case where for some 
$1\leq l\leq n-1$
\begin{equation}
\label{b13}
\partial _{\eta _{l}}\gamma _{1}{\left(
w,E_{0}\right)}_{|w=0}\not=0.
\end{equation}
We notice that if (\ref{b13}) is not true then for some
$1\leq l\leq n-1$
\begin{equation}
\label{bb13}
\partial _{u_{l}}\gamma _{1}{\left(
w,E_{0}\right)}_{|w=0}\not=0.
\end{equation}

\par The construction of the quantum $ \Gamma $ in the case
of (\ref{b13}) and an elaboration of its decay properties will be given in
Section~\ref{sec:5}. A Mourre esimate is given in Section~\ref{sec:6}, and we
complete the proof of Theorem~\ref{thm:1.1} in this case in
Section~\ref{sec:7}. We refer the reader to Remarks~\ref{rem:5.3},
\ref{rem:6.3} and \ref{rem:7.2} for the
modifications needed for showing Theorem~\ref{thm:1.1} in the case of (\ref{bb13}).

\subsection{Outline of proof of Theorem~\ref{thm:1.1}}
\label{sec:outl-proof-theor}

Consider a classical orbit with $
{\left(\hat{x}{\left(t\right)},\xi{\left(t\right)}\right)} \rightarrow
{\left(\omega{\left(E\right)},\xi{\left(E\right)}\right)} $ for $ t\rightarrow
\infty $ (and $E$ near $E_{0}$). How do we prove the bound $ |u|+|\eta
|\leq Ct^{-\delta }$ for some positive $ \delta $?\par

\par We consider the observables
\begin{equation}
\label{b14}
q^{s}=|\gamma ^{s}|^{2},\enspace 
q^{u}=|\gamma ^{u}|^{2},\enspace
q^{-}=q^{u}-q^{s},\enspace q^{+}=q^{u}+q^{s}=|\gamma |^{
2}.
\end{equation}

\par Using \eqref{eq:6} and \eqref{eq:17} we compute
\begin{equation}
\label{b15}
{{d}\over{dt}}\gamma ={{\partial _{\mu 
}h}\over{x_{n}}}{\left\{{\left(
B^{s}{\left(E\right)}\gamma ^{
s},B^{u}{\left(E\right)}\gamma ^{
u}\right)}+O{\left(q^{+}\right)}
\right\}}.
\end{equation}

\par 
For $ \epsilon >0$ small enough in \eqref{eq:18} the
equation (\ref{b15}) leads to 
\begin{equation}
\label{b17}
{\frac{d}{dt}}q^{-}=2\re {\left\langle 
\gamma ^{u},{\frac{d}{dt}}\gamma ^{
u}\right\rangle }_{\mathbf{C}^{n^{u}}
}-2\re {\left\langle \gamma ^{s},{
\frac{d}{dt}}\gamma ^{s}\right\rangle }_{%
\mathbf{C}^{n^{s}}}\geq \delta ^{-}
t^{-1}q^{+}
\end{equation}
for some positive $ \delta ^{-}$ (which may be chosen independent of $E$ close
enough to $ E_{0}$) and for all $ t\geq t^{-}$ (with $ t^{-}$ large enough).
\par

\par In particular 
$ q^{-}$ is increasing and hence 
\begin{equation}
\label{b18}
 q^{-}\leq 0;\quad t\geq t^{-}.
\end{equation}
\par Using (\ref{b15}), (\ref{b18}) and the Cauchy-Schwarz
inequality we compute
\begin{equation}
\label{b19}
{{d}\over{dt}}q^{s}=2\re {\left\langle 
\gamma ^{s},{{d}\over{dt}}\gamma ^{s}\right\rangle }_{\mathbf{C}^{n^{s}}
}\leq -2\delta ^{s}t^{-1}q^{s}
\end{equation}
for some positive $ \delta ^{s}$ and all $ t\geq t^{s}$.\par

\par Integrating (\ref{b19}) yields
\begin{equation}
\label{b20}
q^{s}\leq C^{s}t^{-2\delta ^{s}},\quad t\geq t^{s}.
\end{equation}

\par Finally from (\ref{b18}) and (\ref{b20}) we conclude
that $ q^{+}\leq 2C^{s}t^{-2\delta ^{s}} $ and therefore that
\begin{equation}
\label{b21}
 |\gamma|\leq Ct^{-\delta };\quad \delta \leq \delta ^{s}.
\end{equation}

This classical proof will be the basis for our quantum arguments in Section 4 which
constitute step I) of the proof of Theorem \ref{thm:1.1}.

\begin{remarks}\label{rem:2.1} \hfill
\begin{enumerate}
\item We may choose the positive $ \delta $ in (\ref{b21}) as close to the
    (optimal) exponent $ \min {\left(\re {\left(-\beta
            ^{s}_{1}{\left(E_{0}\right)} \right)},\dots ,\re {\left(-\beta
            ^{s}_{n^{s}}{\left(E_{0}\right)} \right)}\right)}$ as we wish
    (provided $E$ is taken close enough to $ E_{0}$).

\item Although not needed, one may easily prove using similar differential
    inequalities that indeed $q^{u}=O{\left({\left(q^{s}\right)}^{ 2}\right)}$
    in complete agreement with the stable manifold theorem.
\end{enumerate}
\end{remarks}

\noindent{\bf Classical $\Gamma$}.
%\subsubsection {Classical $\Gamma$}
%\label{subsec:classical-gamma}
To implement step  II) of the proof, we shall for each $ m\in {\left\{1,\dots ,m_{0}\right\}} $ construct a $
\gamma ^{{\left(m\right)}}$ of the form \eqref{eq:22} such that
\begin{equation}
\label{B15}
{{d}\over{dt}}\gamma ^{{\left(m\right)} }={{\partial _{\mu }h}\over{x_{n}
}}\beta ^{s}_{1}{\left\{\gamma ^{
{\left(m\right)}}+O{\left(|
\gamma |^{m+1}\right)}\right\}}
;\quad \beta ^{s}_{1}=\beta ^{
s}_{1}{\left(E\right)}.
\end{equation}
Specifically we shall require
\begin{equation}
\label{BB21}
\gamma ^{{\left(1\right)}}=\gamma _{
1},\text{ and }\gamma ^{{\left(m\right)}
}=\gamma _{1}+\sum _{2\leq {\left\lvert 
\alpha \right\rvert}\leq m}c_{\alpha }
\gamma ^{\alpha };\quad m\geq 2,
\end{equation}
with $ \gamma ^{\alpha }=\gamma ^{\alpha _{1} }_{1}\cdots \gamma ^{\alpha
  _{2n-2} }_{2n-2}$. (It will follow from the construction below that the
coefficients $ c_{\alpha }=c_{\alpha }{\left(E\right)} $ will be smooth; this
will be important for {``}quantizing{''} the symbol.)\par

\par We proceed inductively. Clearly by (\ref{b15}) we have (\ref{B15}) for
  $m=1$. Now suppose we have constructed a 
function 
$\gamma ^{{\left(m-1\right)}}=\sum 
_{{\left\lvert \alpha \right\rvert}\leq m-1}
c_{\alpha }\gamma ^{\alpha }$
 obeying 
 \begin{equation*}
{{{d}\over{dt}}\gamma ^{{\left(
m-1\right)}}={{\partial _{\mu }
h}\over{x_{n}}}\beta ^{s}_{1}
{\bigg(\gamma ^{{\left(m-1\right)}
}+\sum _{{\left\lvert \alpha \right\rvert}
=m}d_{\alpha }\gamma ^{\alpha }+
O{\left({\left\lvert \gamma \right\rvert}^{m+1}\right)}\bigg)},} 
\end{equation*}
then we add to $ \gamma ^{{\left(m-1\right)}}$ a function of the form $ \sum
_{{\left\lvert \alpha \right\rvert}= m}c_{\alpha }\gamma ^{\alpha }$ and we
need to solve 
\begin{equation}
\label{BB22}
 {{d}\over{dt}}\sum _{{\left\lvert
      \alpha \right\rvert}=m}c_{\alpha }
\gamma ^{\alpha }={{\partial _{
\mu }h}\over{x_{n}}}\beta ^{s
}_{1}\sum _{{\left\lvert \alpha \right\rvert}
=m}{\left(c_{\alpha }-d_{\alpha 
}\right)}\gamma ^{\alpha }+O{\left(
{\left\lvert \gamma \right\rvert}^{m+1}\right)}.
\end{equation}

\par For that we compute the derivative using again (\ref{b15}).
Let us denote by $ \beta _{ij}$ the $ij${'}th entry of the matrix $ \hbox{\rm
  diag}{\left(B^{s}{\left(E\right)}^{ tr},B^{u}{\left(E\right)}^{
      tr}\right)}$. Then (\ref{BB22}) reduces to solving
\begin{equation}
\label{BB23}
 \sum_{{\left\lvert \tilde{\alpha} \right\rvert}=m}\sum _{i,j}
\tilde{\alpha}_{i}\beta _{ij}c_{\tilde{\alpha}}\gamma ^{\tilde{\alpha}-e_{i}+e_{j}
}=\beta ^{s}_{1}\sum 
_{{\left\lvert \alpha \right\rvert}=m}{\left(
c_{\alpha }-d_{\alpha }\right)}
\gamma ^{\alpha },
\end{equation}
which in turn reduces to solving the system of algebraic equations
\begin{equation}
\label{BB24}
\sum _{i,j}{\left(\alpha _{i
}+1-\delta _{ij}\right)}\beta _{%
ij}c_{\alpha +e_{i}-e_{j}}=
\beta ^{s}_{1}{\left(c_{\alpha }
-d_{\alpha }\right)};\quad |
\alpha |=m.
\end{equation}
Here $ e_{i}$ and $e_{j}$ denote canonical basis vectors in $
\mathbf{R}^{2n-2}$ and $ \delta _{ij}$ is the Kronecker symbol.\par

\par Clearly (\ref{BB24}) amounts to showing that
$\beta ^{s}_{1}$
 is not an eigenvalue of the linear map
$\tilde{B}$
 on 
$
\mathbf{C}^{\tilde{n}}$
with 
\begin{equation*}
{\tilde{n}=\#{\left\{\alpha 
\in {\left({\bf N\cup }{\left\{0\right\}}
\right)}^{2n-2}|\quad |
\alpha |=m\right\}}={{{\left(m
+2n-3\right)}!}\over{{\left(2n-3\right)}
!m!}}} 
\end{equation*}
given by 
\begin{multline*}
\mathbf{C}^{\tilde{n}}\ni c={\left(
c_{\alpha }\right)}_{\alpha }
\\
\rightarrow 
{\big(\tilde{B}c\big)}_{\alpha }=\bigg(\sum _{i,j}
{\left(\alpha _{i}+1-\delta _{ij}\right)}
\beta _{ij}c_{\alpha +e_{i}-e_{j}
}\bigg)_{\alpha }\in \mathbf{C}^{\tilde{n}}.
\end{multline*}
Since 
$
\beta _{ij}=\beta _{ij}{\left(E\right)}
$
 depends continuously on 
$
E\in I_{0}$
we only need to show that 
\begin{equation}
\label{BBB100} 
\tilde{B}{\left(E_{0}\right)}
-\beta ^{s}_{1}{\left(E_{0}\right)}
I\text{  is invertible.}
\end{equation}
By the condition \eqref{H8} indeed (\ref{BBB100}) holds since
$m\leq m_{0}$
and the spectrum
\begin{equation*}
{\sigma {\left(\tilde{B}{\left(
E_{0}\right)}\right)}={\left\{
\beta {\left(E_{0}\right)}\cdot 
\alpha |\quad |\alpha |=m\right\}}.}
\end{equation*}
The latter is obvious if 
$\diag {\left(B^{s}{\left(E_{0}\right)}^{tr},B^{u}{\left(
E_{0}\right)}^{tr}\right)}
$ is diagonal. In general this may be seen 
by a perturbation argument, see [N, p. 37].\par 

\par Finally we define
\begin{equation*}
\Gamma =\gamma ^{{\left(m_{0}\right)}
}.
\end{equation*}

\par If we have 
$ m_{0}$ so large that $ \delta {\left(m_{0}+1\right)}> -\beta
^{s}_{1}{\left(E\right)} $ where $ \delta $ is given as in (\ref{b21}) we
infer by integrating (\ref{B15}) (since $ \lim_{t\rightarrow \infty }
t{{\partial _{\mu }h}\over{x_{n}}} =1$ ) that 
\begin{equation}
\label{b23}
\Gamma =\gamma _{1}+O{\left(|\gamma |^{%
2}\right)}=O{\left(t^{\beta ^{%
s}_{1}{\left(E\right)}+\epsilon ^{%
\prime}}\right)};\quad \epsilon ^{%
\prime}>0.
\end{equation}

\begin{remark}\label{rem:2.2}
  We could have used a different observable constructed by a
  similar iteration using as $ \gamma ^{{\left(1\right)}}$ a component of $
  \gamma $ corresponding to an eigenvector with eigenvalue $ \lambda
  {\left(E\right)}$ having $ \re {\left(\lambda {\left(E\right)} \right)}>0$.
  We would again need smoothness of the eigenvector and a non-resonance
  condition for $ \lambda {\left(E_{0}\right)}$, cf. \eqref{H7} and \eqref{H8}.
  The analogous observable $ \gamma ^{{\left(m\right)}}$ decreases as $
  t^{-\delta {\left(m+1\right)}}$ with no upper bound on $m$ (assuming $
  E_{0}$ is not a resonance of any order). But as we will see below, the
  correspondence between classical and quantum behavior is not so precise as
  to allow a similar statement in Quantum Mechanics. Thus it does not much
  matter which of these observables is used.
\end{remark}

%\subsection{Quantum $\Gamma$}
%\label{sec:quantum-gamma}
\noindent {\bf Quantum $\Gamma$}.
To get a statement like (\ref{b23}) in Quantum Mechanics we
need to quantize the classical symbol 
$
\gamma ^{{\left(m\right)}}=\gamma ^{%
{\left(m\right)}}{\left(x,\xi
\right)}$. We choose a quantization that takes into account
localizations of the states 
$
\psi =f{\left(H\right)}\psi $
 obeying \eqref{eq:8} and \eqref{eq:9}. We
fix $m=m_{0}$
 depending on an analogof the classical bound
(\ref{b21}), cf.  the classical case discussed above. Without
going into details, in the case of (\ref{b13}) this operator takes
the form
\begin{equation*}
\Gamma =\Gamma {\left(t\right)}
={\left(p-\xi{\left(E_{0}\right)}
\right)}\cdot \omega_{l}{\left(E_{
0}\right)}+B_{1}{\left(t\right)}
;\quad B_{1}{\left(t\right)}
\text{  bounded}.
\end{equation*}
We want $B_1(t)$ to be bounded to facilitate our uncertainty principle
argument (see Section 6). The fact that this works even though the
classical $\Gamma$ does not have this form rests on the localizations
of $\psi$.  
Strictly speaking, to get the above expression we first make the modification
of the classical $ \Gamma $ of dividing by the constant
$
c_{l}=\partial _{\eta _{l}
}\gamma _{1}{\left(w,E_{0}\right)}_{|w=0}$
 and then taking the real part; we shall not discuss
the case of (\ref{bb13}) here. We show the following analog of
(\ref{b23}):\par 

\par Given $ \sigma >0$ we have for some
$ \Gamma $ of this form the strong localization
\begin{equation}
\label{H14}
\left\|1_{[t^{\sigma -1},\infty )
}{\left(|\Gamma |\right)}
e^{-itH}\psi \right\|\rightarrow 0\text{ for }t\rightarrow 
\infty .
\end{equation}

\par We notice that (\ref{H14}) is a weaker bound than (\ref{b23}); to control
various commutators we need to have 
$ \sigma $ positive. On the other hand it may appear somewhat
surprising that such localization result can be proved at all for
$
\sigma <2^{-1}$. According to folklore wisdom there is usually a
strong connection for pseudodifferential operators between the
functional calculus and the pseudodifferential calculus, see for example
[DG, Appendix D]. In our case one might think that (\ref{H14}) is
equivalent to a statement like
\begin{equation*}
e^{-itH}\psi\approx a^{w}_{t}{\left(
x,p\right)}e^{-itH}\psi \text{  for }
t\rightarrow \infty , 
\end{equation*}
where the symbol $a_{t}=h{\left(t^{1-\sigma }\re {\left(c^{-1}_{l}\gamma
        ^{{\left( m_{0}\right)}}\right)} \right)}$ for suitable $ h\in
C^{\infty }_{0}{\left(\mathbf{R}\right)} $ and $ \gamma ^{{\left(m_{0}\right)}
}$ given by the classical symbol (possibly modified by cut-offs) discussed
above. However for $ \sigma <2^{-1}$ such symbols $ a_{t}$ do not fit into any
standard (parameter-dependent) pseudodifferential calculus which by
        the uncertainty principle  essentially
would require the uniform bounds $ \partial ^{\beta }_{\xi}\partial ^{\alpha
}_{x}a_{t}=O{\big(t^{\delta _{2}|\beta |-\delta _{1}|\alpha |}\big)}$ with
${\delta _{2}}<\delta _{1}$. As a consequence we shall base our proof of
(\ref{H14}) on a functional calculus approach. Using a differential
\textit{equality} related to (\ref{B15}) we can indeed bound certain quantum
errors in a calculus even for $ \sigma <2^{-1}$. It is important that we can
take $ \sigma $ small; see below. Somewhat related problems were
studied in [G1] and [CHS1].\par
\begin{remarks*}\hfill 
\begin{enumerate}
\item There is a subtle point suppressed in the above discussion which is
very important technically. Although the $t^{-\delta}$--localization
proved in step I)  is needed to construct the quantum $\Gamma$ and
prove \eqref{H14}, its full force cannot be used for this purpose. The
reason is that an effective use of the operator calculus  limits the strength of this
localization (this is basically
the uncertainty principle again). Thus using a strong $t^{-\delta}$--localization results in a weaker localization for $\Gamma$. The full force of the $t^{-\delta}$--localization is only exploited at the very end of the proof of Theorem \ref{thm:1.1} in Section 7.

\item Another technical point not discussed here is the use of a certain
hierarchy of localizations in the construction of $\Gamma$ (and
$\bar{A}$ below) necessary because of the variation of
$(\omega(E),\xi(E))$ with $E$. The fact that our procedure here
actually works may look  almost miraculous at first glance (see
\eqref{Ff1} and \eqref{Ff2}).
\end{enumerate}
\end{remarks*}

%\subsection{Implementing the uncertainty principle}
%\label{sec:impl-uncert-princ}
\noindent {\bf Implementing the uncertainty principle}.
The last step in our proof of Theorem~\ref{thm:1.1} is the decisive one; here
Quantum Mechanics enters crucially. We show that a localization similar to the
classical bound (\ref{b21}) and (\ref{H14}) are incompatible unless $ \psi
=0$. First fix $ \delta >0$ in agreement with (\ref{b21}). More precisely we
need the localization 
\begin{equation}
\label{HHH15}
e^{-itH}\psi \approx h_{2}{\left(\bar{A}\right)}e^{-itH}\psi 
\rightarrow 0\text{ for }t\rightarrow \infty ,
\end{equation}
for some 
$h_{2}\in C^{\infty }_{0}{\left(
\mathbf{R}\right)}$
 and some operator of the form 
 \begin{equation*}
{\bar{A}=t^{\delta -1}x_{l}+B_{2}{\left(t\right)};\quad 
B_{2}{\left(t\right)}=O{\big(
t^{\delta }\big)},\quad x_{%
l}=x\cdot \omega_{l}{\left(E_{0}
\right)}.}
\end{equation*}
Then fix any $ \sigma \in {\left(0,\delta \right)}$ and introduce with $
\Gamma $ as in (\ref{H14}) the operator $ \bar{H}=t^{1-\delta
}\Gamma$.

\par We prove a global Mourre estimate
\begin{equation}
\label{H15} 
i{\left[\bar{H},\bar{A}
\right]}\geq 2^{-1}I.
\end{equation}

\par Abstract Mourre theory and (\ref{H15}) lead to the bound
\begin{equation}
\label{H16} 
\left\|h_{2}{\left(\bar{A}
\right)}h_{1}{\left(t^{\delta 
-\sigma }\bar{H}\right)}
\right\|\leq Ct^{{\left(\sigma -\delta \right)}
/2},
\end{equation}
valid for all $ h_{1},h_{2}\in C^{\infty }_{0} {\left(\mathbf{R}\right)}$.\par

\par Finally picking localization functions in agreement with (\ref{HHH15})
and (\ref{H14}) we conclude from (\ref{H16}) 
that 
\begin{equation*}
e^{-itH}\psi \approx h_{2}{\left(\bar{A}\right)}h_{1}{\left(
t^{\delta -\sigma }\bar{H}\right)}
e^{-itH}\psi \rightarrow 0\text{ for }t\rightarrow \infty ,
\end{equation*}
completing the proof.\par

\section{Preliminaries}
\label{sec:3}

\par We use the notation 
$\Psi {\left(m,g\right)}$ for the space of operators given by quantizing
symbols in the symbol class $ S{\left(m,g\right)}$ as defined by
[H\"o, (18.4.6)]. For the weight functions $m$ and metrics $g$ relevant for
this paper it does not matter here whether {``}quantize{''} refers to Weyl or
Kohn-Nirenberg quantization. For $ a\in S{\left(m,g\right)}$ we use the
notation $ a^{w}{\left(x,p\right)}$ to denote the Weyl quantization of $a$. We
refer the reader to [DG, Appendix D] and [H\"o, Chapter 18] for a detailed
account of the calculus of pseudodifferential operators. We shall deal with
various kinds of parameter-dependent symbols. In one case the parameter is
time $ t\geq 1$ and for that we introduce the following shorthand
notation.\par

\begin{definition}
\label{def:3.1}
A family $ {\left\{a_{t}|t\geq 1\right\}}$ of symbols in $
S{\left(m,g\right)}$ is said to be uniform in $ S{\left(m,g\right)}$ if for
all semi-norms $ ||\cdot ||_{k}$ on $ S{\left(m,g\right)}$ (cf. [H\"o,
(18.4.6)]) $ \sup_{t}\| a_{t}\|_{k}<\infty $. In this
case we write $ a_{t}\in S_{unif}{\left(m,g\right)} $ and $
a^{w}_{t}{\left(x,p\right)} \in \Psi _{unif}{\left(m,g\right)} $.
\end{definition}

\par Given this uniformity property various bounds from the calculus of
pseudodifferential operators are uniform in the parameter (by continuity
properties of the calculus). \par 

\par We shall also deal with parameter-dependent metrics. Specifically
we shall consider for $ 0\leq \delta _{2}<\delta _{1}\leq 1$ and $ t\geq 1$
\begin{equation}
\label{c-10} 
g_{t}=g^{\delta _{1},\delta _{2}
}_{t}=t^{-2\delta _{1}}dx^{2
}+t^{2\delta _{2}}d\xi^{2}.
\end{equation}
Similarly to Definition~\ref{def:3.1} we shall write (for given $ l\in
\mathbf{R}$), $ a_{t}\in S_{unif}{\left(t^{l},g_{t}\right)}$ and $
a^{w}_{t}{\left(x,p\right)}
\in \Psi _{unif}{\left(t^{l},g_{%
      t}\right)}$ meaning that for all (time-dependent) semi-norms
$\sup_{t}\|a_{t}\|_{t,k}<\infty $. Also in this case various bounds from the
calculus of pseudodifferential operators will be uniform in the parameter.
Some extensions of this idea will be used without further comment.\par

\par One may verify that \eqref{eq:10} follows from \eqref{eq:8} by applying a
partition of unity to the $f$ of any state $ \psi =f{\left(H\right)}\psi $ of
\eqref{eq:8} to decompose it as $ f=\sum f_{i}$ and by noticing that
\eqref{eq:8} remains
valid for the sharper localized states $ \psi \rightarrow \psi
_{i}=f_{i}{\left( H\right)}\psi $.  (Notice that if $\supp{\left(f_{i}\right)}
$ is located near $ E_{i}$ this leads to $ 
t^{-1}x\approx k{\left(E_{i}\right)} \omega{\left(E_{i}\right)}$ and $
p\approx \xi{\left(E_{i}\right)} $ along $ \psi _{i}{\left(t\right)}$.) The latter
follows readily upon commutation and applying Lemma~\ref{lem:3.2} stated below. The same
argument shows that indeed $ \mathcal{H}_{0}$ is $H${--}reducing. (This property
may also be verified without appealing to Lemma~\ref{lem:3.2}.)\par

\par Pick non-negative 
$g_{1},\tilde{g}_{1},\tilde{\tilde{g}}_{1} \in C^{\infty
}_{0}{\left(\mathbf{R}^{ n}\right)}$ such that $ g_{1}=1$ in a (small)
neighborhood of $k{\left(E_{0}\right)}\omega_{0}$, $\tilde{g}_{1}=1$
in a
neighborhood of $ \supp{\left(g_{1}\right)} $ and $ \tilde{\tilde{g}}_{1}=1$
in a neighborhood of $\supp{\left(\tilde{g}_{1}\right)}$. Similarly, pick 
non-negative   $g_{2},\tilde{g}_{2},\tilde{\tilde{g}}_{2}\in C^{\infty
}_{0}{\left(\mathbf{R}^{n}\right)}$ such that $ g_{2}=1$ in a neighborhood of $
\xi_{0}$, $ \tilde{g}_{2}=1$ in a neighborhood of $\supp{\left(g_{2}\right)} $ and $
\tilde{\tilde{g}}_{2}=1$
in a neighborhood of 
$\supp{\left(\tilde{g}_{2}\right)}$. We suppose 
$\supp {\left(\tilde{\tilde{g}}_{1}\right)}\times 
\supp {\left(\tilde{\tilde{g}}_{2}\right)}\subseteq 
\mathcal{U}_{0}$
 (with 
$\mathcal{U}_{0}$
 given as in \eqref{eq:9}), and in fact that the
supports are so small that for some 
$t_{0}\geq 1$ the symbol
\begin{equation}
\label{c1} 
\begin{split}
  h_{t}{\left(x,\xi\right)}:&=h
{\left(x,\xi\right)}\tilde{g}_{%
  1}{\left(t^{-1}x\right)}\tilde{g}_{2}{\left(\xi\right)}
\\
&=h{\left(r_{0}\hat{x},\xi \right)}\tilde{g}_{1}
{\left(t^{-1}x\right)}\tilde{g}_{%
  2}{\left(\xi\right)};
\quad t\geq t_{0},
\end{split}
\end{equation}
cf. (\ref{H1}). By the assumption (\ref{H2}) we then have 
\begin{equation}
\label{c2} 
h_{t}\in S_{unif}{\left(1,g_{0}
\right)}\cap S_{unif}{\left(1,g^{%
1,0}_{t}\right)}.
\end{equation}

\begin{lemma}\label{lem:3.2}
  For all
  $ f\in C^{\infty }_{0}{\left(\mathbf{R}\right)} $ the family
  \begin{equation}
    \label{CC2} 
     f{\left(h^{w}_{t}{\left(x,p\right)} \right)}\in \Psi _{unif}{\left(
        1,g_{0}\right)}\cap \Psi _{unif} {\left(1,g^{1,0}_{t}\right)} 
\end{equation}
and 
\begin{equation}
\label{cc2} 
 \left\|g_{1}{\left(t^{-1}x\right)} g_{2}{\left(p\right)}{\left\{
    f{\left(h^{w}_{t}{\left(x,p\right)} \right)}-f{\left(H\right)}\right\}}
\right\|=O{\left(t^{-\infty }\right)}.
\end{equation}
\end{lemma}

This lemma facilitates the transition between the functional calculus and the pseudo-differential operator calculus, both of which are used in this paper.

\begin{proof}
  As for (\ref{CC2}) we may proceed as in the proofs of [DG, Propositions
  D.4.7 and D.11.2].  (One verifies the Beals criterion using the
  representation (\ref{82a}) given below and the calculus of
  pseudodifferential operators.)\par

\par For (\ref{cc2}) we let 
$ B=h^{w}_{t}{\left(x,p\right)} $ and $ G=h^{w}{\left(x,p\right)}-h^{w
}_{t}{\left(x,p\right)}$. By (\ref{82a}) 
\begin{equation}
\label{c3} 
 f{\left(h^{w}_{t}{\left(x,p\right)}
\right)}-f{\left(H\right)}={{%
  1}\over{\pi }}\int _{\mathbf{C}} {\left(\bar{\partial }\tilde{f} \right)}{\left(z\right)}{\left(
B-z\right)}^{-1}G{\left(H-z\right)}^{-1}dudv.
\end{equation}
For any large $ m\in \mathbf{N}$ we may decompose
\begin{multline}
\label{ccc3} 
 {\left(B-z\right)}^{-1}G=
\\
\sum _{k=1}^{m}ad^{k}_{B}{\left(
G\right)}{\left(B-z\right)}^{%
-k}+{\left(B-z\right)}^{-1}a d^{m}_{B}{\left(G\right)}
{\left(B-z\right)}^{-m},
\end{multline}
yielding (by the calculus) 
\begin{equation}
\label{c4} 
\begin{split}
  &g_{1}{\left(t^{-1}x\right)} g_{2}{\left(p\right)}{\left(
      B-z\right)}^{-1}G=\sum _{k=1
}^{m}R_{k}{\left(B-z\right)}^{%
  -k}
\\
&\qquad\qquad+g_{1}{\left(t^{-1}x\right)} g_{2}{\left(p\right)}{\left(
    B-z\right)}^{-1}ad^{m}_{B}
{\left(G\right)}{\left(B-z\right)}^{%
  -m};
\\
& R_{k}=O{\left(t^{- \infty }\right)}.
\end{split}
\end{equation}
By (\ref{H2}), $ ad^{m}_{B}{\left(G\right)}
\in \Psi _{unif}{\left({\left\langle x\right\rangle }^{%
      l-m},g_{0}\right)}$ and therefore $ ad^{m}_{B}{\left(G\right)}
=O{\left(t^{l-m}\right)}$, whence 
\begin{equation}
\label{cc3} 
 \left\|g_{1}{\left(t^{-1}x\right)} g_{2}{\left(p\right)}{\left(
B-z\right)}^{-1}G\right\|\leq Ct^{l-m}|\im z|^{-{\left(m+1\right)} }
\end{equation}
uniformly in $ z\in \supp {\big(\tilde{f}
  \big)}$.\par

\par Clearly (\ref{cc2}) follows from (\ref{c3}) and (\ref{cc3}). 
\end{proof}

\begin{remark}\label{rem:3.3}
The statements of Lemma~\ref{lem:3.2} extend to any smooth function $f$ with $
{{d^{k}}\over{d\lambda ^{k}}} 
f{\left(\lambda \right)}=O{\left(\lambda ^{%
      m-k}\right)}$ (for fixed $ m\in \mathbf{R}$); in particular Lemma~\ref{lem:3.2}
holds for $ f{\left(\lambda \right)}=\lambda$.
\end{remark}

\begin{definition}\label{def:3.4}
Let
  $ \mathcal{F}_{+}$ denote the largest set of $ F=F_{+}\in C^{\infty
  }{\left(\mathbf{R} \right)}$, such that $ 0\leq F\leq 1$, $ F^{\prime}\geq
  0,F^{\prime}\in C^{\infty
}_{0}{\left({\left({{1}\over{%
            2}},{{3}\over{4}}\right)} \right)}$, $
F{\left({{1}\over{2}}\right)} =0$, $F{\left({{3}\over{4}} \right)}=1$ and $
\sqrt{1-F} $, $\sqrt{F}$, $\sqrt{F^{\prime}} \in C^{\infty }$, which is stable
under the maps $ F\rightarrow F^{m}$  and  $ F\rightarrow 1-{\left(1-
    F\right)}^{m};$ $m\in \mathbf{N}$. Let $ \mathcal{F}_{-}$ denote the
set of functions $ F_{-}=1-F_{+}$ where $ F_{+} \in \mathcal{F}_{+}$.
\end{definition}

\par We shall in Section~\ref{sec:5} use a modification of the abstract calculus
[D, Lemma A.3 (b)], see also [DG, Appendix C], [G1, Appendix] or
[M{\o}].

\begin{lemma}\label{lem:3.5}  
  Suppose $ \bar{H}$ and $B$ are self-adjoint operators on a complex Hilbert
  space $ \mathcal{H}$ , and that $ {\left\{B{\left(t\right)}\mid
      t>t_{0}\right\}}$ is a family of self-adjoint operators on $
  \mathcal{H}$ with the common domain $ \mathcal{D}{\left(B{\left(t\right)}
    \right)}=\mathcal{D}{\left(B\right)} $. Suppose that $ \bar{H}$ is
  bounded, that the commutator form $ i{\left[\bar{H},B{\left(
          t\right)}\right]}$ defined on $ \mathcal{D}{\left(B\right)}$ is a
  symmetric operator with same (operator) domain $
  \mathcal{D}{\left(B\right)}$ and that the $
  \mathcal{B}{\left(\mathcal{H}\right)}\hbox{\rm -valued} $ function $
  B{\left(t\right)}{\left(B{\left(t\right)}-i\right)}^{ -1}$ is continuously differentiable.
  Then

\par (A)\enspace For any given $ F\in C^{\infty }_{0}{\left(\mathbf{R}\right)} $ we let $
  \tilde{F}\in C^{\infty }_{0 }{\left(\mathbf{C}\right)}$ denote
  an almost analytic extension. In particular 
  \begin{equation}
    \label{82a} 
    F{\left(B{\left(t\right)}\right)} ={{1}\over{\pi }}\int _{\mathbf{C}
    }{\left(\bar{\partial} \tilde{F}\right)}{\left(
        z\right)}{\left(B{\left(t\right)} -z\right)}^{-1}dudv,\quad z=u+
iv.
\end{equation}

The $ \mathcal{B}{\left(\mathcal{H}\right)}$-valued function $
F{\left(B{\left(t\right)}\right)} $ is continuously differentiable, and
introducing the Heisenberg derivative $
\mathbf{D}={{d}\over{dt}}+i{\left[\bar{H},\cdot \right]}$, the form $$
{{d}\over{dt}}F{\left(B{\left(t
\right)}\right)}+i{\left[\bar{H},F{\left(B{\left(t\right)} \right)}\right]}$$
is given by the bounded 
operator 
\begin{equation}
\label{83a} 
 \mathbf{D}F{\left(B{\left(t\right)}\right)} =-{{1}\over{\pi }}\int _{\mathbf{C}
}{\left(\bar{\partial} \tilde{F}\right)}{\left(
    z\right)}{\left(B{\left(t\right)} -z\right)}^{-1}{\left(\mathbf{D}B{\left(
        t\right)}\right)}{\left(B{\left( t\right)}-z\right)}^{-1}dud v.
\end{equation}

\par  In particular if 
$ \mathbf{D}B{\left(t\right)}$ is bounded then for any $ \epsilon >0$ (with $
{\left\langle z\right\rangle }={\left(1+|z|^{
      2}\right)}^{{{1}\over{2}} }$) 
\begin{equation}
\label{8883a} 
\left\|\mathbf{D}F{\left(B{\left(t\right)} \right)}\right\|\leq C_{\epsilon
}\sup_{z\in \mathbf{C}}{\Big({\left\langle z\right\rangle }^{%
      \epsilon +2}|\im z|^{-2}| {\big(\bar{\partial }\tilde{F}
      \big)}{\left(z\right)}| \Big)}
 ||\mathbf{D}B{\left(t\right)} ||.
\end{equation}

\par {(B)}\enspace
Suppose in addition that we can split $ \mathbf{D}B{\left(t\right)}=D{\left(
    t\right)}+D_{r}{\left(t\right)} $, where $ D{\left(t\right)}$  and
$D_{r} {\left(t\right)}$ are symmetric operators on $
    \mathcal{D}{\left(B\right)}$ and that the form $
    i^{k}ad^{k}_{B{\left(t\right)} 
}{\left(D{\left(t\right)}\right)} =i{\big[i^{k-1}ad^{k-1}_{B{\left(
          t\right)}}{\left(D{\left(t\right)}
      \right)},B{\left(t\right)}\big]} $ for $k=1$ defined on $ \mathcal{
  D}{\left(B\right)}$ is a symmetric operator on $
\mathcal{D}{\left(B\right)}$; $ ad^{0}_{B{\left(t\right)}}
{\left(D{\left(t\right)}\right)} =D{\left(t\right)}$. (No assumption is made
for the form when $k=2$.) Then the contribution from $ D{\left(t\right)}$ to
(\ref{83a}) can be written as

\begin{align}\label{84a} 
-{{1}\over{\pi }}\int
  _{\mathbf{C}} &{\big(\bar{\partial} \tilde{F}\big)}{\left(
      z\right)}{\left(B{\left(t\right)} -z\right)}^{-1}D{\left(t\right)}
{\left(B{\left(t\right)}-z\right)}^{
  -1}dudv\nonumber
\\
&={{1}\over{2}}{\left(
    F^{\prime}{\left(B{\left(t\right)} \right)}D{\left(t\right)}+D{\left(
        t\right)}F^{\prime}{\left(B{\left( t\right)}\right)}\right)}
+R_{1}{\left(t\right)};
\\
R_{1}{\left(t\right)}&={{1}\over{
  2\pi }}\int _{\mathbf{C}}{\big( \bar{\partial}\tilde{F}
\big)}{\left(z\right)}{\left(
B{\left(t\right)}-z\right)}^{
-2}\nonumber\\
&\qquad\qquad\qquad\cdot ad^{2}_{B{\left(t\right)} }{\left(D{\left(t\right)}\right)}
{\left(B{\left(t\right)}-z\right)}^{
  -2}dudv.\nonumber
\end{align}
For all  $ f\in C^{\infty }_{0}{\left( \mathbf{R}\right)}$
\begin{align}\label{85a} 
  & {{\tfrac{1}{2}}{\left(f^{ 2}{\left(B{\left(t\right)}\right)}
        D{\left(t\right)}+D{\left(t\right)} f^{2}{\left(B{\left(t\right)}
          \right)}\right)}}\nonumber
  \\
  &=f{\left( B{\left(t\right)}\right)}D{\left(
      t\right)}f{\left(B{\left(t\right)} \right)}+R_{2}{\left(t\right)} ;
  \\
  R_{2}{\left(t\right)}&=\frac{1}{2\pi ^2}\int _{\mathbf{C}}\int
  _{\mathbf{C}}{\big(\bar{\partial} \tilde{f}\big)}{\left(
      z_{2}\right)}{\big(\bar{\partial }\tilde{f}\big)}
  {\left(z_{1}\right)}{\left(
      B{\left(t\right)}-z_{2}\right)}^{-1}{\left(B{\left(t\right)}-z_{1}\right)}^{-1}\nonumber
\\
&ad^{ 
    2}_{B{\left(t\right)}}{\left( D{\left(t\right)}\right)}{\left(
      B{\left(t\right)}-z_{1}\right)}^{ -1}{\left(B{\left(t\right)}-z_{
        2}\right)}^{-1}du_{1}dv_{ 1}du_{2}dv_{2}.\nonumber
\end{align}

\par {(C)} \enspace
Suppose in addition to previous assumptions that for all $ t>t_{0}$ the form $
i{\left[D{\left(t\right)},B{\left( t\right)}\right]}$ extends from $ \mathcal{
  D}{\left(B\right)}$ to a bounded self-adjoint operator. Similarly suppose
the operator $ D_{r}{\left(t\right)}$ extends to a bounded self-adjoint
operator. Then for all $ F\in \mathcal{ F}_{+}$ the $ \mathcal{
  B}{\left(\mathcal{H}\right)}$-valued function $
F{\left(B{\left(t\right)}\right)} {\left(B-i\right)}^{-1}$ is continuously
differentiable, and there is an almost analytic extension with
\begin{equation}
\label{87a} 
\big|{\big(\bar{\partial}\tilde{F}\big)}{\left(z\right)}
\big|\leq C_{k}{\left\langle z\right\rangle }^{%
  -1-k}|\im z|^{k};\ k\in \mathbf{N},
\end{equation}
yielding the representation 
\begin{equation}
\label{86a} 
 \mathbf{D}F{\left(B{\left(t\right)}\right)} =F^{\prime{{1}\over{2}}}{\left(
    B{\left(t\right)}\right)}D{\left( t\right)}F^{\prime{{1}\over{2}}
}{\left(B{\left(t\right)}\right)} +R_{1}{\left(t\right)}+R_{2}
{\left(t\right)}+R_{3}{\left( t\right)},
\end{equation}
where $ R_{1}{\left(t\right)}$ is given by (\ref{84a}), $
R_{2}{\left(t\right)}$ by (\ref{85a}) with $ f=\sqrt{F^{\prime}} $ and $
R_{3}{\left(t\right)}$ is the contribution from $ D_{r}{\left(t\right)}$ to
(\ref{83a}).
\end{lemma}

\begin{remarks*}\hfill
  \begin{enumerate}
  \item The left hand side of (\ref{86a}) is initially defined as a form on $
    \mathcal{D}{\left(B\right)}$ while the terms on the right hand side are
    bounded operators. We shall use the stated representation formulas for
    bounding these operators in an application in the proof of
    Proposition~\ref{prop:5.1}; this will be in the spirit of (\ref{8883a})
    although somewhat more sophisticated.
  \item  There are versions of Lemma~\ref{lem:3.5} % here were an error
                                                  % refered to lemma 3.4
    without the assumption that $\bar{H}$ is bounded; they are not needed in this
    paper.
\end{enumerate}
\end{remarks*}

\section{$t^{-\delta }${--}localization}
\label{sec:4}

\par Let 
$ \psi =f{\left(H\right)}\psi $ be any state obeying \eqref{eq:8} and
\eqref{eq:9} with $f$ supported in a very small neighborhood of $
E_{0}$ (in agreement with the smallness of the neighborhood $ I_{0}$
of Theorem~\ref{thm:1.1}). Let $
g_{1},\tilde{g}_{1},g_{2},\tilde{g}_{2}\in C^{\infty
}_{0}{\left(\mathbf{R}^{n}\right)} $ be given as in (\ref{c1}) and
(\ref{cc2}). In particular we have 
\begin{equation*}
g_{1}{\left(k{\left(E\right)}
    \omega{\left(E\right)}\right)} f{\left(E\right)}=f{\left(E\right)},\; g_{2}{\left(\xi{\left(E\right)}
  \right)}f{\left(E\right)}=f{\left( E\right)}. 
\end{equation*}

Consider for 
$ t,\kappa \geq 1$ symbols 
\begin{equation}
\label{d1} 
a=a_{t,\kappa }{\left(x,\xi\right)}
=F_{+}{\left(\kappa q^{-}{\left(
x,\xi\right)}\right)}\tilde{g}_{%
1}{\left(t^{-1}x\right)}\tilde{g}_{2}{\left(\xi\right)},
\end{equation}
where $ F_{+}$ is given as in Definition~\ref{def:3.4} and $q^{-}$ is
built from the $ q^{-}$ of (\ref{b14}) by writing $
q^{-}=q^{-}{\left(w{\left(E\right)},E\right)}$ and substituting for
$E$ the symbol $ h{\left(r_{0}\hat{x},\xi \right)}$ cf. (\ref{c1}),
\begin{equation}
\label{d2} 
q=q^{-}{\left(w{\left(h{\left(r_{%
0}\hat{x},\xi\right)}
\right)},h{\left(r_{0}\hat{x},\xi\right)}\right)}.
\end{equation}

\par We shall consider 
$ \kappa \in {\left[1,t^{\nu }\right]} $ with $ \nu >0$. To have a
good calculus for the symbol $a$ we need $ \nu <1/2$. Notice that
\begin{equation}
\label{d3} 
a_{t,\kappa }\in S_{unif}{\left(1,g^{%
1-\nu ,\nu }_{t}\right)},
\end{equation} 
and that the ``Planck constant'' for this symbol class is $h=t^{2\nu-1}$.

\par 
Denoting by $ {\left\langle \cdot \right\rangle }_{t}$ the expectation
in the state $ \psi {\left(t\right)}=e^{-itH} \psi $ we have the
following localization.\par

\begin{lemma}\label{lem:4.1}
  For all $ \nu \in {\left(0,2/5\right)}$
  \begin{equation}
\label{D3} 
 {\left\langle a^{w}_{t,t^{\nu }}{\left( x,p\right)}\right\rangle
}_{t}\rightarrow 0\;\hbox {\rm for }t\rightarrow \infty .
\end{equation}
\end{lemma}

This lemma is a quantum version of (\ref{b18}).

\begin{proof}
  We shall use a scheme of proof from [D]. Let
  \begin{equation}
\label{d4} 
A_{t,\kappa }=L_{1}{\left(t\right)}^{
  *}a^{w}_{t,\kappa }{\left(x,p\right)} L_{1}{\left(t\right)};\quad
L_{1}{\left(t\right)}=g_{1} {\left(t^{-1}x\right)}g_{2}
{\left(p\right)}.
\end{equation}

\par From \eqref{eq:10} and the calculus of pseudodifferential
operators we immediately conclude that for fixed $ \kappa $
\begin{equation*}
{\left\langle A_{s,\kappa }\right\rangle }_{%
s}\rightarrow 0\text{ for }s\rightarrow \infty ,
\end{equation*}
yielding 
\begin{equation}
\label{d6} 
-{\left\langle A_{t,\kappa }\right\rangle }_{%
  t}=\int _{t}^{\infty }{\left\langle \mathbf{D}A_{s,\kappa
    }\right\rangle }_{s} ds,
\end{equation}
where $ \mathbf{D}$ refers to the Heisenberg derivative $
\mathbf{D}={{d}\over{ds}}+i{\left[H,\cdot \right]}$. We shall show
that the expectation of $ \mathbf{D}A_{s,\kappa }$ is essentially
positive (in agreement with (\ref{b17})). Up to terms $
O{\left(s^{-\infty }\right)}$ we may replace $ \mathbf{D}$ by $
\mathbf{D}_{s}={{d}\over{ds}}+i{\left[
    h^{w}_{s}{\left(x,p\right)},\cdot \right]}$, cf.
Remark~\ref{rem:3.3}. First we notice that 
\begin{equation}
\label{d7} 
 g_{2}{\left(p\right)}g_{1} {\left(s^{-1}x\right)}{\left(
    \mathbf{D}_{s}a^{w}_{s,\kappa }{\left(
        x,p\right)}\right)}g_{1}{\left( s^{-1}x\right)}g_{2}{\left(
    p\right)}\geq -Cs^{5\nu -3},
\end{equation}
where $C>0$ is independent of $ \kappa \in
{\left[1,t^{\nu }\right]} $. \par

\par This bound follows from the calculus. The classical Poisson bracket
contributes by a positive symbol when differentiating $
q{\left(x,\xi\right)}$. The Fefferman-Phong inequality (see [H\"o,
Theorem 18.6.8 and  Lemma 18.6.10]) for this term yields the lower
bound $ O{\left(s^{\nu -1}{\left(s^{2\nu -1}\right)}^{2}\right)}
=O{\left(s^{5\nu -3}\right)}$.\par

\par Hence (uniformly in $ \kappa $)
\begin{align*}
  &\mathbf{D}A_{s,\kappa }\geq {\left\{T+T^{*
      }\right\}}-Cs^{5\nu -3};
  \\ 
&T=g_{2}{\left(p\right)}
g_{1}{\left(s^{-1}x\right)}
a^{w}_{s,\kappa }{\left(x,p\right)}
\mathbf{D}_{s}{\left(g_{1}{\left(
s^{-1}x\right)}g_{2}{\left(
p\right)}\right)}.
\end{align*}
For the contribution from the first term on the right hand side we
invoke \eqref{eq:9} after symmetrizing. We conclude that
\begin{equation}
\label{d9} 
 \int _{t}^{\infty }{\left\langle {\bf D}A_{s,\kappa }\right\rangle
}_{s}d
s\geq o{\left(t^{0}\right)}-Ct^{%
  5\nu -2}\text{ uniformly in }\kappa \in {\left[ 1,t^{\nu
    }\right]}.
\end{equation}
Pick $ \kappa =t^{\nu }$.\par

\par By combining (\ref{d6}), (\ref{d9}), and the Fefferman-Phong inequality, we infer that 
\begin{equation*}
\left\langle A_{t,t^{\nu }}\right\rangle _{%
t}\rightarrow 0\text{  for }t\rightarrow \infty ,
\end{equation*}
and therefore (\ref{D3}). 
\end{proof}

\par Let 
$ q^{+},q^{s}$ and $ q^{u}$ be given as in (\ref{b14}) upon
substituting the symbol $ h{\left(r_{0}\hat{x},\xi \right)}$ for $E$,
cf. the use of $ q^{-}$ above. We introduce the symbols 
\begin{align*}
  &a^{1}_{t}=t^{\nu -1}
  q^{-}{\left(x,\xi\right)}F^{
    \prime}_{+}{\left(t^{\nu }q^{
        -}{\left(x,\xi\right)}\right)}
  \tilde{g}_{1}{\left(t^{
        -1}x\right)}\tilde{g}_{
    2}{\left(\xi\right)},
  \\ 
  &a^{2}_{t}=t^{\nu -
    1}q^{+}{\left(x,\xi\right)}
  F^{\prime}_{+}{\left(t^{\nu }
      q^{-}{\left(x,\xi\right)}\right)}
  \tilde{g}_{1}{\left(t^{
        -1}x\right)}\tilde{g}_{
2}{\left(\xi\right)}. 
\end{align*}

\par We get the following integral estimate from the above proof
employing the uniform boundedness of the family of {``}propagation
observables{''} $ A_{t,t^{\nu }}$, cf. a standard argument of
scattering theory see for example [D, Lemma A.1 (b)].\par

\begin{lemma}\label{lem:4.2}
  In the state $ \psi _{1}{\left(t\right)}=L_{ 1}{\left(t\right)}\psi
  {\left( t\right)}$
  \begin{equation*}
    \int _{1}^{\infty }{\left(
        |{\left\langle {\left(a^{1}_{t}\right)}^{
              w}{\left(x,p\right)}\right\rangle }_{
          t}|+|{\left\langle {\left(a^{2}_{
                  t}\right)}^{w}{\left(x,p\right)}
          \right\rangle }_{t}|\right)}d
    t<\infty .
\end{equation*}
\end{lemma}

\begin{proof}
  We substitute $ \kappa =t^{\nu }$ in the construction (\ref{d4}).
  Then up to integrable terms the left hand side of (\ref{d7}) (with $
  s=t$ ) is given by $ c^{w}_{t}{\left(x,p\right)} $ with
  \begin{multline*}
    c_{t}{\left(x,\xi\right)}
    =g_{2}{\left(\xi\right)}^{
      2}g_{1}{\left(t^{-1}x\right)}^{
      2}\\{\left(\nu t^{\nu -1}q^{-}
        {\left(x,\xi\right)}+t^{\nu }
        {\left\{h{\left(x,\xi\right)},
            q^{-}{\left(x,\xi\right)}\right\}}
      \right)}F^{\prime}_{+}{\left(
        t^{\nu }q^{-}{\left(x,\xi\right)}
      \right)},
\end{multline*}
where $ {\left\{\cdot ,\cdot \right\}}$ signifies Poisson bracket.\par

\par We have the bounds for some $C>0$ and all large enough
$t$ 
\begin{equation*}
C^{-1}c_{t}{\left(x,\xi\right)}
\leq g_{2}{\left(\xi\right)}^{%
2}g_{1}{\left(t^{-1}x\right)}^{%
2}{\left(a^{1}_{t}{\left(x,\xi\right)}+a^{2}_{t}{\left(
x,\xi\right)}\right)}\leq Cc_{%
t}{\left(x,\xi\right)},
\end{equation*}
from which we readily get the lemma by the Fefferman-Phong inequality.
\end{proof}

\begin{remark}\label{rem:4.3}
  We shall not directly use Lemma~\ref{lem:4.2}. However the proof will be
  important. In particular we shall need the non-negativity of the
  above symbol $ c_{t}$.
\end{remark}

\par Let for 
$ t,\kappa \geq 1$ and $ 0<2\delta <\min {\left(\nu ,2\delta
    ^{s}\right)}$ with $ \nu <    2/5$ and $ \delta ^{s}$ as in
(\ref{b19}) (this number may be taken independent of $E$ close to $
E_{0}$, cf. Remarks~\ref{rem:2.1} (1)), 
\begin{multline*}
b_{t,\kappa }{\left(x,\xi\right)}
=F_{+}{\left(\kappa ^{-1}t^{2\delta 
}q^{s}{\left(x,\xi\right)}
\right)}F_{-}{\left(t^{\nu 
}q^{-}{\left(x,\xi\right)}
\right)}\\\tilde{g}_{1}
{\left(t^{-1}x\right)}\tilde{g}_{%
2}{\left(\xi\right)}\in S_{%
unif}{\left(1,g^{1-\nu ,\nu }_{%
t}\right)}.
\end{multline*}

\begin{lemma}\label{lem:4.4}
  For all $ \epsilon >0$ 
  \begin{equation}
    \label{d12} 
 \left\langle b^{w}_{t,t^{\epsilon }}{\left(
        x,p\right)}\right\rangle _{t}\rightarrow 0\text{ for
        }t\rightarrow \infty . 
\end{equation}
\end{lemma}
\begin{proof}
 We shall use another scheme of
proof from [D]. Let 
\begin{equation}
\label{d13} 
B_{t,\kappa }=L_{1}{\left(t\right)}^{
  *}b^{w}_{t,\kappa }{\left(x,p\right)} L_{1}{\left(t\right)},
\end{equation}
cf. (\ref{d4}), and write for any (large) $ t_{0}$
\begin{equation}
\label{d14} 
{\left\langle B_{t,\kappa }\right\rangle }_{t}={\left\langle
    B_{t_{0},\kappa }\right\rangle }_{
t_{0}}+\int _{t_{0}}^{t
}{\left\langle \mathbf{D}B_{s,\kappa }\right\rangle }_{
  s}ds.
\end{equation}

\par To show that the left hand side of (\ref{d14}) vanishes as
$ t\rightarrow \infty $ (with $ \kappa =t^{\epsilon }$ ) we look at
the integrand on the right hand side: As in the proof of
Lemma~\ref{lem:4.1} we may replace $ \mathbf{D}$ by $ \mathbf{D}_{s}$
up to a term $ r_{s,\kappa }$ such that 
\begin{equation*}
\int _{t_{0}}^{t}r_{%
s,\kappa }ds\rightarrow 0\text{ uniformly in }\kappa \geq 
1\text{ and }t\geq t_{0}\text{ as }t_{0}
\rightarrow \infty .
\end{equation*}

\par Using \eqref{eq:9} and Remark~\ref{rem:4.3} we may estimate the
integrand up to terms of this type as 
\begin{equation*}
\cdots \leq {\left\langle L_{1}{\left(s\right)}^{
*}{\left(b^{1}_{s,\kappa }\right)}^{
w}{\left(x,p\right)}L_{1}{\left(
s\right)}\right\rangle }_{s},
\end{equation*}
where 
\begin{align*}
  &b^{1}_{s,\kappa }{\left(
      x,\xi\right)}=\kappa ^{-1}s^{2\delta 
  }{\left(2\delta s^{-1}q^{s}{\left(
          x,\xi\right)}+{\left\{h{\left(
              x,\xi\right)},q^{s}{\left(x,
              \xi\right)}\right\}}\right)}
  c_{s,\kappa }{\left(x,\xi\right)}
  ;
  \\ 
  &c_{s,\kappa }{\left(x,\xi
    \right)}=F^{\prime}_{+}{\left(
      \kappa ^{-1}s^{2\delta }q^{s}{\left(
          x,\xi\right)}\right)}F_{-}
  {\left(s^{\nu }q^{-}{\left(
          x,\xi\right)}\right)}\tilde{g}_{%
    1}{\left(s^{-1}x\right)}\tilde{g}_{2}{\left(\xi\right)}.
\end{align*}

\par We compute, cf. (\ref{b19}), that for all large
$s$ and a large constant $C>0$ 
\begin{align*}
  &-Cs^{2\delta -\nu -1}-Cs^{2
    \delta -1}q^{s}{\left(x,\xi\right)}
  c_{s,\kappa }{\left(x,\xi
    \right)}
  \\ 
&\leq b^{1}_{s,\kappa }{\left(
x,\xi\right)}\leq Cs^{2\delta -\nu -
1}-C^{-1}s^{2\delta -1}q^{s}{\left(
  x,\xi\right)}c_{s,\kappa 
}{\left(x,\xi\right)},
\end{align*}
from which we conclude that 
\begin{equation}
\label{d15} 
 \limsup  _{t_{0}\rightarrow \infty
}\sup  _{\kappa \geq 1,t\geq t_{0}}\int _{t_{0}
}^{t}{\left\langle \mathbf{D}B_{s,\kappa } \right\rangle }_{s}ds\leq
0.
\end{equation}

\par As for the first term on the right hand side of (\ref{d14}),
obviously for fixed $ t_{0}$
\begin{equation}
\label{d16} 
{\left\langle B_{t_{0},\kappa }\right\rangle }_{%
  t_{0}}\rightarrow 0\;\hbox{\rm for }\kappa \rightarrow \infty .
\end{equation}
\par Combining (\ref{d15}) and (\ref{d16}) we conclude (by
first fixing $ t_{0}$ ) that 
\begin{equation*}
\limsup  _{t\rightarrow \infty 
}\quad {\left\langle B_{t,t^{\epsilon 
}}\right\rangle }_{t}\leq 0,
\end{equation*}
whence we infer (\ref{d12}). 
\end{proof}

\par Next we {``}absorb{''} the $ \epsilon $ of Lemma~\ref{lem:4.4}
into the $ \delta $ and introduce the symbols 
\begin{equation}
\label{D16} 
\begin{split}
&b_{t}{\left(x,\xi\right)}=
F_{
  -}{\left(t^{2\delta }q^{s}{\left( x,\xi\right)}\right)}F_{-}
{\left(t^{\nu }q^{-}{\left(
x,\xi\right)}\right)}\tilde{g}_{
1}{\left(t^{-1}x\right)}\tilde{g}_{2}{\left(\xi\right)},
\\
&b^{1}_{t}{\left(x,\xi\right)} 
=-t^{-1}F^{\prime}_{-}{\left( t^{2\delta }q^{s}{\left(x,\xi\right)}
  \right)}F_{-}{\left(t^{\nu }q^{-}{\left(x,\xi\right)}
  \right)}\tilde{g}_{1}
{\left(t^{-1}x\right)}\tilde{g}_{
  2}{\left(\xi\right)},
\end{split}
\end{equation}
where $ 0<2\delta <\min {\left(\nu ,2\delta ^{ s}\right)}$ with $ \nu
<2/5$ and $ \delta ^{s}$ as in (\ref{b19}). Clearly 
\begin{equation*}
b_{t}{\left(x,\xi\right)}
\in S_{unif}{\left(1,g^{1-\nu ^{\prime
},\nu ^{\prime}}_{t}\right)}
\subseteq S_{unif}{\left(1,g^{1-\nu ,\nu 
}_{t}\right)};\quad \nu ^{%
\prime}=\nu -\delta .
\end{equation*}

\par We have the following integral estimate.\par 

\begin{lemma}\label{lem:4.5}
  In the state $ \psi _{1}{\left(t\right)}=L_{ 1}{\left(t\right)}\psi
  {\left( t\right)}$
  \begin{equation}
    \label{d17} 
 \int _{1}^{\infty }|{\left\langle
{\left(b^{1}_{t}\right)}^{
w}{\left(x,p\right)}\right\rangle }_{
t}|dt<\infty .
\end{equation}
\end{lemma}
\begin{proof}
  We use the proofs of Lemmas~\ref{lem:4.2} and \ref{lem:4.4}. Notice
  that to leading order {``}the derivative{''} of the symbol
  \begin{equation*}
    F_{+}{\left(t^{2\delta }q^{s}
        {\left(x,\xi\right)}\right)}
    F_{-}{\left(t^{\nu }q^{-}{\left(
            x,\xi\right)}\right)}\tilde{g}_{
      1}{\left(t^{-1}x\right)}\tilde{g}_{2}{\left(\xi\right)}
  \end{equation*}
is indeed non-positive, and that $ F^{\prime}_{+}=-F^{\prime}_{-} $.
\end{proof}

\par By combining Lemmas~\ref{lem:4.1} and \ref{lem:4.4} we conclude
the following localization result.

\begin{proposition}\label{prop:4.6}
  For any state $ \psi =f{\left(H\right)}\psi $ obeying \eqref{eq:8}
  and \eqref{eq:9} with $ f\in C^{\infty }_{0}{\left(I_{0} \right)}$
  where $ I_{0}$ is a sufficiently small neighborhood of $ E_{0}$
  \begin{equation}
    \label{d18} 
    ||\psi {\left(t\right)}-b^{
      w}_{t}{\left(x,p\right)}\psi {\left(t\right)}||\rightarrow
    0\text{ for } t\rightarrow \infty .
  \end{equation}
\end{proposition}

\par Using the symbol 
$ b_{t}{\left(x,\xi\right)}$ we can bound powers of $ \gamma $, cf.
(\ref{B15}). If we define $ \gamma =\gamma {\left(x,\xi\right)} $ as
in \eqref{eq:19} upon substituting $E$ by the symbols $
h{\left(r_{0}\hat{x},\xi \right)}$ we may consider the symbol
\begin{equation}
\label{d20} 
\gamma ^{\alpha }_{t}{\left(x,\xi
\right)}:=\gamma ^{\alpha }{\left(
x,\xi\right)}b_{t}{\left(x,\xi
\right)};\quad \alpha \in {\left(
\mathbf{N}\cup {\left\{0\right\}}\right)}^{%
2n-2}.
\end{equation}

\par We have the bounds 
\begin{equation}
\label{d19} 
||{\left(\gamma ^{\alpha }_{t}
\right)}^{w}{\left(x,p\right)}
||=O{\left(t^{-\delta |\alpha |
}\right)}.
\end{equation}

Proposition \ref{prop:4.6} and the accompanying (\ref{d19}) give the
$t^{-\delta}$--localization of step I) of the proof of Theorem
\ref{thm:1.1}. We will also need the integral estimate of 
Lemma \ref{lem:4.5} as well as Remark \ref{rem:4.3} in the proof that $\Gamma$ is well localized in the state $\psi(t)$ (see the proof of Proposition \ref{prop:5.1}).

\section{$\Gamma $ and its localization}
\label{sec:5}

\par With the assumption (\ref{b13}) we define operators
$G$ and $ \Gamma $ as follows: The right hand side
of (\ref{BB21}) with $m=m_0$ is of the form 
\begin{equation*}
\gamma ^{{\left(m_{0}\right)}
}=\gamma _{1}+\sum _{2\leq {\left\lvert 
\alpha \right\rvert}\leq m_{0}}c_{
\alpha }\gamma ^{\alpha },
\end{equation*}
with  $c_{\alpha }$ as well as $\gamma _{1}$  and $ \gamma ^{\alpha }$
depending  smoothly of $E$. As
done in (\ref{d20}) we substitute 
\begin{equation}
\label{e2} 
E=h{\left(r_{0}\hat{x},
\xi\right)}
\end{equation}
and multiply suitably by the factors $ \tilde{\tilde{g}}_{
  1}{\left(t^{-1}x\right)}\text{ and } \tilde{\tilde{g}}_{
  1}{\left(\xi\right)}$ as introduced in Section~\ref{sec:3} (with small
supports).  Precisely we pick $ l\leq n-1$ such that (\ref{b13}) holds
and write 
\begin{equation*}
\gamma _{1}=c_{l}{\left(\xi-
\xi{\left(E_{0}\right)}\right)}
\cdot \omega_{l}{\left(E_{0}\right)}
+r_{E}{\left(x,\xi\right)};\quad 
c_{l}=\partial _{\eta _{l}}\gamma _{%
1}{\left(w,E_{0}\right)}_{%
|w=0}.
\end{equation*}
Then we define the operator $ G=G_{t}=\gamma ^{w}_{t}{\left(
    x,p\right)}$ by the symbol 
\begin{equation}
\label{e1} 
\begin{split}
  &\gamma _{t}{\left(x,\xi\right)} =\gamma ^{1}{\left(x,\xi\right)}
  +\gamma ^{2}_{t}{\left(x,\xi\right)} ;
  \\
  &\gamma ^{1}{\left(x,\xi \right)}={\left(\xi-\xi{\left(
          E_{0}\right)}\right)}\cdot \omega_{ l}{\left(E_{0}\right)},
  \\
  &\gamma ^{2}_{t}{\left(x,\xi\right)} ={\left(c_{l}\right)}^{-1}
  {\left(r_{E}{\left(x,\xi\right)} +\sum _{2\leq {\left\lvert \alpha
          \right\rvert} \leq m_{0}}c_{\alpha }\gamma ^{ \alpha
      }{\left(x,\xi\right)}\right)}
  \tilde{\tilde{g}}_{1}{\left(t^{-1}x\right)}\tilde{\tilde{g}}_{2}
  {\left(\xi\right)}.
\end{split}
\end{equation}
For the second term the substitution (\ref{e2}) is used. Let $ \Gamma
=\Gamma _{t}=\re {\left(G\right)} $.\par

\par Clearly the quantization of this second term 
$ B_{1}{\left(t\right)}={\left(\gamma
    ^{2}_{t}\right)}^{w}{\left(x,p\right)}$ is bounded. \par

\par We shall assume that 
\begin{equation}
\label{E1} 
\delta {\left(m_{0}+1\right)}\geq 
1,
\end{equation}
where $ \delta <2^{-1}\min {\left(\nu , 2\delta _{s}\right)}$ is given
as in Proposition~\ref{prop:4.6}.\par

Our proof that $\Gamma$ is well localized in the state $e^{-itH}\psi$
(see Corollary 5.2) rests on the quantum analog of the differential equation (\ref{B15}) and the $t^{-\delta}$--localizations proved in Section 4. In addition we will need integral estimates to bound terms which arise when these ``$t^{-\delta}$--localizations'' are differentiated (see the proof of Proposition \ref{prop:5.1}).

\par We shall use the operator 
$ L_{1}{\left(t\right)}$ given in (\ref{d4}). Let us introduce the
notation $ L_{2}{\left(t\right)}=b^{w}_{ t}{\left(x,p\right)}$ for the
quantization of the first symbol of (\ref{D16}). Let us also introduce
the {``}bigger{''} localization operator 
\begin{align*}
&L_{3}{\left(t\right)}
={\left(\tilde{b}_{t}\right)}^{
w}{\left(x,p\right)};
\\
&\tilde{b}_{t}
{\left(x,\xi\right)}=F_{-}{\left(
2^{-1}t^{2\delta }q^{s}{\left(
x,\xi\right)}\right)}F_{-}
{\left(2^{-1}t^{\nu }q^{-}
{\left(x,\xi\right)}\right)}
\tilde{g}_{1}{\left(t^{
-1}x\right)}\tilde{g}_{
2}{\left(\xi\right)}. 
\end{align*}
Notice that also 
\begin{equation*}
\tilde{b}_{t}{\left(
x,\xi\right)}\in S_{unif}{\left(
1,g^{1-\nu ^{\prime},\nu ^{\prime}
}_{t}\right)};\quad \nu ^{\prime}=\nu -\delta ,
\end{equation*}
 and that indeed for example 
 \begin{equation}
\label{jjunk} 
{\left(I-L_{3}{\left(t\right)}
\right)}L_{2}{\left(t\right)}
L_{1}{\left(t\right)}=O{\left(
t^{-\infty }\right)}.
\end{equation}

\par We obtain from (\ref{B15}), (\ref{E1}) and bounds like
(\ref{d19}) that 
\begin{equation}
\label{e4} 
L_{3}i{\left[H,G\right]}L_{3}
=-L_{3}\tilde{t}^{-1}GL_{%
3}+O{\left(t^{-2}\right)},
\end{equation}
where $t$ is omitted in the notation and 
$\tilde{t}^{-1}$
 is the Weyl quantization of the symbol 
 \begin{equation*}
-{{\partial _{\mu }h{\left(x,\xi
\right)}}\over{x\cdot \omega{\left(h{\left(
x,\xi\right)}\right)}}}
\beta ^{s}_{1}{\left(h{\left(
x,\xi\right)}\right)}\tilde{\tilde{g}}_{1}{\left(
t^{-1}x\right)}\tilde{\tilde{g}}_{2}{\left(\xi\right)}.
\end{equation*}

\par We may assume that the supports of 
$
\tilde{\tilde{g}}_{1}$  and $\tilde{\tilde{g}}_{1}$
 are so small that 
 \begin{equation}
\label{E4} 
\re {\left(\tilde{t}^{%
-1}\right)}\geq t^{-1}\re {\left(
\tilde{\tilde{g}}_{%
1}{\left(t^{-1}x\right)}\tilde{\tilde{g}}_{2}
{\left(p\right)}\right)}+O{\left(
t^{-2}\right)}.
\end{equation}

Next introduce $ P=P_{t}=GG^{*}+G^{*}G$ where $
  G=G_{t}$ is given as above. 
Using the calculus we compute
(with some patience) 
\begin{align*}
  &L_{3}i{\left[H,P\right]}
  L_{3}=2\re {\left(L_{
        3}i{\left[H,G\right]}L_{3}G^{
        *}+G^{*}L_{3}i{\left[H,G\right]}
      L_{3}\right)}
  \\ 
  &\qquad+\re {\left(L_{3}i{\left[H,G\right]}{\left[G^{*},L_{3}
        \right]}-{\left[G^{*},L_{3}
        \right]}i{\left[H,G\right]}L_{
        3}\right)}
  \\ 
  &\qquad+\re {\left(L_{3}i{\left[H,G^{*}\right]}{\left[G,L_{3}
        \right]}-{\left[G,L_{3}\right]}
      i{\left[H,G^{*}\right]}L_{3}
    \right)}
  \\ 
  &=2\re {\left(L_{3}
      i{\left[H,G\right]}L_{3}G^{*}
      +G^{*}L_{3}i{\left[H,G\right]}
      L_{3}\right)}+c^{w}_{t}{\left(
      x,p\right)}+O{\big(t^{2\nu ^{\prime
        }-3}\big)},
\end{align*}
where \
\begin{equation*}
  c_{t}{\left(x,\xi\right)}
  =c_{t}=2\re {\big (\big\{\tilde{b}_{t},{\big\{\tilde{b}_{
                t},\overline{{\gamma }_{t}
              }\big\}}\big\}}{\left\{
          h,{\gamma }_{t}\right\}
    \big )}\in S_{unif}{\big(t^{
        3\nu ^{\prime}-3},g^{1-\nu ^{\prime
        },\nu ^{\prime}}_{t}\big)}.
\end{equation*}
Applying (\ref{e4}) to the first two terms on the right hand side and
symmetrizing yields 
  \begin{multline}\label{eee18e}  
  L_{3}i{\left[H,P\right]} L_{3}
  =-L_{3}{\left\{P\re {\left(\tilde{t}^{-1}\right)}
        +h.c.\right\}}L_{3}\\+\re {\left( GO{\left(t^{-2}\right)}+G^{*}
        O{\left(t^{-2}\right)}\right)} +O{\left(t^{2\nu
          ^{\prime}-3}\right)}.
\end{multline}
(Here and henceforth the notation $h.c.$ refers to hermitian conjugate,
 viz. $S+h.c.=S+S^*$.) 
Notice that the contribution from $ c^{w}_{t}{\left(x,p\right)} $
disappears and that we use 
\begin{equation}
\label{EXT} 
P\re {\left(\tilde{t}^{
      -1}\right)}+h.c.=2G\re {\left( \tilde{t}^{-1}\right)}
G^{*}+2G^{*}\re {\left(\tilde{t}^{-1}\right)}G+O{\left( t^{-3}\right)}.
\end{equation}

\par We have the following localization result.\par 

\begin{proposition}\label{prop:5.1}
  Let $ \psi , \nu $ and $ \delta $ be given as in
  Proposition~\ref{prop:4.6} and suppose (\ref{E1}). Then for all $
  \sigma \in {\left(\nu ^{\prime},1-\nu ^{ \prime}\right)}$, $ \nu
  ^{\prime}=\nu -\delta $, and with $ P=P_{t}=GG^{*}+G^{*}G$ where $
  G=G_{t}$ is given as above 
  \begin{equation}
    \label{dd222} 
    \left\|F_{+}{\left(t^{2-2\sigma } P\right)}\psi {\left(t\right)}
    \right\|\rightarrow 0\text{ for }t\rightarrow \infty .
  \end{equation}
\end{proposition}
\begin{proof}
 We shall use the scheme of the proof of
Lemma~\ref{lem:4.4}. Consider with $ \kappa =t^{\epsilon }$ for a
small $ \epsilon >0$ the observable 
\begin{align*}
 & A{\left(t,\kappa \right)}
 =L_{1}{\left(t\right)}^{*}
 F_{+}{\left(B{\left(t\right)}
   \right)}L_{2}{\left(t\right)}^{2}F_{+}{\left(B{\left(t\right)}
   \right)}L_{1}{\left(t\right)}
 ;
 \\
 & B{\left(t\right)}=B{\left(
     t;\kappa \right)}=\bar{G}
 \bar{G}^{*}+\bar{G}^{*}\bar{G},\quad \bar{G} 
 =\bar{G}{\left(
     t;\kappa \right)}=\kappa ^{-1}t^{1-
   \sigma }G_{t}.
\end{align*}

\par As before we first compute the Heisenberg derivative treating
$ \kappa $ as a parameter and split (with $
L_{j}=L_{j}{(t)}$) 
\begin{align*}
  &\mathbf{D}A{\left(t,\kappa \right)} =T_{1}{\left(t,\kappa
    \right)}+T_{ 2}{\left(t,\kappa \right)}+T_{3} {\left(t,\kappa
    \right)};
  \\
  &T_{1}=L^{*}_{1}F_{+}{\left( B{\left(t\right)}\right)}L^{
    2}_{2}{\left(\mathbf{D}F_{+}{\left(
          B{\left(t\right)}\right)}\right)} L_{1}+h.c.,
  \\
  &T_{2}=L^{*}_{1}F_{+}{\left( B{\left(t\right)}\right)}{\left(
      \mathbf{D}L^{2}_{2}\right)} F_{+}{\left(B{\left(t\right)}
    \right)}L_{1},
  \\
  &T_{3}=L^{*}_{1}F_{+}{\left( B{\left(t\right)}\right)}L^{
    2}_{2}F_{+}{\left(B{\left( t\right)}\right)}\mathbf{D}L_{1} +h.c.
\end{align*}
\par 

\par The analog of (\ref{d14}) is
\begin{equation}
  \label{e19} 
  \left\langle A{\left(t,\kappa \right)}
  \right\rangle _{t}={\left\langle A{\left(
        t_{0},\kappa \right)}\right\rangle }_{%
  t_{0}}+\int _{t_{0}}^{t }{\left\langle T_{1}{\left(s,\kappa
      \right)}+T_{2}{\left(s,\kappa \right)}
    +T_{3}{\left(s,\kappa \right)}\right\rangle }_{%
  s}ds.
\end{equation}

\par We shall prove that 
\begin{equation}
\label{e20} 
\limsup  _{t_{0}\rightarrow \infty
}\sup  _{t\geq
t_{0}}\int _{t_{0}}^{
t}{\left\langle T_{i}{\left(s,\kappa \right)}\right\rangle }_{s}ds\leq
0 ;\quad i=1,2,3.
\end{equation}

\par To do this we may replace 
$ \mathbf{D}$ by the modified Heisenberg derivative 
\begin{equation*}
\mathbf{D}_{3}={{d}\over{dt}}+i{\left[
\bar{H},\cdot \right]};\quad 
\bar{H}=L_{3}HL_{3},\quad 
L_{3}=L_{3}{\left(t\right)},
\end{equation*}
cf. (\ref{jjunk}) and arguments below for (\ref{YYYY1}).

\par With this modification we first look at the most interesting bound
(\ref{e20}) with $i=1$. We use (\ref{86a}) to write
\begin{align}\label{BBbb} 
  &  \mathbf{D}_{3}F_{+}{\left( B{\left(t\right)}\right)}=F^{
    \prime{{1}\over{2}}}_{+}{\left( B{\left(t\right)}\right)}D{\left(
      t\right)}F^{\prime{{1}\over{2}} }_{+}{\left(B{\left(t\right)}
    \right)}+R_{1}{\left(t\right)} +R_{2}{\left(t\right)}+R_{3}
  {\left(t\right)};\nonumber
  \\
  &D{\left(t\right)} ={{2-2\sigma
    }\over{t}}B{\left(t\right)} -L_{3}{\left\{B{\left(t\right)} \re
      {\left(\tilde{t}^{ -1}\right)}+h.c.\right\}}L_{ 3}.
\end{align}
Notice that here $ R_{3}{\left(t\right)}$ is given by the integral
representation (\ref{83a}) of Lemma~\ref{lem:3.5} in terms of the
bounded operator $ D_{r}{\left(t\right)}=\mathbf{D}_{
  3}B{\left(t\right)}-D{\left(t \right)}$ which by (\ref{eee18e}) is
of the form 
\begin{equation}
\label{YYY} 
\begin{split}
&D_{r}{\left(t\right)}
=\kappa ^{-2}t^{2-2\sigma }{{d}\over{%
    dt}}P+{\left\{\kappa ^{-2}t^{2-2\sigma
    }L_{3}Hi{\left[L_{3},P\right]}
    +h.c.\right\}}
\\
&+\kappa ^{-2}
t^{2-2\sigma }{\left\{\re {\left( GO{\left(t^{-2}\right)}\right)}
+\re {\left(G^{*}O{\left(t^{%
          -2}\right)}\right)}+O{\left( t^{2\nu ^{\prime}-3}\right)}
\right\}}.
\end{split}
\end{equation}

\par  First we examine the contribution from the expectation of the term
\begin{equation*}
\cdots L_{2}{\left(s\right)}^{
2}{\left\{R_{1}{\left(s\right)}
+R_{2}{\left(s\right)}\right\}}
L_{1}{\left(s\right)}+h.c. 
\end{equation*}
of the integrand of (\ref{e20}) (after substituting (\ref{BBbb})). We
may write, omitting here and henceforth the argument $s$,
\begin{equation}
\label{JUNK1} 
\begin{split}
  &i{\left[D,B\right]}=-i{\left[ L_{3}{\left\{B\re
          {\left(\tilde{t}^{-1}\right)}+h.c.\right\}}L_{3},B\right]}
  \\
  &= -{\left(L_{3}{\left\{B\re {\left( \tilde{t}^{-1}\right)}
          +h.c.\right\}}i{\left[L_{3},B\right]}
      +h.c.\right)}\\&-L_{3}{\left\{B\re
{\left(i{\left[\tilde{t}^{%
          -1},B\right]}\right)}+h.c.\right\}} L_{3}.
\end{split}
\end{equation}
Substituted into the representation formulas (\ref{84a}) and
(\ref{85a}) of Lemma \ref{lem:3.5} the first term to the right can be
shown to contribute by terms of the form $ \kappa
^{-2}O{\left(s^{-\infty }\right)} $ (using the factors of $ L_{1}$ and
$ L_{2}$ and commutation), however the bound $ \kappa
^{-1}O{\big(s^{\nu ^{\prime }-1-\sigma }\big)}$ suffices. Here and
henceforth $ O{\left(s^{-\tilde{\epsilon }} \right)}$
refers to a term bounded by $ Cs^{-\tilde{\epsilon }}$
uniformly in $t$ (recall that $ B$ contains a factor $ \kappa
^{-2}=t^{-2\epsilon }$).  To demonstrate this weaker bound we compute
\begin{align*}
  &i{\left[L_{3},B\right]}
  =\kappa ^{-1}s^{1-\sigma }i{\left[L_{
        3},G\right]}\bar{G}^{
    *}+\kappa ^{-1}s^{1-\sigma }\bar{G}^{
    *}i{\left[L_{3},G\right]}+h.
  c.,
  \\
  &i{\left[L_{3},G\right]}
=O{\left(s^{\nu ^{\prime}-1}\right)}.
\end{align*}
Since the middle factor $ \re {\left(\tilde{t}^{
      -1}\right)}=O{\left(s^{-1}\right)} $ we get the bound $$ \kappa
^{-1}O{\left(s^{1-\sigma }\right)} O{\big(s^{\nu ^{\prime}-2}\big)}
=\kappa ^{-1}O{\big(s^{\nu ^{\prime }-1-\sigma }\big)}.$$ We used that
$ \bar{G},\quad \bar {G}^{ *}$ and $B$ may be considered as bounded in
combination with the resolvents of $B$; explicitly we exploited the
uniform bounds (after commutation) 
\begin{equation}
\label{YYYY} 
\begin{split}
  &||\bar{G}{\left(
      B-z\right)}^{-1}||,\quad ||\bar{G}^{*}{\left(
      B-z\right)}^{-1}||\leq C{{%
      {\left\langle z\right\rangle }^{1/2}}\over{| \im z|}},
  \\
  &||{\left( B-z\right)}^{-1}||\leq
  C|\im z|^{-1},\quad ||
  B{\left(B-z\right)}^{-1}|| \leq {{%
      2{\left\langle z\right\rangle }}\over{| \im z|}}.
\end{split}
\end{equation}

\par Similarly, since
\begin{equation}
\label{NNN} 
\re {\left(i{\left[\tilde{t}^{
-1},B\right]}\right)}=\kappa ^{
-1}O{\left(s^{-1-\sigma }\right)} \bar{G}^{*}+\kappa ^{-1}
O{\left(s^{-1-\sigma }\right)}\bar{G}+h.c.
\end{equation}
the second term to the right in (\ref{JUNK1}) contributes by a term of
the form $ \kappa ^{-1}O{\left(s^{-1-\sigma }\right)}$.

\par Using the representation for 
$ R_{3}=R_{3}{\left(s\right)} $ and commutation we claim the bound
\begin{multline}
\label{YYYY1} 
\cdots L^{2}_{2}R_{3}L_{%
  1}+h.c.\\=\kappa ^{-1}O{\left(
    s^{-1}\right)}+\kappa ^{-1}O{\left( s^{-1-\sigma }\right)}+\kappa
^{-2} O{\left(s^{2\nu ^{\prime}-1-2\sigma } \right)}.
\end{multline}
The contributions from the first two terms of (\ref{YYY})
are $ \kappa ^{-2}O{\left(s^{-\infty }\right)} $ and therefore in
particular $ \kappa ^{-1}O{\left(s^{-1}\right)} $. Let us elaborate on
this weaker bound for the first term: Write 
\begin{equation*}
\kappa ^{-2}s^{2-2\sigma }{{d}
\over{ds}}P=\kappa ^{-1}s^{1-\sigma }
{\left\{\bar{G}{{d}\over{ds}}G^{*}+\bar{G}^{
*}{{d}\over{ds}}G+h.c.\right\}},
\end{equation*}
and compute the time-derivative of the symbol $
\tilde{\tilde{g}}_{%
  1}{\left(s^{-1}x\right)}$ that defines the time-dependence of the
symbol of $G$ 
\begin{equation*}
{{d}\over{ds}}\tilde{\tilde{g}}_{1}{\left(s^{-1}
x\right)}=-s^{-2}x\cdot {\left(\nabla 
\tilde{\tilde{g}}_{
1}\right)}{\left(s^{-1}x\right)}.
\end{equation*}
The contribution from this expression is treated by using the factor $
g_{1}{\left(s^{-1}x\right)} $ of $ L_{1}$. First we may insert the
$j${'}th power of $ F=\tilde{g}_{1}{\left(s^{ -1}x\right)}$ next to a
factor $ L_{1}$. Then we place one factor of $ F$ next to any of the
factors of the time-derivative of $G$ by commuting through the
resolvent of $B$, and repeat successively this procedure for the
{``}errors{''} given in terms of intermediary commutators. At each
step a factor of $ \kappa ^{-1}s^{\nu ^{\prime}-\sigma
}=O{\left(s^{\nu ^{\prime}-\sigma }\right)}$ will be gained. (In fact
for the first term of (\ref{YYY}) treated here we have the stronger
estimate $ O{\big(s^{-\sigma }\big)}$.) This means that if we put $
\sigma ^{\prime}=\sigma -\nu ^{\prime} $ then $ h=s^{-\sigma
  ^{\prime}}$ will be an {``}effective Planck constant{''}.  Notice
that 
\begin{align*}
  &  i\left[{\left(B-z\right)}^{ -1},F\right]
\\
&=\kappa ^{-1}s^{1-\sigma }
{\left(B-z\right)}^{-1}{\left\{
    \bar{G}O{\left(s^{\nu ^{
            \prime}-1}\right)}+\bar{G}^{
      *}O{\left(s^{\nu ^{\prime}-1}
      \right)}+h.c.\right\}}{\left(
    B-z\right)}^{-1}.
\end{align*}
Repeated commutation through such an expression by factors of $ F$
provides eventually the power $ h^{j}=s^{-\sigma ^{\prime}j}$. Again
the finite numbers of factors like $ \bar{G}{\left(B-z\right)}^{ -1}$
and $ \bar{G}^{*}{\left(B-z\right)}^{ -1}$ may be estimated by
(\ref{YYYY}) before integrating with respect to the $z${--}variable.
We choose $j$ so large that $ \sigma ^{\prime}{\left(j+1\right)} \geq
1$.\par

\par The contribution to (\ref{YYYY1}) from the second term of
(\ref{YYY}) may be treated very similarly.\par

\par Clearly the last term of (\ref{YYY}) contributes by terms of
the form of the last two terms to the right in (\ref{YYYY1}).\par

\par Next we move the factors of 
$ L_{2}$ next to those of $ L_{1}$ (and other commutation) for the
contribution to (\ref{e20}) from the first term to the right in
(\ref{BBbb}) yielding, as a conclusion, that 
\begin{equation}
\label{e22} 
\begin{split}
  &\left\langle T_{1} \left(s,\kappa \right) \right\rangle _{s}\leq
{\left\langle \breve{\psi },D{\left(s\right)} \breve{\psi
    }\right\rangle }+\kappa ^{ -1}O{\left(s^{-1}\right)}+O{\left(
    s^{\nu ^{\prime}-2}\right)} ;
\\
&\breve{\psi }={\left( F^{2\prime}_{+}\right)}^{{{
      1}\over{2}}}{\left(B{\left(s\right)}
  \right)}L_{2}{\left(s\right)} L_{1}{\left(s\right)}\psi {\left(
    s\right)}.
\end{split}
\end{equation}
Notice that commutation of $ D{\left(s\right)}$ with the factors of $
L_{2}{\left(s\right)}$, $ F^{\prime{{1}\over{2}}}_{+}{\left(
    B{\left(s\right)}\right)}$ and ${\left(F^{2\prime}_{+}\right)}^{
  {{1}\over{2}}}{\left(B{\left( s\right)}\right)}$ (when symmetrizing)
involves the calculus of Lemma \ref{lem:3.5} and the effective Planck
constant $ h=s^{-\sigma ^{\prime}}$ in a similar fashion as above.\par

\par For the first term on the right hand side of (\ref{e22}) we
infer from (\ref{E4}) and (\ref{EXT}) that 
\begin{equation}
\label{OO11} 
 {\left\langle \breve{\psi },D{\left( s\right)}\breve{\psi }
   \right\rangle } \leq C_{1}\kappa 
^{-2}s^{-1-2\sigma } +C_{2}s^{-2}.
\end{equation}
By combining (\ref{e22}) and (\ref{OO11}) we finally
conclude (\ref{e20}) for $i=1$.\par

\par As for (\ref{e20}) for $i=2$ we use Remark~\ref{rem:4.3},
the integral estimate of Lemma~\ref{lem:4.5} and the factors of $
L_{1}$. Notice that the leading (classical) term from differentiating
the symbol $ b_{t}$ may be written as a sum of three terms: The
contribution from {``}differentiating{''} the factor $
F_{-}{\left(t^{\nu }q^{-}{\left( x,\xi\right)}\right)}$ is
non-positive, cf. Remark~\ref{rem:4.3}. The contribution from
{``}differentiating{''} the first factor $ F_{-}{\left(t^{2\delta
    }q^{s} {\left(x,\xi\right)}\right)} $ may after a symmetrization
be treated by Lemma~\ref{lem:4.5}.  The commutation through the
factors of $ F_{+}{\left(B{\left(s\right)} \right)}$ (when
symmetrizing) involves the calculus of Lemma~\ref{lem:3.5} in a
similar fashion as above. Finally the contribution from
{``}differentiating{''} the last two factors are integrable due to the
factors of $ L_{1}$. We omit further details.\par

\par As for (\ref{e20}) for $i=3$ we use the integral
estimate \eqref{eq:9} and commutation. We omit the details.\par

\par We conclude (\ref{e20}), and therefore by
Proposition~\ref{prop:4.6} the bound (\ref{dd222}) first with $ \sigma
$ replaced by $ \sigma +\epsilon $ and then (since $ \epsilon $ is
arbitrary) by any $ \sigma $ as specified in the proposition.
\end{proof}

\begin{corollary}\label{cor:5.2}
  Under the conditions of Proposition~\ref{prop:5.1} and with $
\Gamma =\Gamma _{t}=\re {\left(G\right)} $ 
\begin{equation}
\label{dd2222} 
\big\|F_{+}{\left(t^{1-\sigma } |\Gamma |\right)}\psi {\left(
    t\right)}\big\|\rightarrow 0\text{ for }t\rightarrow \infty .
\end{equation}
\end{corollary}
\begin{proof}
  Let $ \sigma \in {\left(2\nu ^{\prime},1\right)} $ be given. Fix $
  \sigma _{1}\in {\left(2\nu ^{\prime },\sigma \right)}$. By
  Proposition~\ref{prop:5.1} it suffices to show that 
  \begin{equation*}
    \big\|F_{+}{\left(t^{1-\sigma }
|\Gamma |\right)}F_{-}{\left(
t^{2-2\sigma _{1}}P\right)}\big\|=O{\left(t^{\sigma _{1}-\sigma }
\right)}.
\end{equation*}

\par Clearly by the spectral theorem this estimate follows from 
\begin{equation*}
\big\|t^{1-\sigma _{1}}\Gamma F_{%
-}{\left(t^{2-2\sigma _{1}}P\right)}
\big\|\leq 1,
\end{equation*}
which in turn follows from substituting $ \Gamma
=2^{-1}{\left(G+G^{*}\right)} $ and then estimating 
\begin{align*}
  &\big\|t^{1-\sigma _{1}}\Gamma F_{ -}{\left(t^{2-2\sigma
        _{1}}P\right)} \big\|
  \\
  &\leq 2^{-1}\big\|t^{1-\sigma _{ 1}}GF_{-}{\left(\cdot \right)}
  \big\|+2^{-1}\big\|t^{1-\sigma _{1 }}G^{*}F_{-}{\left(\cdot \right)}
  \big\|
  \\
  &\leq 2^{-1}\big\|t^{2-2\sigma _{ 1}}F_{-}{\left(\cdot \right)}
  G^{*}GF_{-}{\left(\cdot \right)} \big\|^{1/2}+2^{-1}\big\|t^{2-
    2\sigma _{1}}F_{-}{\left(\cdot \right)} GG^{*}F_{-}{\left(\cdot
    \right)} \big\|^{1/2}
  \\
  &\leq \big\|F_{-}{\left( \cdot \right)}t^{2-2\sigma _{1}}
  PF_{-}{\left(\cdot \right)}\big\|^{ 1/2}\leq 1.
\end{align*}
\end{proof}

\begin{remark}\label{rem:5.3}
In the case of (\ref{bb13}) we define $ \Gamma $ as follows: We pick
$ l\leq n-1$ such that (\ref{bb13}) holds and write 
\begin{align*}
  &\gamma _{1}=c_{l}{{x }\over{\tilde{x}_{n}}} \cdot
  \omega_{l}{\left(E_{0}\right)} +r_{t,E}{\left(x,\xi\right)} ;
  \\
  &c_{l}=\partial _{u_{l}} \gamma _{1}{\left(w,E_{0}\right)}_{
    |w=0},\quad \tilde{x}_{ n}=tk{\left(E_{0}\right)}.
\end{align*}
The operator $ G=G_{t}=\gamma ^{w}_{t}{\left( x,p\right)}$ is given by
the symbol (using the substitution (\ref{e2}))
\begin{align}\label{EE1} 
  &\gamma _{t}{\left(x,\xi\right)} =\gamma
  ^{1}_{t}{\left(x,\xi\right)} +\gamma ^{2}_{t}{\left(x,\xi\right)} ;
  \\
  &\gamma ^{1}_{t}{\left( x,\xi\right)}=t^{-1}x\cdot \omega_{
    l}{\left(E_{0}\right)},\nonumber
  \\
  &\gamma ^{ 2}_{t}{\left(x,\xi\right)} ={{k{\left(E_{0}\right)}}
    \over{c_{l}}}{\bigg(r_{t,E}{\left( x,\xi\right)}+\sum _{2\leq
        {\left\lvert
            \alpha \right\rvert}\leq m_{0}}c_{%
        \alpha }\gamma ^{\alpha }{\left(x,\xi
        \right)}\bigg)}\tilde{\tilde{g}}_{1}{\left(t^{-1}
      x\right)}\tilde{\tilde{g}}_{ 2}{\left(\xi\right)},\nonumber
  \end{align}
cf. (\ref{e1}). One proves Proposition 5.1 with this $G$
in the same way as before. Let $ \Gamma =\re {\left(G\right)}$. We
have (\ref{dd2222}) for this $ \Gamma $.
\end{remark}

\section{Mourre theory for $\Gamma $}
\label{sec:6}
The goal of this section is to show that $\Gamma$ (modified by a
constant)  and a certain conjugate
operator which we introduce below satisfy a version of the uncertainty
principle.  We accomplish this using Mourre theory. The abstract
version of  the uncertainty
principle we shall  need is the following.

\begin{lemma}\label{lem:6.10} Suppose $\bar{H}$ and $\bar{A}$ are two
  self-adjoint operators on the same Hilbert space such that
\begin{enumerate}
   \item \label{item:96} $\mathcal{D}(\bar{H})\cap \mathcal
   D(\bar{A})$ is dense in  $\mathcal D(\bar{H})$.
\item \label{item:97} $\sup_{|s|<    1}\|\bar{H}e^{is\bar{A}}\psi\|<\infty$ for all
  $\psi\in\mathcal{D}(\bar{H})$. 
\item \label{item:98} The form $i[\bar{H},\bar{A}]$ extends to an $\bar{H}$-bounded
  operator satisfying $$i[\bar{H},\bar{A}]\geq c_1>0.$$
\item \label{item:99} The form $i{\left[i{\left[\bar{H},\bar{A}\right]},\bar{A}
    \right]}$ extends to a bounded
  operator $B$ satisfying $$\left\|B\right\|\leq C_1<\infty.$$
\end{enumerate}

Then there exists $C_2=C(c_1,C_1)>0$ such that for all 
$ h\in C^{\infty
  }_{0}{\left(\mathbf{R}\right)} $ (with 
$ {\left\langle \bar{A}\right\rangle }
:={\left(1+\bar{A}^{2}\right)}^{1/2}$ )
  \begin{equation}
    \label{f1} 
    \big\|{\left\langle \bar{A}\right\rangle }^{ 
      -1}h{\left(\bar{H}\right)}
    {\left\langle \bar{A}\right\rangle }^{
      -1}\big\|\leq C_2\|h\|_{L^{
  1}}.
\end{equation} 

In particular, for all $ h_{1},h_{2}\in C^{\infty }_{0}
  {\left(\mathbf{R}\right)}$, 
   $ \delta_2>\delta_1\geq 0 $ and $t\geq 1$
\begin{align}
    \label{f88} 
     \left\|h_{1}{\left(t^{-\delta_1}\bar{A}
\right)}h_{2}{\left(t^{\delta_2}\bar{H} \right)}\right\|\leq
&C_3t^{(\delta_1-\delta_2)/2 };\\&C_3=C_2\|h_2\|_{L^{
  2}}\sup|\langle x\rangle h_1(x)|.\nonumber
\end{align}

\end{lemma}
\begin{proof} We readily obtain by
keeping track of constants in the method of [M] that for some positive
 constant $C$ depending only on $c_1$ and $C_1$ 
\begin{equation}
\label{f7} 
\big\|{\left\langle \bar{A}\right\rangle }^{
-1}{\left(\bar{H}-z\right)}^{
-1}{\left\langle \bar{A}\right\rangle }^{
-1}\big\|\leq C;\quad \im z\not=0.
\end{equation}

\par  Representing 
$ h{\left(\bar{H}\right)} =\pi ^{-1}\lim _{ \epsilon
  \downarrow 0}\int h{\left(\lambda
  \right)}\im {\big({\big(
\bar{H}-\lambda -i\epsilon \big)}^{
-1}\big)}d\lambda $ and then using (\ref{f7}) we conclude
(\ref{f1}).

As for \eqref{f88} we use  (\ref{f1}) with $\bar{A}\to
t^{-\delta_1}\bar{A}$ and $\bar{H}\to t^{\delta_1}\bar{H}$, and with $h(x)=|h_2(t^{\delta_2-\delta_1}x)|^2$. Notice
that (\ref{item:98}) and
  (\ref{item:99}) hold with the same constants for this
  replacement.
\end{proof}

To apply Lemma \ref{lem:6.10} we shall need a specific construction of
$\Gamma$ given in terms of a
hierarchy of sharp localizations in our observables (see (\ref{Ff1})
and (\ref{Ff2})). We are forced to use such hierarchy due to the energy variation of $(\omega(E),\xi(E))$. 

\par Let 
$ \Gamma $ be as in Section 5 (assuming first (\ref{b13})). The $
m_{0}$ of (\ref{e1}) is here considered as arbitrary (but fixed); the
condition (\ref{E1}) (needed before for dynamical statements) is not
imposed. 

\par We introduce for 
$ 0<\bar{\delta}\leq 1$ the operators
\begin{align}\label{FF11} 
   &\bar{H}=t^{1-\bar{\delta }}\Gamma ,\quad
  \bar{A}=\bar{a}^{w}_{ t}{\left(x,p\right)};
  \\
  &\bar{a}_{t}{\left(x,\xi\right)} =t^{\bar{\delta}-1}{\left( x\cdot
      \omega_{l}{\left(E_{0}\right)} 
+x\cdot \big(
          \omega_{l}{\left(h{\left( x,\xi\right)}\right)}-
          \omega_{ l}{\left(E_{0}\right)} \big)\tilde{\tilde{g}}_{
        1}{\left(t^{-1}x\right)}\tilde{\tilde{g}}_{2}
      {\left(\xi\right)}\right)}.\nonumber
\end{align}

\par We shall need a specific construction of the functions
$ \tilde{\tilde{g}}_{ 1}$ and $ \tilde{\tilde{g}}_{ 2}$ in the
definitions \eqref{FF11} in terms of a
small parameter $ \epsilon >0$:

\par The factor 
$ \tilde{\tilde{g}}_{1}{\left(t^{-1}x\right)}$ is the product of the
$n$ functions 
\begin{equation}
\label{Ff1} 
\begin{split}
  &F_{-}{\left(\epsilon ^{-3}
      |t^{-1}x\cdot \omega_{j}{\left(E_{ 0}\right)}|\right)};\quad
  j=1,\dots ,n-1,
  \\
  &F_{-}{\left(\epsilon ^{-2} |t^{-1}x\cdot
      \omega_{n}{\left(E_{ 0}\right)}-k{\left(E_{0}\right)}
      |\right)}.
\end{split}
\end{equation}

\par The factor 
$ \tilde{\tilde{g}}_{ 2}{\left(\xi\right)}$ is the product of the $n$
functions 
\begin{equation}
\label{Ff2} 
\begin{split}
 & F_{-}{\left(\epsilon ^{-2}
      |{\left(\xi-\xi{\left(E_{0} \right)}\right)}\cdot \omega_{l
      }{\left(E_{0}\right)}|\right)},
  \\
  &F_{-}{\left(\epsilon
      ^{-3} |{\left(\xi-\xi{\left(E_{0} \right)}\right)}\cdot
      \omega_{j }{\left(E_{0}\right)}|\right)}
  ;\quad j=1,\dots ,n-1,\quad j\not=l,
  \\
  &F_{
    -}{\left(\epsilon ^{-4}|{\left( \xi-\xi{\left(E_{0}\right)}
        \right)}\cdot \omega_{n}{\left(E_{
            0}\right)}|\right)}.
\end{split}
\end{equation}

\par Now, indeed for we may apply Lemma \ref{lem:6.10} to the example
 introduced by 
 \eqref{FF11}.

\begin{lemma}\label{lem:6.1}
  There exists $ \epsilon _{0}>0$ such that for all positive $
  \epsilon \leq \epsilon _{0}$ there exists $ t_{0}\geq 1$ such
  that for all $ t\geq t_{0}$ the conditions of Lemma \ref{lem:6.10}
  are fulfilled for $\bar{H}=\bar{H}_{t,\epsilon}$ and
  $\bar{A}=\bar{A}_{t,\epsilon}$ with constants independent of $ t\geq t_{0}$.
\end{lemma}
\begin{proof}
\par We shall  verify Lemma \ref{lem:6.10}  (\ref{item:98}) and
  (\ref{item:99}) only (Lemma \ref{lem:6.10}  (\ref{item:96}) and
  (\ref{item:97}) follow readily from the calculus of pseudodifferential operators).  As for (\ref{item:98}) we claim that for all small enough $ \epsilon $
\begin{equation}
\label{f5} 
 i{\left[\bar{H},\bar{A} \right]}\geq 2^{-1};\quad t\geq
t_{0}=t_{0}{\left(\epsilon \right)}.
\end{equation}
\par To see this we notice that clearly the first term in (\ref{e1})
and the first term of the symbol  
$ \bar{a}$ contribute by 
\begin{equation*}
i{\left[t^{1-\bar{\delta}
}{\left(\gamma ^{1}\right)}^{%
w}{\left(x,p\right)},t^{\bar{\delta}-1}x\cdot \omega_{l}{\left(
E_{0}\right)}\right]}=1,
\end{equation*}
so it remains to estimate 
\begin{equation}
\label{F5} 
\left\|i{\left[t^{1-\bar{\delta} }{\left(\re {\left(\gamma ^{2
}_{t}\right)}\right)}^{
w}{\left(x,p\right)},\bar{A} \right]}\right\|\leq 4^{-1};\quad t\geq
t_{0}, 
\end{equation}
and 
\begin{equation}
\label{FF5} 
 \left\|i{\left[t^{1-\bar{\delta}
}{\left(\gamma ^{1}\right)}^{ 
  w}{\left(x,p\right)},\bar{A}
-t^{\bar{\delta}-1}x\cdot \omega_{
  l}{\left(E_{0}\right)}\right]} \right\|\leq 4^{-1};\quad t\geq
t_{0}. 
\end{equation}

\par Let us denote by 
$ a_{t}{\left(x,\xi\right)}$ the Weyl symbol of the operator in
(\ref{F5}) or the one in (\ref{FF5}). We have in both cases that $
a_{t}\in S_{unif}{\left(1,g^{1,0}_{ t}\right)}$, so it suffices to
show (cf. [H\"o, Theorem 18.6.3] and the proof of [DG, Proposition
D.5.1]) that 
\begin{equation}
\label{F6} 
\sup  _{x,\xi\in \mathbf{R}^{
    n},t\geq t_{0}}|a_{t}{\left( x,\xi\right)}|\leq \nu _{0} ,
\end{equation}
where $ \nu _{0}$ is a (universal) small positive
constant associated for example to the $ L^{2}$ {--}boundedness result
[H\"o, Theorem 18.6.3].
\par 

\par For (\ref{F6}) we note the uniform bounds
\begin{align*}
  &h{\left(x,\xi\right)} -E_{0}=O{\left(\epsilon ^{4}\right)},
  \\
  &t\partial _{x_{j}}h{\left( x,\xi\right)}=O{\left(\epsilon ^{2
      }\right)},
  \\
  & \partial _{\xi_{j}}h{\left( x,\xi\right)}=O{\left(\epsilon ^{2
      }\right)}\text{ for }j\leq n-1,\quad \partial
  _{\xi_{n}}h{\left(x,\xi \right)}=O{\left(\epsilon ^{0}\right)},
  \\
  &\gamma _{j}{\left(x,\xi
    \right)}=O{\left(\epsilon ^{2}\right)},\quad t\partial _{x}
  \gamma _{j}{\left(x,\xi\right)}=O{\left(      \epsilon
  ^{0}\right)},\quad \partial _{ 
    \xi}\gamma _{j}{\left(
      x,\xi\right)}=O{\left(\epsilon ^{0 }\right)},
\end{align*}
on the support of the function $ \tilde{\tilde{g}}_{
  1}{\left(t^{-1}x\right)}\tilde{\tilde{g}}_{2} {\left(\xi\right)}$
given by (\ref{Ff1}) and (\ref{Ff2}).  Here we used \eqref{eq:3} and
\eqref{eq:4}, and the notation 
\begin{equation*}
x_{j}=x\cdot \omega_{j}{\left(E_{0}\right)},\quad \xi_{j}
=\xi\cdot \omega_{j}{\left(E_{0}
\right)}.
\end{equation*}

\par By estimating the leading term of the symbol using these bounds we
may show (with some patience) that 
\begin{equation}
\label{Ff3} 
\sup _{x,\xi\in \mathbf{R}^{n},t\geq t_{0}}|a_{t}{\left(
    x,\xi\right)}|\leq C\epsilon , 
\end{equation}
from which (\ref{F6}) and (therefore) (\ref{f5})
follow.\par

\par As for (\ref{item:99}) we  have the bound 
\begin{equation}
\label{f6} 
\left\|i{\left[i{\left[\bar{H},\bar{A}\right]},\bar{A}
    \right]}\right\|=O{\left( t^{\bar{\delta}-1}\right)} 
=O{\left(1\right)}.
\end{equation}
\end{proof}

As an immediate  consequence of Lemmas \ref{lem:6.10}  and
\ref{lem:6.1} we have.

\begin{corollary}\label{cor:6.2}
  Suppose $ h_{1},h_{2}\in C^{\infty }_{0} {\left(\mathbf{R}\right)}$
  and $ 0\leq \sigma <\bar{\delta}\leq 1$. Then there exists $
  \epsilon _{0}>0$ such that for all positive $ \epsilon \leq \epsilon
  _{0}$ there exists $C>0$ such that for all $ t\geq 1$
  \begin{equation}
    \label{f8} 
     \left\|h_{1}{\left(\bar{A}
\right)}h_{2}{\left(t^{\bar{\delta}-\sigma }\bar{H} \right)}\right\|\leq
Ct^{{\left(\sigma 
    -\bar{\delta}\right)}/2 }.
\end{equation}
\end{corollary}

\begin{remark}\label{rem:6.3}
  In the case of (\ref{bb13}) we introduce (with $ \Gamma $ as in
  Remark 5.3) 
  \begin{equation}
    \label{FF12} 
    \begin{split}
      &\bar{H}=t^{1-\bar{\delta}}\Gamma ,\quad \bar{A}=\bar{a}^{w}_{
        t}{\left(x,p\right)};
      \\
    &\bar{a}_{t}{\left(x,\xi\right)} =t^{\bar{\delta}}{\left(
          {\left(\xi-\xi{\left(E_{0}\right)} \right)}\cdot
          \omega_{l}{\left(E_{ 0}\right)}+b{\left(x,\xi\right)}
          \tilde{\tilde{g}}_{ 1}{\left(t^{-1}x\right)}\tilde{\tilde{g}}_{2}
          {\left(\xi\right)}\right)},
      \\
      &b{\left(x,\xi\right)}={\left(
          \xi-\xi{\left(h{\left(x,\xi\right)} \right)}\right)}\cdot
      \omega_{l }{\left(h{\left(x,\xi\right)}
        \right)}-{\left(\xi-\xi{\left(E_{ 0}\right)}\right)}\cdot
      \omega_{ l}{\left(E_{0}\right)}.
    \end{split}
  \end{equation}
  Here the factor $ \tilde{\tilde{g}}_{ 1}{\left(t^{-1}x\right)}$ is
  the product of the $n$ functions 
  \begin{align*}
    &F_{-}{\left(\epsilon ^{-2} |t^{-1}x\cdot \omega_{l}{\left(E_{
              0}\right)}|\right)},
    \\
    & F_{-}{\left(\epsilon ^{- 3}|t^{-1}x\cdot \omega_{j}{\left(
            E_{0}\right)}|\right)} ;\quad j=1,\cdots ,n-1,\quad
    j\not=l,
    \\
    &F_{-}{\left(\epsilon ^{-2}|t^{ -1}x\cdot \omega_{n}{\left(E_{0}
          \right)}-k{\left(E_{0}\right)} |\right)},
  \end{align*}
  while the factor $ \tilde{\tilde{g}}_{ 2}{\left(\xi\right)}$ is the
  product of 
  \begin{align*}
    &F_{-}{\left(\epsilon ^{-3} |{\left(\xi-\xi{\left(E_{0}
              \right)}\right)}\cdot \omega_{j
        }{\left(E_{0}\right)}|\right)} ;\quad j=1,\cdots ,n-1,
    \\
    &F_{-}{\left(\epsilon ^{-4}|{\left( \xi-\xi{\left(E_{0}\right)}
          \right)}\cdot \omega_{n}{\left(E_{ 0}\right)}|\right)}.
\end{align*}
\end{remark}
\par One verifies (\ref{f8}) under the same conditions as in
Corollary~\ref{cor:6.2} along the same line as before.

\section{Proof of Theorem~\ref{thm:1.1}}
\label{sec:7}

The proof of Theorem~\ref{thm:1.1} is based on
Proposition~\ref{prop:4.6}, and Corollaries~\ref{cor:5.2} and
\ref{cor:6.2} (with the assumption (\ref{b13}));  we show that the
$t^{-\delta}$--localization and the strong localization of
$\Gamma$ are incompatible with the uncertainty
principle as expressed in Corollary \ref{cor:6.2}.

\par We recall the assumptions of Proposition~\ref{prop:4.6}: 
$ 0<2\delta <\min {\left(\nu ,2\delta ^{ s}\right)}$ with $ \nu
<2/5$ and $ \delta ^{s}$ as in (\ref{b19}).\par

\begin{lemma}\label{lem:7.1}
  With $ \bar{A}=\bar{A}_{ t}$ given in terms of any (small) $
  \epsilon >0$ and of $ \bar{\delta}=\delta $ (with $ \delta $ as
  above) by either (\ref{FF11}) (in the case of (\ref{b13})) or
  (\ref{FF12}) (in the case of (\ref{bb13})) 
  \begin{equation}
\label{G1} 
\lim  _{t\rightarrow \infty }
||F_{+}{\left(|\bar{A} |\right)}\psi {\left(t\right)} ||=0,
\end{equation}
where $ \psi =f{\left(H\right)}\psi $ is given as in
Proposition~\ref{prop:4.6} (with the support of $f$ being sufficiently
small possibly depending on $
\epsilon $).
\end{lemma}
\begin{proof}
  We fix $ \delta _{1}$ such that $ 2\delta <2\delta _{1}<\min {\left(
      \nu ,2\delta ^{s}\right)}$. Let $ b_{t}{\left(x,\xi\right)}$ be
  given by (\ref{D16}) in terms of $ \delta _{1}$ and $\nu $.\par

\par By Proposition~\ref{prop:4.6} it suffices to show that 
\begin{equation*}
||F_{+}{\left(|\bar{A}
|\right)}b^{w}_{t}{\left(
x,p\right)}||\rightarrow 0\text{ for }t\rightarrow 
\infty ,
\end{equation*}
and therefore in turn 
\begin{equation*}
||\bar{A}b^{w}_{
t}{\left(x,p\right)}||=O{\left(
t^{\delta -\delta _{1}}\right)}.
\end{equation*}
For the latter bound one easily check that the symbol of $
\bar{A}b^{w}_{t}{\left( x,p\right)}$ belongs to 
\begin{equation*}
S_{unif}{\left(t^{\delta -\delta _{
1}},g^{1-\nu ^{\prime},\nu ^{
\prime}}_{t}\right)};\quad 
\nu ^{\prime}=\nu -\delta _{1}.
\end{equation*}
\end{proof}

\par Now, we first fix 
$ \delta $ as above and conclude from Lemma~\ref{lem:7.1} that
\begin{equation}
\label{g1} 
||\psi {\left(t\right)}-F_{
-}{\left(|\bar{A}|
\right)}\psi {\left(t\right)}
||\rightarrow 0\text{ for }t\rightarrow \infty ,
\end{equation}
where $ \psi =f{\left(H\right)}\psi $ is given as in
Proposition~\ref{prop:4.6}. This holds for $ f\in C^{\infty
}_{0}{\left(I_{0} \right)};$ $I_{0}=I_{0} {\left(\epsilon
  \right)}$.
\par 

\par Next we fix any 
$ \sigma \in {\left(0,\delta \right)}$ in agreement with
Corollary~\ref{cor:5.2} which means that 
\begin{equation}
\label{g2} 
||F_{+}{\left(|t^{1-\sigma 
}\Gamma |\right)}\psi {\left(
t\right)}||\rightarrow 0\text{ for }t\rightarrow 
\infty .
\end{equation}
Here the input of $ \delta $ in Proposition~\ref{prop:5.1} say $
\delta _{1}$ (needed to fix the $ m_{0}$ in the definition of the $
\Gamma $ of Corollary~\ref{cor:5.2}) is different; we need to have $
\sigma >\nu ^{\prime}$, $\nu ^{ \prime}=\nu _{1}-\delta _{1}$, for
which $ \delta _{1}<\delta $ is needed. The construction of this $
\Gamma $ depends on the same $ \epsilon $ as above, cf.
Section~\ref{sec:6}.\par

\par Combining (\ref{g1}) and (\ref{g2}) leads to
\begin{equation}
\label{g3} 
||\psi {\left(t\right)}-F_{
-}{\left(|\bar{A}|
\right)}F_{-}{\left(|t^{1
-\sigma }\Gamma |\right)}\psi {\left(
t\right)}||\rightarrow 0\text{  for }t\rightarrow 
\infty .
\end{equation}

\par By combining Corollary~\ref{cor:6.2} and (\ref{g3}) we conclude (by
finally fixing $ \epsilon >0$ sufficiently small) that
\begin{equation}
\label{g4} 
||\psi {\left(t\right)}|
| \rightarrow  0\text{  for }t\rightarrow \infty ,
\end{equation}
and therefore that 
$
\psi =0$
 proving Theorem~\ref{thm:1.1}.

 \begin{remark}\label{rem:7.2}
   With the assumption (\ref{bb13}) we proceed similarly using
   Remarks~\ref{rem:5.3} and \ref{rem:6.3}, and Lemma~\ref{lem:7.1}.
\end{remark}

\section{Proof of Theorem~\ref{thm:1.2}}
\label{sec:8}

\par We shall here elaborate on the derivation of
   Theorem~\ref{thm:1.2} from our
general result Theorem~\ref{thm:1.1}.\par

\par First we remove the singularity at $x=0$ by defining
\begin{equation*}
h{\left(x,\xi\right)}=2^{-1}
\xi^{2}+\tilde{V}{\left(
x\right)};\quad \tilde{V}
{\left(x\right)}=F_{+}{\left(
|x|\right)}V{\left(\hat{x}
\right)},
\end{equation*}
where (as before) $V$ is a Morse function on $ S^{n-1}$. (See
Remarks~\ref{rem:8.3} for extensions.) In this case clearly the
hypotheses \eqref{H1}{--}\eqref{H3} of Section~\ref{sec:2} are
satisfied, and \eqref{H4} holds for any critical point $ \omega_{l}\in
C_{r}$ and energy $ E>V{\left(\omega_{j}\right)} $ upon putting $
\omega{\left(E\right)}=\omega_{l}$, $\xi{\left(E\right)}
=k{\left(E\right)}\omega_{l}$ and $ k{\left(E\right)}=\sqrt{2{\left(
      E-V{\left(\omega_{l}\right)}\right)} }$.

\par For \eqref{eq:6} we put 
\begin{equation*}
g{\left(u,\eta ,E\right)}=\sqrt{
2{\left(E-V{\left(\omega_{l}\right)}
\right)}} -\sqrt{2E-\eta ^{2
}-2V{\left(\omega_{l}+u\right)}
}, 
\end{equation*}
yielding \eqref{eq:7} with 
\begin{equation*}
A{\left(E\right)}=k{\left(E\right)}^{
-1}
\begin{pmatrix} &V^{{\left(2\right)}}(
\omega_{l})&0\\
&0&I\\ 
\end{pmatrix}. 
\end{equation*}

\par We may choose an orthonormal basis in 
$ {\left\{\omega_{l}\right\}}^{\bot }\subseteq \mathbf{R}^{n}$ for
which $ V^{{\left(2\right)}}{\left( \omega_{l}\right)}$ is diagonal,
say $ V^{{\left(2\right)}}{\left( \omega_{l}\right)}=\diag {\left(
    q_{1},\dots ,q_{n-1}\right)}$. The eigenvalues of $B(E)$ take the
form 
\begin{equation}
\label{AX} 
\begin{split}
  & \beta ^{+}_{j}{\left(E \right)}=-{{1}\over{2}}+{{
      1}\over{2}}\sqrt{1-2q_{j}/{\left(
        E-V{\left(\omega_{l}\right)}\right)} } \text{ or }
  \\
  &\beta ^{-}_{ j}{\left(E\right)}=-{{1}\over{
      2}}-{{1}\over{2}}\sqrt{1-2q_{j
    }/{\left(E-V{\left(\omega_{l}\right)} \right)}}
\end{split}
\end{equation}
say with $ \sqrt{\zeta } :=i\sqrt{-\zeta } $ if $ \zeta <0$.\par

\par Clearly the hypothesis \eqref{H5} is the non-degeneracy condition,
$ q_{j}\not=0$ for all $j$, while hypothesis \eqref{H6} amounts to $
q_{j}<0$ for some $j$, i.e.  $ \omega_{l}$ is a local maximum or a
saddle point of $V$.\par

\par As for \eqref{H7} one easily checks that there exists a smooth basis of
eigenvectors of $ B{\left(E\right)}^{tr}$ for $
E-V{\left(\omega_{l}\right)}\in {\left(0,\infty \right)}\setminus{\left\{
    2q_{1},\dots ,2q_{n-1}\right\}}$.\par

\par Elementary analyticity arguments show that given any
$ m\in {\left\{2,3,\dots \right\}}$ the set of resonances of order $m$
for any of the eigenvalues of $B(E)$ is discrete in $
{\left(V{\left(\omega_{l}\right)},\infty \right)}$.\par

\par In conclusion, the hypotheses \eqref{H1}{--}\eqref{H8} are
satisfied for any local maximum or saddle point $ \omega_{l}$ of a
Morse function $V$ for $E_{0}\in {\left(V{\left(
        \omega_{l}\right)},\infty \right)} \setminus\mathcal{D}$ where $
\mathcal{D}$ is discrete in $ {\left(V{\left(\omega_{l}\right)},\infty
  \right)}$.\par

\par Due to the possible existence of bound states we change the
definition of $ P_{l}$  to be 
\begin{equation*}
P_{l}=s-\lim _{t
\rightarrow \infty }e^{itH}\chi _{l}{\left(
\hat{x}\right)}e^{-itH
}E_{ac}{\left(H\right)}, 
\end{equation*}
where $ E_{ac}{\left(H\right)}$ is the orthogonal projection onto the
absolutely continuous subspace of $H$, see [H] and [ACH, Theorem C.1].
This gives \eqref{eq:12} with the left hand side replaced by $
E_{ac}{\left(H\right)}$. \par

\par Now, to get \eqref{eq:14} it suffices by Theorem~\ref{thm:1.1} to
verify \eqref{eq:13} for any $ E_{0}\in {\left(V{\left(\omega_{
          l}\right)},\infty \right)}$. Invoking the discreteness of
the set of eigenvalues of $H$ on the complement of the set of critical
values of $V$, cf. [ACH, Theorem C.1], one may easily conclude
\eqref{eq:13} from the following statement:\par

\par Consider any open set 
$ I_{0}\subseteq {\left(V{\left(\omega_{ l}\right)},\infty \right)}$
such that $$ I_{0}\cap {\left(\sigma _{pp}{\left( H\right)}\cup
    V{\left(C_{r}\right)} \right)}=\emptyset.$$ Let $ \mathcal{H}_{0}$
be the closure of the subspace of states $ \psi =f{\left(H\right)}\psi
$, $ f\in C^{\infty }_{0}{\left(I_{0} \right)}$, obeying \eqref{eq:8}
and \eqref{eq:9}.  Then for all $ \psi =P_{l}f{\left(H\right)}\psi $
where $ f\in C^{\infty }_{0}{\left(I_{0} \right)}$
\begin{equation}
\label{olew} 
\psi \in \mathcal{H}_{0}.
\end{equation}

\par We shall verify (\ref{olew}) by showing that indeed
$ \psi =P_{l}f{\left(H\right)}\psi $ obeys \eqref{eq:8} and
\eqref{eq:9}. We shall proceed a little more generally than needed in
that we here assume that the $ \mathcal{U}_{0}$ of \eqref{eq:9} is
given by 
\begin{align*}
  &\mathcal{U}_{0}=\mathcal{U}_{\epsilon } =\tilde{\mathcal{C}}_{\epsilon
  } \times \mathbf{R}^{n};
  \\
  &\tilde{\mathcal{C}}_{ \epsilon }={\left\{x\in \mathbf{R}^{n}\setminus
      {\left\{0\right\}}|\hat{x} \in \mathcal{C}_{\epsilon
      }\right\}},\quad \mathcal{C}_{\epsilon }={\left\{\omega\in S^{
        n-1}\mid |\omega-\omega_{l} |<\epsilon \right\}},
\end{align*}
where $ \epsilon >0$ is taken so small that $ \mathcal{C}_{\epsilon
}\cap C_{r}={\left\{ \omega_{l}\right\}}$.
\par 

\par Pick 
$ \tilde{f}\in C^{\infty }_{0 }{\left(I_{0}\right)}$ such that $ 0\leq
\tilde{f}\leq 1$ and $ \tilde{f}=1$ in a neighborhood of $ \supp
{\left(f\right)}$. Let $ r\in C^{\infty }{\left(\mathbf{R}^{n}\right)}
$ be given in terms of any $ F_{+} \in\mathcal{F}_{+}$ by
\begin{equation}
\label{z49} 
r{\left(x\right)}=\int _{0}
^{|x|}F_{+}{\left(s\right)}
ds+\int _{0}^{1}F_{-}{\left(
s\right)}ds.
\end{equation}
(Notice that $ r{\left(x\right)}=|x|$ for $ |x|\geq 1$.) Let
\begin{equation*}
  p_{||}={{1}\over{
      2}}{\left(\nabla r\cdot p+h.c.\right)},\quad
  \tilde{p}_{ 
    ||}=\tilde{f}{\left(
      H\right)}p_{||}\tilde{f}
  {\left(H\right)}.
\end{equation*}

\begin{lemma}\label{lem:8.1}
  Let $ \chi _{l}\in C^{\infty }_{0}{\left(\mathcal{C}_{\epsilon
      }\right)} $ be given with $ 0\leq \chi _{l}\leq 1$ and $ \chi
  _{l}=1$ in a neighborhood of $ \omega_{l}$, and $ \tilde{g}_{2}\in
  C^{\infty }_{0}{\left(\mathbf{R}\right)} $ by $$
  \tilde{g}_{2}{\left(s\right)} =\tilde{f}{\left(2^{-1}s^{
        2}+V{\left(\omega_{l}\right)} \right)}1_{{\left(0,\infty
      \right)} }{\left(s\right)}.$$ Let real-valued $
  g^{-}_{1},g^{+}_{1}\in C^{ \infty }_{0}{\left(\mathbf{R}\right)} $
  be given with 
  \begin{align*}
    &c^{-}_{+}<\tilde{c}_{ -};\quad c^{-}_{+}=\sup {\left(\supp
        {\left(g^{-}_{ 1}\right)}\right)},\quad \tilde{c}_{-}=\inf
    {\left( \supp {\left(\tilde{g}_{ 2}\right)}\right)},
    \\
    & c^{+}_{-}>\tilde{c}_{ +};\quad c^{+}_{-} =\inf
    {\left(\supp {\left(g^{ +}_{1}\right)}\right)},\quad
    \tilde{c}_{+}=\sup {\left(\supp {\left(\tilde{g}_{
              2}\right)}\right)}.
\end{align*}
Let $F_{+}\in \mathcal{F}_{+}$, $F_{-}\in \mathcal{F}_{-}$ and
\begin{equation*}
C>2\sqrt{2{\left(\sup {\left(
\supp {\left(f\right)}\right)}
-\min {\left(V\right)}\right)}
} .
\end{equation*}
Then, in the state $ \psi {\left(t\right)}=e^{-itH}
P_{l}f{\left(H\right)}\psi $
{\allowdisplaybreaks
  \begin{align}
&\label{esa3} 
 \int _{-\infty }^{\infty }{\left\langle r^{-1-\delta }\right\rangle }_{ t}dt<\infty
;\quad \delta >0,
\\
&\label{esa5} 
\int _{-\infty }^{\infty }|{\left\langle p\cdot r^{
      {\left(2\right)}}p\right\rangle }_{ t}|dt<\infty ,
\\
&\label{esa1} 
\int _{-\infty }^{\infty }{\left\langle r |\nabla \tilde{V} |^{2}\right\rangle }_{
  t}dt<\infty ,
\\
&\label{esa111} 
\int _{-\infty}^{\infty }{\left\langle \tilde{\chi }_{ l}r^{-{{1}\over{2}}}{\left( \eta
        ^{2}+u^{2}\right)}r^{-{{1}\over{2}}}\tilde{\chi }_{
      l}\right\rangle }_{t}dt<\infty ;\quad \tilde{\chi }_{l}=\chi _{
  l}{\left(\hat{x}\right)} F_{+}{\left(r\right)},
\\
&\label{esa333} 
\int _{1}^{\infty }-t^{ -1}{\left\langle
    F^{\prime}_{-}{\left( C^{-1}t^{-1}r\right)}\right\rangle }_{
 t}dt<\infty ,
\\
&\label{esa3333} 
\int _{1}^{\infty }t^{ -1}||g{\left(\tilde{p}_{
      ||}\right)}F_{-}{\left( C^{-1}t^{-1}r\right)}\psi {\left(
    t\right)}||^{2}dt<\infty ;\\&\quad g\in
C^{\infty }_{0}{\left( {\left(-\infty ,0\right)}\right)},\quad
\bar{g}=g,\nonumber
  \\
&\label{esa999} 
\int _{1}^{\infty }t^{ -1}||{\left(1-\tilde{g}_{ 2}{\left(\tilde{p}_{|
          |}\right)}\right)}F_{- }{\left(C^{-1}t^{-1}r\right)}
\tilde{\chi}_{l}\psi {\left( t\right)}||^{2}dt<\infty ,
\\
&\label{esa4} 
\int _{1}^{\infty }t^{-1}| |B^{-}{\left(t\right)}\psi
{\left(t\right)}||^{2}d t<\infty ;\quad B^{-}{\left(t\right)}
=g^{-}_{1}{\left(t^{-1}r\right)}
\tilde{g}_{2}{\left(\tilde{p}_{||}\right)},
\\
&\label{aaar} 
\int _{1}^{\infty }t^{-1}| |B^{+}{\left(t\right)}\psi
{\left(t\right)}||^{2}d t<\infty ;\quad B^{+}{\left(t\right)}
=g^{+}_{1}{\left(t^{-1}r\right)}
\tilde{g}_{2}{\left(\tilde{p}_{||}\right)}.
\end{align}
}
\end{lemma}

\begin{proof}
  For (\ref{esa3}), (\ref{esa5}) and (\ref{esa1}) we refer to [H] and
  [ACH, Theorem C.1]. The bound (\ref{esa111}) follows from those
  estimates by Taylor expansion.\par

\par As for (\ref{esa333}) we consider the {``}propagation
observable{''} 
\begin{equation*}
\Phi {\left(t\right)}=f{\left(
H\right)}F_{-}{\left(C^{-1}
t^{-1}r\right)}f{\left(H\right)}.
\end{equation*}
We may bound its Heisenberg derivative as 
\begin{equation*}
\mathbf{D}\Phi {\left(t\right)}\geq -
\epsilon t^{-1}f{\left(H\right)}F^{
\prime}_{-}{\left(C^{-1}t^{-1
}r\right)}f{\left(H\right)}
+O{\left(t^{-2}\right)};\quad 
\epsilon >0.
\end{equation*}

\par As for (\ref{esa3333}) we consider the observable
\begin{equation*}
\Phi {\left(t\right)}=\tilde{f}
{\left(H\right)}g{\left(\tilde{p}_{
||}\right)}t^{-1}rF_{-
}{\left(C^{-1}t^{-1}r\right)}
g{\left(\tilde{p}_{||
}\right)}\tilde{f}{\left(
H\right)}.
\end{equation*}
We write its Heisenberg derivative as 
\begin{align*}
  &\mathbf{D}\Phi {\left(t\right)} =T_{1}+T_{2}+T_{3};
  \\
  &T_{1}=\tilde{f} {\left(H\right)}{\left(\mathbf{D}g{\left(
          \tilde{p}_{||}\right)} \right)}t^{-1}rF_{-}{\left(
      C^{-1}t^{-1}r\right)}g{\left( \tilde{p}_{||}\right)}
  \tilde{f}{\left(H\right)} +h.c.,
  \\
  &T_{2}=2^{-1}\tilde{f} {\left(H\right)}g{\left(\tilde{p}_{
        ||}\right)}t^{-1}r{\left( \mathbf{D} F_{-}{\left(C^{-1}t^{
            -1}r\right)}\right)}g{\left( \tilde{p}_{||}\right)}
  \tilde{f}{\left(H\right)} +h.c.,
  \\
  &T_{3}=2^{-1}\tilde{f} {\left(H\right)}g{\left(\tilde{p}_{
        ||}\right)}{\left(\mathbf{D}{\left(t^{-1}r\right)}\right)}
  F_{-}{\left(C^{-1}t^{-1}r\right)} g{\left(\tilde{p}_{||
      }\right)}\tilde{f}{\left( H\right)}+h.c.,
\end{align*}
and notice the identities 
\begin{equation}
\label{x2} 
\mathbf{D}r=p_{||},\quad
\mathbf{D}p_{||}=p\cdot r^{
    {\left(2\right)}}p+O{\left(r^{ -3}\right)}.
\end{equation}

\par Using (\ref{esa3}), (\ref{esa5}), the second identity
of (\ref{x2}) and (\ref{83a}) we readily obtain after symmetrization
that
\begin{equation}
\label{x33} 
 \int _{1}^{\infty }|{\left\langle T_{1}\right\rangle }_{t}|dt<\infty.
\end{equation}

\par As for the the term 
$ T_{2}$ we use the first identity of (\ref{x2}) and (\ref{esa333}) to
derive 
\begin{equation}
\label{x3333} 
 \int _{1}^{\infty }|{\left\langle T_{2}\right\rangle }_{t}|dt<\infty
.
\end{equation}

\par For the term 
$ T_{3}$ we compute using the first identity of (\ref{x2}) and
(\ref{83a})

\begin{align}\label{x44} 
&  T_{3}=\re {\left(t^{ -1}\tilde{f}{\left(H\right)}
      g{\left(\tilde{p}_{|| }\right)}{\left(p_{||}
          -t^{-1}r\right)}F_{-}{\left( C^{-1}t^{-1}r\right)}g{\left(
          \tilde{p}_{||}\right)} \tilde{f}{\left(H\right)}
    \right)}+O{\left(t^{-2}\right)}\nonumber
  \\
  & \leq -\epsilon
  t^{-1}\tilde{f}{\left( H\right)}g{\left(\tilde{p}_{
        ||}\right)}F_{-}{\left( C^{-1}t^{-1}r\right)}g{\left(
      \tilde{p}_{||}\right)} \tilde{f}{\left(H\right)}
  +O{\left(t^{-2}\right)};\epsilon > 0.
\end{align}

\par We conclude (\ref{esa3333}) from (\ref{x33}), (\ref{x3333}) and
(\ref{x44}).\par
 
\par The bound (\ref{esa999}) follows from elementary energy
bounds, Taylor expansion and the previous estimates. (For this we need
(\ref{esa3333}) to deal with the {``}region{''} where $ p^{2}_{||}$
energetically has the right size, but $ p_{||}<0$.)\par

\par As for (\ref{esa4}) we consider 
\begin{equation*}
\Phi {\left(t\right)}=\tilde{f}
{\left(H\right)}\tilde{g}_{
2}{\left(\tilde{p}_{|
|}\right)}F{\left(t^{-1}
r\right)}\tilde{g}_{2}
{\left(\tilde{p}_{||
}\right)}\tilde{f}{\left(
H\right)};\quad F{\left(s^{\prime
}\right)}=\int _{-\infty }^{
s^{\prime}}g^{-}_{1}{\left(
s\right)}^{2}ds.
\end{equation*}
We write its Heisenberg derivative as 
\begin{align*}
  &\mathbf{D}\Phi {\left(t\right)} =T_{1}+T_{2};
  \\
  &T_{1}=\tilde{f}
  {\left(H\right)}{\left(\mathbf{D}\tilde{g}_{2}{\left(\tilde{p}_{
            ||}\right)}\right)}F {\left(t^{-1}r\right)}\tilde{g}_{
    2}{\left(\tilde{p}_{| |}\right)}\tilde{f} {\left(H\right)}+h.c.,
  \\
  & T_{2}=\tilde{f} {\left(H\right)}\tilde{g}_{ 2}{\left(\tilde{p}_{|
        |}\right)}{\left(\mathbf{D} F{\left(t^{-1}r\right)}\right)}
  \tilde{g}_{2}{\left(\tilde{p}_{||}\right)}
  \tilde{f}{\left(H\right)}.
\end{align*}
\par Using (\ref{esa3}), (\ref{esa5}), the second identity
of (\ref{x2}) and (\ref{83a}) as for (\ref{esa3333}) we obtain that
\begin{equation}
\label{x3} 
 \int _{1}^{\infty }|{\left\langle T_{1}\right\rangle }_{t}|dt<\infty
.
\end{equation}
\par As for the the term 
$ T_{2}$ we compute using the first identity of (\ref{x2}) and
(\ref{83a}) 
\begin{equation}
\label{x4} 
\begin{split}
  & T_{2}=t^{-1}\tilde{f} {\left(H\right)}B^{-}{\left(
      t\right)}^{*}{\left(p_{|| }-t^{-1}r\right)}B^{-}{\left(
      t\right)}\tilde{f}{\left( H\right)}+O{\left(t^{-2}\right)}
  \\
  &\geq t^{-1}B^{-}{\left( t\right)}^{*}{\left(\tilde{p}_{
        ||}1_{[\tilde{c}_{ -},\infty )}{\left(\tilde{p}_{
            ||}\right)}-c^{-}_{+} \tilde{f}{\left(H\right)}^{
        2}\right)}B^{-}{\left(t\right)} +O{\left(t^{-2}\right)}
  \\
  &\geq \epsilon t^{-1}B^{-}{\left(t\right)}^{
    *}B^{-}{\left(t\right)}+O{\left( t^{-2}\right)};\quad \epsilon =
  \tilde{c}_{-}-c^{-}_{ +}.
  \end{split}
\end{equation}

\par Clearly (\ref{esa4}) follows by combining (\ref{x3})
and (\ref{x4}).\par

\par As for (\ref{aaar}) we may proceed similarly using
\begin{equation*}
\Phi {\left(t\right)}=\tilde{f}
{\left(H\right)}\tilde{g}_{
2}{\left(\tilde{p}_{|
|}\right)}F{\left(t^{-1}
r\right)}\tilde{g}_{2}
{\left(\tilde{p}_{||
}\right)}\tilde{f}{\left(
H\right)};\quad F{\left(s^{\prime
}\right)}=\int _{-\infty }^{
s^{\prime}}g^{+}_{1}{\left(
s\right)}^{2}ds.
\end{equation*}
\end{proof}

\begin{corollary}\label{cor:8.2}
  Let $ \psi $, $ \chi _{l}\in C^{\infty }_{0}{\left(
      \mathcal{C}_{\epsilon }\right)} $ and $ \tilde{g}_{2}$ be given
  as in Lemma~\ref{lem:8.1}. Let $ g_{1}\in C^{\infty }_{0}{\left(
      \mathbf{R}\right)}$ be given such that $ 0\leq g_{1}\leq 1$ and
  $ g_{1}=1$ in an open interval containing  $ \supp
{\left(\tilde{g}_{ 2}\right)}$.
 Then 
 \begin{equation}
\label{olee} 
 ||\psi {\left(t\right)}-g_{ 1}{\left(t^{-1}r\right)}
    \tilde{g}_{2}{\left(\tilde{p}_{ ||}\right)}\chi _{l}{\left( 
    \hat{x}\right)}\tilde{f} {\left(H\right)}\psi {\left(t\right)}
||\rightarrow 0\text{ for }t\rightarrow \infty .
\end{equation}
\end{corollary}
\begin{proof}
  From the very definition of $ \psi $ we have 
  \begin{equation*}
||\psi {\left(t\right)}-
\chi _{l}{\left(\hat{x}
\right)}\tilde{f}{\left(
H\right)}\psi {\left(t\right)}
||\rightarrow 0\text{  for }t\rightarrow \infty .
\end{equation*}
Next, from [H, Theorems 4.10 and 4.12] we learn that
\begin{equation}
\label{oleee} 
 ||\psi {\left(t\right)}-\tilde{g}_{2}{\left(\tilde{p}_{
      ||}\right)}\chi _{l}{\left( \hat{x}\right)}\tilde{f}
{\left(H\right)}\psi {\left(t\right)} ||\rightarrow 0 
\text{ for }t\rightarrow \infty .
\end{equation}
Whence to show (\ref{olee}) it suffices to verify that
\begin{equation*}
||{\left\{g_{1}{\left(t^{
-1}r\right)}-g_{1}{\left(\tilde{p}_{||}\right)}
\right\}}\tilde{g}_{2}
{\left(\tilde{p}_{||
}\right)}\tilde{f}{\left(
H\right)}\psi {\left(t\right)}
||\rightarrow 0\text{  for }t\rightarrow \infty ,
\end{equation*}
which in turn is reduced (by a standard density argument using that
the energy bounds the momentum) to verifying that for all 
constants $C$ large enough
\begin{multline}
\label{oleeee} 
 ||F_{-}{\left(C^{-1}t^{-1 }r\right)}{\left\{g_{1}{\left(
        t^{-1}r\right)}-g_{1}{\left( \tilde{p}_{||}\right)}
  \right\}}\tilde{g}_{2} {\left(\tilde{p}_{||
    }\right)}\tilde{f}{\left( H\right)}\psi {\left(t\right)}
||\\\rightarrow 0\;\hbox {\rm for }t\rightarrow \infty .
\end{multline}

\par For (\ref{oleeee}) we consider the observable 
\begin{multline*}
\Phi _{C}{\left(t\right)}\\=
\tilde{f}{\left(H\right)}
\tilde{g}_{2}{\left(\tilde{p}_{||}\right)}
F_{-}{\left(C^{-1}t^{-1}r\right)}
{\left(\tilde{p}_{||
}-t^{-1}r\right)}^{2}F_{
-}{\left(C^{-1}t^{-1}r\right)}
\tilde{g}_{2}{\left(\tilde{p}_{||}\right)}
\tilde{f}{\left(H\right)}.
\end{multline*}

Using Lemma~\ref{lem:8.1} as well as the proof of this lemma we easily
show that 
\begin{equation*}
\left|\int _{1}^{\infty }
|{{d}\over{dt}}{\left\langle \Phi _{
C}{\left(t\right)}\right\rangle }_{
t}\right|dt,\quad \int 
_{1}^{\infty }t^{-1}{\left\langle \Phi _{
C}{\left(t\right)}\right\rangle }_{
t}dt<\infty , 
\end{equation*}
from which we conclude that along some sequence $ t_{k}\rightarrow
\infty $ indeed $ {\left\langle \Phi _{C}{\left(t_{k}
      \right)}\right\rangle }_{t_{k}} \rightarrow 0$, and then in turn
that 
\begin{equation}
\label{OOl} 
 {\left\langle \Phi _{C}{\left(t\right)} \right\rangle
}_{t}\rightarrow 0.
\end{equation}

\par We easily obtain (\ref{oleeee}) using (\ref{OOl}),
(\ref{82a}) and commutation. 
\end{proof}

\par Now, one may easily verify (\ref{olew}) for 
$ \psi =P_{l}f{\left(H\right)}\psi $ as follows: We introduce a
partition $ f=\sum f_{i}$ of sharply localized $ f_{i}${`}s and for
each of these a {``}slightly larger{''} $ \tilde{f}_{i}$. Using these
functions and the states $ \psi _{i}=P_{l}f_{i}{\left( H\right)}\psi $
as input in Corollary~\ref{cor:8.2} the bounds \eqref{eq:8} follow
from the conclusion of the corollary and [H, Theorems 4.10 and 4.12].
As for \eqref{eq:9} we may use the same partition and then conclude
the result from Lemma~\ref{lem:8.1} (applied with $ \tilde{f}$
replaced by $ \tilde{f}_{i}$ ). \par

\begin{remarks}\label{rem:8.3}
\hfill
  \begin{enumerate}
  \item Using the Mourre estimate [ACH, Theorem C.1] one may easily
    include a short-range perturbation $ V_{1}=O{\left(|x|^{ -1-\delta
        }\right)}$, $\delta >0$, $\partial ^{\alpha }_{
      x}V_{1}=O{\left(|x|^{-2} \right)}$, $|\alpha |=2$, to the
    Hamiltonian $H$. In particular Theorem~\ref{thm:1.2} holds for the
    strictly homogeneous case as discussed in Section~\ref{sec:1}.

  \item The non-degeneracy condition at $ \omega_{l}$ is important for
    the method of proof presented in this paper. However it is not
    important that the set of critical points $ C_{r}$ is finite; it
    suffices that $ \omega_{l}$ is an isolated non-degenerate critical
    point and that $ V{\left(C_{r}\right)}$ is countable.\par

  \item At a local maximum we proved a somewhat better result in [HS1]
    (by a different method): A larger class of perturbations was
    included and we imposed a somewhat weaker condition than the
    non-degeneracy condition. The method of [HS1] yielded only a
    limited result at saddle points. Although there are indications
    that this method of proof might be extended to included
    Theorem~\ref{thm:1.2} (by using a certain complicated iteration
    scheme) the proof presented in this paper is probably much
    simpler.\par

  \item The components of the $ \gamma $ of \eqref{eq:19} may be taken
    of the form 
    \begin{equation*}
      \gamma _{j}=\eta _{j}+\sqrt{2{\left(
            E-V{\left(\omega_{l}\right)}\right)}
      } \beta ^{\#}_{j}{\left(
          E\right)}u_{j,}
    \end{equation*}
    where $ \beta ^{\#}_{j}{\left(E\right)} $ is given by one of the
    expressions of (\ref{AX}). In particular both of the conditions
    (\ref{b13}) and (\ref{bb13}) are satisfied in the potential
    case.\par

  \item We applied the Sternberg linearization procedure in [HS3] to
    the equations \eqref{eq:6} in the case of a local minimum. In this
    case the union of all resonances (of all orders and for all
    eigenvalues) is discrete on ${\left(V{\left(\omega_{l}\right)},
        \infty \right)}$. One needs to exclude this set of resonances
    to construct a smooth Sternberg diffeomorphism, see for example
    [N, Theorem 9]. The construction of the symbol $ \gamma
    ^{{\left(m\right)}}$ in (\ref{BB21}) may be viewed as a rudiment
    of this procedure. However, the union of all resonances at a local
    maximum or a saddle point $ \omega_{l}$ is dense in $
    {\left(V{\left(\omega_{l}\right)},\infty \right)}$, and for that
    reason the smooth Sternberg diffeomorphism (defined at
    non-resonance energies) would not be suited for quantization.
    Although not elaborated, one may essentially view $ \gamma
    ^{{\left(m\right)}}$ as being constructed by a $ C^{m}$ Sternberg
    diffeomorphism.\par
\end{enumerate}
\end{remarks}

\appendix
 
\section{A generalization of the homogeneity 
condition}
\label{sec:generalization}

\par In this appendix we shall discuss possible generalizations of the homogeneity
condition \eqref{eq:234}. We 
elaborate on the structure of the classical mechanics of our
models.  A  possible formulation of the quantum problem will be
proposed although not justified in general. It will be discussed for
various examples.\par

\par The homogeneity condition is best understood as the
invariance of the Hamiltonian under the flow generated by the vector
field 
$v{\left(x,\xi\right)}=\sum 
x_{j}\partial /\partial x_{j}$, or infinitesimally 
\begin{equation}
vh{\left(x,\xi\right)}=0.
\label{eq:15}
\end{equation}
Our goal is thus to find invariance conditions \eqref{eq:15} which
will
\begin{enumerate}
\renewcommand{\theenumi}{\alph{enumi}}
\item  reduce the dimension of phase space by two giving an auto\-no\-mous
dynamical system in dimension $2n-2$ (usually not Hamiltonian)

\item give a natural framework for discussing stability of orbits
which do not lie in a compact set. It will turn out that stability is not
measured using any preexisting metric in the phase space but rather using
bundles of orbits of the vector field $v$ surrounding a given orbit of the
Hamiltonian vector field, $ v_{h}$.
\end{enumerate}

\par The particular vector field 
$ v{\left(x,\xi\right)}=\sum x_{j}\partial _{x_{j}}$ does not generate
a symplectic flow but does satisfy a crucial property. Namely $
\mathcal{L}_{v}\omega=\omega$ where $ \mathcal{L}_{v}$ is the Lie derivative
in direction $v$ and $ \omega$ is the symplectic form. It will turn out (see
Lemma~\ref{thm:1.3}) that a geometric condition such as this, although more
restrictive than necessary, will guarantee that $v$ is a suitable vector
field.\par

\par We will require $v$ to satisfy certain conditions
relative to $v_{h}$, where $ v_{h}$ is a Hamiltonian vector field on a
symplectic manifold $ {\left(M,\omega\right)}$ with Hamiltonian $h$:\par
\begin{enumerate}
\item 
In a neighborhood $\mathcal U_0$ of a point 
$ x_{0}\in M$, the local flow $ \phi ^{v}_{t}{\left(\cdot \right)} $
generated by $v$ exists for all $ t\in {\left(-\epsilon ,\infty \right)}$ for
some $ \epsilon >0$ and there exists a surface $ S\subset \mathcal U_0$ containing $ x_{0}$,
transverse to $v$, and a diffeomorphism $ \sigma:B\rightarrow S$,
where $B$ is a ball in $ \mathbf{R}^{2n-1}$ centered at $0$, such that the map
\begin{equation*}
B\times {\left(-\epsilon ,\infty \right)}
\ni {\left(w,t\right)}\rightarrow \phi ^{%
v}_{t}{\left(\sigma {\left(w\right)}
\right)}
\end{equation*}
is a diffeomorphism onto its image, $ \mathcal{K}_{0}\supseteq \mathcal U_0$. We also assume $v$ and $
v_{h}$ are parallel (and nonzero) along the positive orbit of $v$ originating
at $ x_{0}$ (identified as $0\in B$).

\item There are smooth functions $ \beta $ and
$ \gamma $ such that 
\begin{equation*}
{{\left[v,v_{h}\right]}=\beta 
v_{h}+\gamma v\text{ in }\mathcal{K}_{0}.
}\end{equation*} 

\item
$ vh=0\text{ in }\mathcal{K}_{0}$. 
\end{enumerate}

\par Condition (1) allows us to assume (after a change of coordinates)
that $ \mathcal{K}_{0}=B\times {\left(-\epsilon ,\infty \right)}$,
$x_{0}={\left(0,0\right)}$, and $ v={\left(0,\dots ,0,1\right)}$ in $\mathcal{K}_{0}$. With the notation $x_{\bot }={\left(x_{1},\dots
    ,x_{2n-1}\right)}$ for $ x\in \mathbf{R}^{2n}$, condition (2) implies
\begin{equation*}
{\left(v_{h}\right)}_{
\bot }{\left(x\right)}=k{\left(x\right)}
{\left(v_{h}\right)}_{\bot }
{\left(x_{\bot },0\right)} 
\end{equation*}
 where $k(x)=\exp \big (\int^{x_{2n}}_0 \beta\circ
 \phi^v_s (x_\perp,0)ds\big )$  so that introducing the new time
variable $ \tau $ with $ d\tau /dt=k{\left(x{\left(t\right)} \right)}$ the
first $2n-1$ of Hamilton{'}s equations become 
\begin{equation*}
{{{dx_{\bot }}\over{d\tau }}={\left(
v_{h}\right)}_{\bot }{\left(
x_{\bot },0\right)}}.
\end{equation*}
As long as $dh{\left(x_{0}\right)}\not=0$, using condition (3) we can
eliminate one more variable using energy conservation,
$h{\left(x\right)}=h{\left(x_{\bot },0\right)}=E$. For example if $ \partial
h/\partial x_{2n-1}\not=0$ we obtain $ x_{2n-1}=g{\left(w,E\right)}$ with $
w={\left(x_{1},\dots ,x_{2n-2}\right)} $. Here we assume $
{\left(w,E\right)}$ is in a neighborhood of $ {\left(0,E_{0}\right)}$, $
E_{0}=h{\left(x_{0}\right)} =h{\left(0\right)}$. We obtain
\begin{equation}
{{dw}\over{d\tau }}=f{\left(w,E\right)} ,
\label{eq:16}
\end{equation}
where $f(w,E) = ((v_h)_1(w,g(w,E),0),\dots,(v_h)_{2n-2}(w,g(w,E),0))$.  The orbit of
$v_{h}$ along $v$ corresponds to $ w=0,\quad E=E_{0}$ (in which case
$f{\left(0,E_{0}\right)}=0$ ). If $\det {\left(\partial f_{i}/\partial
    w_{j}{\left(0,E_{0}\right)}\right)} \not=0$ there will be a smooth family
of fixed points of \eqref{eq:16}, $ w=w{\left(E\right)}$, in a
neighborhood of $ E_{0}$ (with $ w{\left(E_{0}\right)}=0$). This situation is
 analogous to the case $ v{\left(x,\xi\right)}=\sum x_{j}\partial _{x_{j}}$ discussed in Section \ref{sec:1}  and we can define stability of orbits in $ M$
in terms of the stability of the fixed points $ w{\left(E\right)}$. In
practice one might want to place the fixed point of \eqref{eq:16} at
the origin by an affine 
change of  variables, cf. Section \ref{sec:1}. In any
case one may check that for the model studied in Section \ref{sec:1}
indeed the systems \eqref{eq:6} and \eqref{eq:16} are smoothly
equivalent systems (up to a conformal factor). Notice that in this  case we may choose
$S\subset S^{n-1}\times \mathbf{R}^{n}$, for example. \par

\par If a proof of absence of channels is contemplated along the lines
carried out in this paper, it is necessary that low order
 resonances  do not occur at more than a discrete set of
energies. In particular, the equations \eqref{eq:16} should not
have a Hamiltonian structure. \par 

\par The only place where the Hamiltonian nature of the equations
appeared above was where we used conservation of energy. To bring in the
symplectic form $ \omega $ we introduce a more geometric condition
which turns out to imply condition (2) above (see Remark \ref
{rem:sympl} for an interpretation):
\par 

\begin{lemma}\label{thm:1.3}
 Fix an
 open set $ U\subseteq M$. 
 \begin{enumerate}
   \renewcommand{\theenumi}{\alph{enumi}}
 \item \label{it:clas1}Suppose $ \mathcal{L}_{v}\omega=\alpha \omega$
   in $ U$ for some $\alpha\in C^{\infty}(U)$.  Suppose in addition that $ vh=0$ in $U$. Then $
{\left[v,v_{h}\right]}=-\alpha v_{%
  h}$ in $U$.

\item \label{it:clas2}Suppose $v$ is nonzero in
  $U$ and for any smooth function $h$ on $U$ satisfying $ vh=0$ in a
  neighborhood of a point of $U$, $v$ satisfies $
{\left[v,v_{h}\right]}=-\alpha v_{%
  h}$ in this neighborhood. Then $ \mathcal{L}_{v}\omega=\alpha \omega$ in
$U$.

\end{enumerate}
\end{lemma}

\begin{proof} We shall use the general relations $dh(w)=\omega(v_h,w)$, $[\mathcal{L}_v, i_w] = i_{[v,w]}$ and
$[\mathcal{L}_w, d] = 0$.  Here $i_w$ represents interior product with
 $w$ (see for example \cite[p. 84] {car} or \cite[p. 198]{a2}).

For  (\ref{it:clas1}) we compute in  $ U$
\begin{equation*}
  i_{[v,v_{h}]}\omega=[\mathcal{L}_v, i_{v_{h}}]
  \omega=\mathcal{L}_vdh-i_{v_{h}}\alpha \omega=d\mathcal{L}_vh-i_{\alpha v_{h}}\omega=i_{-\alpha v_{h}}\omega.
\end{equation*} Since $\omega$ is non-degenerate we conclude
  (\ref{it:clas1}).

As for (\ref{it:clas2}) we use the same computation to conclude that 
\begin{equation*}
i_{v_{h}}(-\mathcal{L}_v\omega+\alpha\omega)=0
\end{equation*} in open subsets where $ vh=0$. Since $v$ is nonzero
there are sufficiently many choices of $h$ to conclude from this
that indeed $\mathcal{L}_v\omega=\alpha\omega$.
\end{proof} 
%As for (\ref{it:clas3}) we obtain by integrating $
%\mathcal{L}_{v}\omega=\alpha \omega$
%\begin{equation*}
%  (\phi^v_t)^*\omega=\exp\Big ({\int^t_0\alpha\circ \phi^v_{s}ds}\Big )\omega,
%\end{equation*} from which it follows that
%if $w_1$
%and $w_2$ 
%are in involution, then also $(\phi^v_t)_*w_1$ and $(\phi^v_t)_*w_2$  are in i%nvolution. 
%
%Conversely, we may consider the forms $\omega$ and
%$\mathcal{L}_v\omega$ as linear functionals $l_1$ and $l_2$,
%respectively, on the tangent space at
%any point $ x\in U$ by ``freezing'' either the first or the second
%entry. In either case the assumption of invariance of pairs in
%involution implies that $\ker l_1\subseteq \ker l_2$, and consequently
%that $l_2=\alpha l_1$. It is readily seen that $\alpha$ only depends on
%$x$ (and not on tangent vectors), and that the dependence is smooth.

\begin {remark}\label{rem:sympl} By integrating the
  condition of  Lemma \ref{thm:1.3}  (\ref {it:clas1}), $
\mathcal{L}_{v}\omega=\alpha \omega$,   we obtain 
\begin{equation}\label{ekssym222}
  (\phi^v_t)^*\omega=\exp\Big ({\int^t_0\alpha\circ \phi^v_{s}ds}\Big )\omega.
\end{equation} In particular  if $
\mathcal{L}_{v}\omega=\alpha \omega$ holds in $M$ and $\phi^v_t$ is a
global flow we see that the diffeomorphisms $\phi^v_t$ preserve the family of Lagrangian
manifolds. 

Conversely one may readily prove that if $\phi^v_t$ is a  global flow
and the diffeomorphisms $\phi^v_t$ preserve the family of Lagrangian
manifolds, then indeed  $
\mathcal{L}_{v}\omega=\alpha \omega$ for some smooth $\alpha$.
\end {remark}
 
We give two  simple examples.  

\begin{example} \label{exam:riema2} Consider the symbol $h$ on
  $\mathbf{R}^{2}\times \mathbf{R}^{2}$, suitably regularized at singularities,
  \begin{equation*}\label{ekssym}
 h=h(x,\xi)=\tfrac {1}{2}\big (x^2-a\xi _2^2\big )^{-1}\xi ^2;\; a >0.   
  \end{equation*}  Let $v(x,\xi)=\tfrac {1}{2} \sum (x_{j}\partial _{x_{j}}+\xi_{j}\partial _{\xi_{j}})$. Then the vector field  $v$
  and the Hamitonian vector field 
  $v_h$ fulfill the conditions (1)--(3) along the positive
  orbit of $v$ originating at  $(1+2E)^{-1/2}(1,0;\sqrt{2E},0)$,    
  $E>0$. Here we  take the $S$ in condition (1) to be a subset of the
  unit-sphere ${S}^3$. Notice also that $(\phi^v_t)^* \omega =
  \exp(t)\omega$, and therefore  $\mathcal{L}_v\omega=\omega$. 
 After linearizing the   reduced
  flow  
  \eqref{eq:16} we  find the   eigenvalues 
  \begin{equation*}
  -\sqrt{2E}\Big(1\pm \sqrt{1+4Ea}\Big ),  
  \end{equation*}
  and we conclude that the family of
fixed points consists of saddle
  points. Resonances
  (of any fixed  order) are discrete in $(0,\infty)$.
\end{example}

\begin{example} \label{exam:riema3} Consider the symbol $h$ on $(\mathbf{R}^{2}\setminus{\left\{
0\right\}})\times \mathbf{R}^{2}$
  \begin{equation*}
 h=h(x,\xi)=\tfrac
 {1}{2}\big(x_1^2+bx_2^2\big)^{\kappa/2}\xi^2;\;b>0,\,\kappa <2,\,\kappa (b-1)<0.   
  \end{equation*} We introduce $s=2/(2-
 \kappa)$ and $v=\sum (sx_j\partial _{x_j}+(1-s)\xi _j\partial _{\xi
 _j})$. The vector field  $v$
  and the Hamitonian vector field 
  $v_h$ fulfill the conditions (1)--(3) along the positive
  orbit of $v$ originating at  $(1,0;\sqrt{2E},0)$,    
  $E>0$. Here we  take  $S\subset \{(x,\xi)|\,x_1=1\}$. We notice that
 the condition $\kappa<2$ assures that the $x$--component of the
 flow $\phi^v_t$ grows as $t\to \infty$; whence there is no conflict
 with a regularization  at  $x=0$. (The fact that for $\kappa \in (0,2)$ the 
 $\xi $--component  decays is irrelevant.) 
 We  find the   eigenvalues  for the linearized reduced flow to be given by 
 \begin{equation*}
 -\tfrac
 {2-\kappa}{4}\sqrt {2E}\Big \{1\pm \sqrt{1-8\kappa(b-1)(2-\kappa)^{-2}}\Big \}.  
 \end{equation*} 
  Since by assumption $\kappa (b-1)<0$ we conclude that the family of
fixed points consists of saddle
  points. For a ``generic''  set of parameters $b$ and $\kappa$ there
  are no resonances (of any order).   
\end{example}

We shall propose a formulation of  the quantum problem
corresponding to  the
 classical framework discussed above, and then relate it to 
Examples \ref{exam:riema2} and \ref{exam:riema3}. 

Let us strengthen  the above
conditions (1)-(3) as follows: We assume that $\epsilon= \infty$ in
$(1)$ so that $\mathcal K_0$ is two-sided invariant under the flow
$\phi^v_\tau$, and furthermore that the condition $
\mathcal{L}_{v}\omega=\alpha \omega$ of Lemma \ref{thm:1.3} (\ref{it:clas1})
   holds in $ U=\mathcal K_0$ (implying (2) with $\beta=-\alpha$ and
   $\gamma=0$). Suppose also that $\alpha>0$.

Under these conditions we may write 
\begin{align*}
  &\phi^v_{\tau(t,E_0)}(x_0)=\phi^{v_h}_t(x_0);\\&{d\tau(t,E_0)\over dt}=\exp \Big (-\int^{\tau(t,E_0)}_0 \alpha\circ
 \phi^v_s (x_0)ds\Big )k(E_0),\\
&v_h(x_0)=k(E_0)v(x_0),\;\tau(0,E_0)=0.
\end{align*} Notice that any maximal solution to this differential equation
 is defined at least on a positive directed half-line
 (i.e. $\tau(t,E_0)$ exists for all large $t$'s). Denoting by $x(E)\in S$ the fixed points for neighboring
energies $E\approx E_0$ we have similar identities for the positive common
orbits originating
at $ x_0\to x(E)$. Whence we may look at localization of states in quantum
mechanics in terms of Weyl quantization of symbols of the form
$a(\phi^v_{-\tau(t,h)})$ where $a\in C^{\infty
  }_{0}{\left(\mathcal{U}_{0 }
    \right)}$. Notice that for the model studied in the bulk of this paper this procedure
  is a slight modification of the one used in 
  \eqref{eq:9} and \eqref{eq:10}. In fact in this case we may take
  $S\subset S^{n-1}\times \mathbf{R}^{n}$ and compute in terms of the
  function $k=k(E)$ of \eqref{eq:4}
  \begin{equation*}
    \tau=\ln (tk(E)+1)
  \end{equation*} yielding 
  \begin{equation*}
    \phi^v_{-\tau(t,h)}(x,\xi)=(x/(tk(h)+1),\xi);\;h=h(x,\xi).
  \end{equation*} We need in this
  setting to replace 
\begin{equation*}
\gamma(I_0)\to   \gamma(I_0)=\{x(E)=(\omega(E),\xi(E))|\;E\in
  I_0\}.   
  \end{equation*}
  There is also a way to interpret the first factor $t^{-1}$ of
  \eqref{eq:9}:
Using \eqref{ekssym222} we may compute the Poisson bracket
\begin{equation*}
  \{h,a(\phi^v_{-\tau(t,h)}(\cdot))\}=\exp \Big (\int^{-\tau(t,h)}_0  \alpha\circ
 \phi^v_s (\cdot)ds\Big )\{h,a\}(\phi^v_{-\tau(t,h)}(\cdot)),
\end{equation*} which indicates that the first factor to the right is a
 ``Planck constant'' (this interpretation is supported by
 the requirement $\alpha>0$). Effectively it is equal to $t^{-1}$ for this 
 example. Whence a possible  reformulation of the integral
 condition \eqref{eq:9} (suited for generalization) is 
  \begin{align}\label{eq:999}
\int _{1}^{\infty }&\left\|b^{w}_t{\left(x,p\right)}\psi {\left(t\right)}
  \right\|^{2}dt<\infty \text{ for all }a\in C^{\infty
  }_{0}{\left(\mathcal{U}_{0 }\setminus\gamma {\left(I_{0}\right)}
    \right)};
\\&a_t(x,\xi)=a(\phi^v_{-\tau(t,h)}(x,\xi)),\nonumber\\
&b_t(x,\xi)=\exp \Big (2^{-1}\int^{-\tau(t,h)}_0  \alpha\circ
 \phi^v_s (x,\xi)ds\Big )a_t(x,\xi),\nonumber\\&\gamma {\left(I_{0}\right)}
  ={\left\{x(E)\mid E\in I_{%
0}\right\}},\,\psi(t)
  ={e^{-itH}f\left(H\right)}\psi,\, f\in C^{\infty }_{0}{\left(I_{0}\right)}.\nonumber
\end{align}
The analogous statement of Theorem \ref{thm:1.11}
in general would read: 

For all $a\in C^{\infty
  }_{0}{\left(\mathcal{U}_{0 }
    \right)}$ and all localized states $\psi(t)
  ={e^{-itH}f\left(H\right)}\psi $, $f\in C^{\infty }_{0}{\left(I_{0}\right)}$, obeying \eqref{eq:999} with
  $I_0\ni E_0$ small enough
\begin{equation}\label{eq:999990}
\big\|a_t^{w}(x,\xi)\psi {\left(t\right)}
\big\|\rightarrow 0\text{ for }t\rightarrow \infty.
\end{equation} 

Now, for Examples \ref{exam:riema2} and \ref{exam:riema3}  we may compute 
\begin{equation}\label{eq:9999}
  \phi^v_{-\tau(t,h)}(x,\xi)=\big({t_0/(t+t_0)}\big)^{1/2}(x,\xi);\;t_0=\big(2\sqrt{2h}(1+2h)\big)^{-1},
\end{equation}
and  
\begin{equation}\label{eq:99991}
  \phi^v_{-\tau(t,h)}(x,\xi)=\Big(\big(\tfrac {t}{s\sqrt {2h}}+1\big)^{-s}x,\big(\tfrac {t}{s\sqrt {2h}}+1\big)^{s-1}\xi
  \Big);\;s=2/(2-\kappa),
\end{equation} respectively. 

We may use the effective Planck
constant $t^{-1}$ like for the other example. In conclusion, the somewhat complicated looking
quantum
condition \eqref{eq:999} reduces to  simple  explicit
requirements. Similarly \eqref{eq:999990} reads in these  cases
\begin{equation}\label{eq:99999}
\big\|a^{w}{\left((t_0(h)/t)^{1/2}(x,p)\right)}\psi {\left(t\right)}
\big\|\rightarrow 0\text{ for }t\rightarrow \infty
\end{equation} and 
\begin{equation}\label{eq:999991}
\big \|a^{w}{\Big (\Big (s\sqrt {2h}\Big )^{s}t^{-s}x,\Big (s\sqrt {2h}\Big )^{1-s}t^{s-1}p\Big )}\psi {\left(t\right)}
\big\|\rightarrow 0\text{ for }t\rightarrow \infty,
\end{equation} respectively.

We remark that  \eqref{eq:999}, \eqref{eq:9999} (or \eqref{eq:99991}) and
\eqref{eq:99999} (or \eqref{eq:999991}) apply
literally for Example \ref{exam:riema2} (or  Example  \ref{exam:riema3}); the conclusion
\eqref{eq:99999} (or \eqref{eq:999991}) for the states considered may be reached using Theorem
\ref{thm:1.11} after a symplectic change of variables and invoking symplectic covariance:

\begin{example} \label{exam:riema4} Consider a  smooth symbol $h$ on $(\mathbf{R}^{n}\setminus{\left\{
0\right\}})^2$ obeying one of the  homogeneity properties  1) 
\begin{equation*}
 h(\lambda x,\lambda  \xi)=h(x,\xi); \text { for all }\lambda >0,
  \end{equation*} 
or 2)  for 
some $ \kappa _2\neq 0$ and some $ \kappa _1 \neq \kappa _2$ 
  \begin{equation*}
 h(\lambda _1x,\lambda _2 \xi)=\lambda _1^{\kappa _1}\lambda
 _2^{\kappa _2}h(x,\xi); \text { for all }\lambda _1,\lambda _2 >0.
  \end{equation*} 

For 2) the change of variables $x=|y|^{s}\hat y= |y|^{s-1}y$, where
 $s=\kappa_2/(\kappa_2- \kappa_1)$,  induces a symplectic map on $(\mathbf{R}^{n}\setminus{\left\{
0\right\}})^2$. The Hamiltonian  in the corresponding new
 variables, denoted again by $x$ and $\xi$, reads 
 \begin{equation*}
  \tilde h(x,\xi )=h(\hat x , \xi
 +(s^{-1}-1)\langle \hat x ,\xi \rangle \hat x). 
 \end{equation*}
The same  change of variables with $s=\tfrac {1}{2}$ leads for 1) to
a Hamiltonian of the same   form.
In particular \eqref{eq:234} holds (in both cases) for the new symbol $ \tilde h$. Up to
other conditions we may therefore apply Theorem
\ref{thm:1.11}. Clearly Examples \ref{exam:riema2} and 
\ref{exam:riema3}  are  concrete
examples.  To
stress the symplectic covariance let us note that indeed $v:=\sum (sx_j\partial _{x_j}+(1-s)\xi _j\partial _{\xi _j})\to \tilde v
:=\sum x_j\partial _{x_j}$. 
\end{example}

We give yet another example from Riemannian
  geometry.

\begin{example}
\label{exam:riema33} Consider the symbol $h$ on $(\mathbf{R}^{2}\setminus{\left\{
0\right\}})\times \mathbf{R}^{2}$
  \begin{equation*}
 h=h(x,\xi)=\tfrac
 {1}{2}g^{-1} \xi^2,   
  \end{equation*} where the conformal (inverse) metric factor is
 specified in polar coordinates 
 $x=(r\cos \theta, r\sin \theta)$ as
 $g^{-1}=e^f;\;f=f(\theta-c\ln r)$. We assume $f$ is a given smooth 
 non-constant $2\pi$--periodic function and that $c>0$. 
We introduce $v=(x_1-cx_2)\partial_{x_1}+(cx_1+x_2)\partial_{x_2}-c\xi
_2\partial_{\xi_1}+c\xi_1 \partial_{\xi_2}$. Computations show that  $v$
  and the Hamitonian vector field 
  $v_h$ fulfill the conditions (1)--(3) along the positive
  orbit of $v$ originating at  $(r_0,0;\rho_0,c\rho_0)$; here $\rho_0=\sqrt
  {2E(1+c^2)^{-1}e^{-f_0}}$ where    
  $f_0 =f(\theta_0)$ is given in terms
  of any $r_0>0$ satisfying  the
  equation
  \begin{equation}
    \label{eq:237}
 -f'(\theta_0) =2c(1+c^2)^{-1};\;\theta_0=-c\ln r_0, 
  \end{equation}  and $E=h>0$ is arbitrary.  (Notice that there 
  are at least two solutions to \eqref{eq:237} for  
  all small as well as for  all large values of $c$.) The $x$--space part of the orbit (a geodesic) is the logarithmic spiral
 given by the equation $\theta-c\ln r= \theta_0$. 
 We  take  $S\subset \{(x,\xi)|\,x_2=0\}$ and compute the eigenvalues
  for the linearized reduced flow to be given by
\begin{equation} \label{eq:238}
 -\rho_0 \tfrac
 {1}{2} \Big \{1\pm \sqrt{1-2(1+c^2)^2f''_0}\Big \};\;f''_0=f''(\theta_0).
 \end{equation}  For $f''_0<0$ the
family of fixed points consists of  saddles. There are no resonances for ``generic''
values of $c$, and we also notice that  taking  $c\to 0$ in
\eqref{eq:237} and \eqref{eq:238} yields the formulas for the
corresponding homogeneous model (here the equations are  considered to be
 equations  in $c$ and $\theta_0$).

Finally, using the new angle $\tilde \theta = \theta-c\ln r$ one may again
conjugate to a homogeneous model. 
More precisely the relevant  symplectic change of variables is induced
(expressed here in terms of  rectangular coordinates) by the map $x\to \tilde x=(x_1g_1+x_2g_2,x_2g_1-x_1g_2)$, where
$g_1=\cos (c\ln |x|)$ and $g_2=\sin (c\ln |x|)$. One may check that  
$v\to \tilde v
:=\sum x_j\partial _{x_j}$, and that $h\to \tilde h$ given by
\begin{equation*}
\tilde h =\tfrac
 {1}{2}e^{f(\theta)}\Big ( \big \{(c\sin \theta +\cos
 \theta)\xi_1+(\sin \theta-c\cos \theta)\xi _2\big  \}^2+\{-\sin
 \theta  \xi _1 +\cos \theta  \xi _2\}^2 \Big );
\end{equation*}
 we changed notation back to the old one, $x=(r\cos \theta, r\sin
 \theta)$ for position and $\xi$ for momentum.
\end{example}

\begin{remark} \label{rem:slut} Although we shall not elaborate,  due to the general
nature of   the  method used 
in the bulk of this paper the  method   should be generalizable to apply to  the quantum
problem for Examples
\ref{exam:riema2}, \ref{exam:riema3} and \ref{exam:riema33} (without
changing variables). We believe it would apply to the quantum
problem for a variety of other  examples of the classical
theory. However we have not pursued the outlined general scheme for two reasons: 1)
There are additional complications related to the pseudodifferential
calculus, cf. \cite [Section 18] {ho}. The
treatment of these
complications is somewhat cumbersome and does not  add new insight to
the problem. 2) The condition \eqref{eq:999} has a certain global
flavour in our opinion, whence it does not entirely stand alone.  For instance its verification in the context of proving
asymptotic completeness, cf. \cite{h}, \cite {hs3},  \cite{chs2} 
and Section \ref{sec:8}, 
relies on {\it global}  information on the dynamics. 
 
To illustrate this point  further let us look at Example \ref{exam:riema3} in the case
$\kappa <0$ and $b>1$. For the classical problem any orbit $x(t)$
going to infinity will roughly follow either the $x_1$--axis or the
$x_2$--axis. As a first step of proving asymptotic
completeness in Quantum Mechanics (for the regularized Hamiltonian)
one may derive estimates for states in
the continuous subspace with roughly
the same content, in particular the bound \eqref{eq:999}.  Due to  the
eigenvalue calculation of Example
\ref{exam:riema3} only the
$x_2$--axis is ``stable'' for the classical orbits.  The corresponding statement in Quantum Mechanics
given by \eqref{eq:999991}  then leads to the preliminary information
for  asymptotic
completeness, $\|x_1/|x| \psi(t)\|\to 0$ for $t\to \infty$. 
Although the
dynamics of   Example
\ref{exam:riema33}  in general
is more complicated than Example
\ref{exam:riema3} we remark that the attractive spirals (cf. the
eigenvalue calculation \eqref{eq:238}) similarly define non-trivial quantum channels. One can show in some cases,
for example if $f'(\theta) + 2c(1+c^2)^{-1}\leq 0$ on an interval of
length $(1+c^2)\pi /2$, that those channels are the only occuring ones.
\end{remark}

\end{document}